\documentclass{article}
\usepackage[utf8]{inputenc}
\usepackage{fullpage}
\usepackage{tikz}
\usepackage{amsthm,amsmath,amssymb,amsfonts}

\usepackage{authblk}

\usepackage{thmtools}
\usepackage{thm-restate}

\usepackage[hidelinks,draft=false, colorlinks=true]{hyperref}
\usepackage[capitalize]{cleveref}
\usepackage{bbold} % for the identity symbol

\usepackage[
    style=ieee-alphabetic,
    sorting=nyt,
    backend=biber,
    url=false,
    isbn=false,
    hyperref=true,
    eprint=true,
    maxbibnames=7,
]{biblatex}
\AtEveryBibitem{%
  \clearfield{eprintclass}%
}

\newbibmacro{string+doi}[1]{%
  \iffieldundef{doi}{#1}{\href{http://dx.doi.org/\thefield{doi}}{#1}}}
\DeclareFieldFormat{title}{\usebibmacro{string+doi}{\mkbibemph{#1}}}
\DeclareFieldFormat[article]{title}{\usebibmacro{string+doi}{\mkbibquote{#1}}}

\hypersetup{
    pdftitle={Weak Hopf algebra injectivity},
    pdfauthor={AM, JGR, DPG, NS, JIC},
    bookmarks=true,
    bookmarksnumbered=true,
    bookmarksopen=true,
    bookmarksopenlevel=1,
    colorlinks,
    pdfstartview=Fit,
    pdfpagemode=UseOutlines,    % this is the option you were lookin for
    pdfpagelayout=TwoPageRight
}

\newtheorem{lemma}{Lemma}[section]
\newtheorem{proposition}{Proposition}[section]
\newtheorem{theorem}{Theorem}[section]

\newtheorem{definition}{Definition}[section]

\numberwithin{equation}{section}

\newcommand{\tr}{\operatorname{Tr}}
\newcommand{\id}{\mathrm{Id}}
\newcommand{\End}{\mathrm{End}}
\newcommand{\ket}[1]{\vert #1 \rangle}
\newcommand{\bra}[1]{\langle #1 \vert}
\newcommand{\scalprod}[2]{\langle #1 \vert #2 \rangle}
\newcommand{\irr}{\mathrm{Irr}}

\usetikzlibrary{decorations.pathreplacing,decorations.markings,calc,3d,positioning,quotes}

\tikzset{
  tensor/.style={
    inner sep = 0.05cm,
    shape = circle,
    draw,
    fill
  },
  t/.style={
    inner sep = 0.03cm,
    shape = circle,
    draw,
    fill
  },
}

\tikzset{ z={(30:1)}}

\tikzset{>=stealth}

\tikzset{
  every picture/.style = {
    baseline={([yshift=-.5ex]current bounding box.center)},
    font=\scriptsize
  },
  irrep/.style={
    anchor=south,
    font = \tiny,
    inner sep=2pt,
    black,
    midway
  },
  irrep south/.style={
    anchor=north,
    font = \tiny,
    inner sep=3pt,
    black,
    midway
  }
}

\tikzset{
  fusion tensor/.style = {
    rectangle,
    minimum width= 1mm,
    minimum height = 7mm,
    rounded corners = 0.5mm,
    inner sep = 0,
    draw,
    fill,
  }
}

\tikzset{
  every label/.style = {
    text depth=0pt,
    text height=1ex,
  },
}

\tikzset{
  stringnet mpo/.pic={
        \def\d{0.15}; % lines distance
        \def\l{0.4}; % sticking out line length
        \def\h{0.4}; % rectangle half-size
        \draw (-\h,-\h) rectangle (\h,\h);
        \draw (-\l-\h,0) -- (-\h,0) coordinate[midway] (-left-mid);
        \draw[densely dotted] (-\h,0) -- (\h,0);
        \draw (\h,0) --(\h+\l,0) coordinate[midway] (-right-mid);
        \draw (0,-\l-\h) -- (0,-\h) coordinate[midway] (-down-mid);
        \draw[densely dotted] (0,-\h) -- (0,\h);
        \draw (0,\h) --(0,\h+\l) coordinate[midway] (-up-mid);
          \foreach \i/\label/\ca/\cb in {%
                    0/b/up/right,%
                    90/f/left/up,%
                    180/e/down/left,%
                    270/k/right/down}{
            \begin{scope}[rotate = \i]
              \draw (\d,\h) -- (\d,\h+\l) coordinate[midway] (-\ca-\cb);
              \draw[densely dotted] (\d,\h) arc (180:270:\h-\d);
              \draw (\h,\d) -- (\h+\l,\d) coordinate[midway] (-\cb-\ca);
            \end{scope}
          }
  },
  stringnet fusion/.pic={
    \def\d{0.15};
    \def\h{1.0};
    \def\w{0.25};
    \def\y{0.6};
    \def\l{0.4};
    \draw (-\w,-\h) rectangle (\w,\h);
    \draw (\w,0)--(\w+\l,0) coordinate[midway] (-right-mid);
    \draw (\w,\d)--(\w+\l,\d) coordinate[midway] (-right-up);
    \draw (\w,-\d)--(\w+\l,-\d) coordinate[midway] (-right-down);
    \draw (-\w,\y)--(-\w-\l,\y) coordinate[midway] (-left-up-mid);
    \draw (-\w,\y-\d)--(-\w-\l,\y-\d) coordinate[midway] (-left-up-down);
    \draw (-\w,\y+\d)--(-\w-\l,\y+\d) coordinate[midway] (-left-up-up);
    \draw (-\w,-\y)--(-\w-\l,-\y) coordinate[midway] (-left-down-mid);
    \draw (-\w,-\y-\d)--(-\w-\l,-\y-\d) coordinate[midway] (-left-down-down);
    \draw (-\w,-\y+\d)--(-\w-\l,-\y+\d) coordinate[midway] (-left-down-up);
    \def\r{0.2}
    \draw[densely dotted] (-\w,\y-\d) arc (90:0:\r) -- (-\w+\r,-\y+\d+\r) arc (0:-90:\r);
    \draw[densely dotted] (\w,\d) arc (-90:-180:\r) -- (\w-\r,\y+\d-\r) arc (0:90:\r) --(-\w,\y+\d);
    \draw[densely dotted] (\w,-\d) arc (90:180:\r) -- (\w-\r,-\y-\d+\r) arc (0:-90:\r) --(-\w,-\y-\d);
  },
  stringnet unit/.pic={
    \def\d{0.15};
    \def\l{0.4};
    \def\w{0.1}
    \draw (-\w,-0.25) rectangle (\w,0.25);
    \draw (\w,0) -- (\w+\l,0) coordinate[midway] (-mid);
    \draw (\w,\d) -- (\w+\l,\d) coordinate[midway] (-up);
    \draw (\w,-\d) -- (\w+\l,-\d) coordinate[midway] (-down);
    \draw[densely dotted] (\w,-\d) arc (-90:-180:\w) --(0,\d-\w) arc (180:90:\w);
  },
  stringnet projector/.pic={
    \def\d{0.15};
    \def\l{0.4};
    \def\w{0.15}
    \draw (-\w,-0.25) rectangle (\w,0.25);
    \draw (\w,0) -- (\w+\l,0) coordinate[midway] (-right-mid);
    \draw (\w,\d) -- (\w+\l,\d) coordinate[midway] (-right-up);
    \draw (\w,-\d) -- (\w+\l,-\d) coordinate[midway] (-right-down);
    \draw (-\w,0) -- (-\w-\l,0) coordinate[midway] (-left-mid);
    \draw (-\w,\d) -- (-\w-\l,\d) coordinate[midway] (-left-up);
    \draw (-\w,-\d) -- (-\w-\l,-\d) coordinate[midway] (-left-down);
  },
  ->-/.style={postaction={decorate,decoration={
        markings,
        mark=at position #1 with {\arrow{stealth}}
      }}},
  ->-/.default=0.6,
  -<-/.style={postaction={decorate,decoration={
        markings,
        mark=at position #1 with {\arrow{stealth reversed}}
      }}},
  -<-/.default=0.6,
  virtual/.style = {red},
}

\tikzset{
  on each segment/.style={
    decorate,
    decoration={
      show path construction,
      moveto code={},
      lineto code={
        \path [#1]
        (\tikzinputsegmentfirst) -- (\tikzinputsegmentlast);
      },
      curveto code={
        \path [#1] (\tikzinputsegmentfirst)
        .. controls
        (\tikzinputsegmentsupporta) and (\tikzinputsegmentsupportb)
        ..
        (\tikzinputsegmentlast);
      },
      closepath code={
        \path [#1]
        (\tikzinputsegmentfirst) -- (\tikzinputsegmentlast);
      },
    },
  },
  mid arrow/.style={postaction={decorate,decoration={
        markings,
        mark=at position .5 with {\arrow[#1]{stealth}}
      }}},
}

\usepackage{graphicx}

\newcommand{\tinyprime}[1]{#1^{\scalebox{.6}{$\scriptscriptstyle\prime$}}}
\newcommand{\tinypprime}[1]{#1^{\scalebox{.6}{$\scriptscriptstyle\prime\prime$}}}

\newcommand{\rem}[1]{\textcolor{red}{\textbf{[#1]}}}

\usetikzlibrary{external}
\tikzexternaldisable
\title{Matrix product operator algebras I: representations of weak Hopf algebras and projected entangled pair states}

\author[1]{Andras Molnar}
\author[2,3]{Alberto Ruiz de Alarcón}
\author[1]{Jos\'e Garre-Rubio}
\author[1,4]{Norbert Schuch}
\author[5,6]{J.~Ignacio Cirac}
\author[2,3]{David P\'erez-Garc\'ia}
\affil[1]{University of Vienna, Faculty of Mathematics, Oskar-Morgenstern-Platz 1, 1090 Vienna, Austria}
\affil[2]{Dpto. de An\'alisis Matem\'atico y Matem\'atica Aplicada, Universidad Complutense de Madrid, 28040 Madrid, Spain}
\affil[3]{Instituto de Ciencias Matem\'aticas, C/ Nicol\'as Cabrera 13-15, 28049 Madrid, Spain}
\affil[4]{University of Vienna, Faculty of Physics, Boltzmanngasse 5, 1090 Vienna, Austria}
\affil[5]{Max-Planck-Institute of Quantum Optics, Hans-Kopfermann-Str. 1, 85748 Garching, Germany}
\affil[6]{Munich Center for Quantum Science and Technology, Schellingstr. 4, 80799 München, Germany}
\renewcommand\Affilfont{\itshape\small}

\begin{document}

\maketitle

\begin{abstract}
Matrix Product Operators (MPOs) are tensor networks representing operators acting on 1D systems.  They model a wide variety of situations, including communication channels with memory effects, quantum cellular automata, mixed states in 1D quantum systems, or holographic boundary models associated to 2D quantum systems.  A scenario where MPOs have proven particularly useful is to represent algebras of non-trivial symmetries.  Concretely, the boundary of both symmetry protected and topologically ordered phases in 2D quantum systems exhibit symmetries in the form of MPOs.

In this paper, we develop a theory of MPOs as representations of algebraic structures.  We establish a dictionary between algebra and MPO properties which allows to transfer results between both setups, covering the cases of pre-bialgebras, weak bialgebras, and weak Hopf algebras. We define the notion of pulling-through algebras, which abstracts the minimal requirements needed to define topologically ordered 2D tensor networks from MPO algebras.  We show, as one of our main results, that any semisimple pivotal weak Hopf algebra is a pulling-trough algebra.  We demonstrate the power of this framework by showing that they can be used to construct Kitaev's quantum double models for Hopf algebras solely from an MPO representation of the Hopf algebra, in the exact same way as MPO symmetries obtained from fusion categories can be used to construct Levin-Wen string-net models, and to explain all their topological features; it thus allows to describe both Kitaev and string-net models on the same formal footing.
\end{abstract}

\tableofcontents

\section{Introduction}

\subsection{Context}

Tensor Networks \cite{Okunishi2022,Cirac2021,Hackbusch2019o} allow to express certain high-dimensional
tensors $T_{i_1,\dots,i_N}$ -- that is, tensors with a large number $N$ of
indices -- efficiently as a network of simple few-index
tensors. Those elementary tensors are arranged on a graph (``network''),
and their ``auxiliary'' indices are contracted (i.e., identified and
summed over) with those of the adjacent tensor, as prescribed by the edges
of the graph. The resulting object describes a multi-dimensional tensor
whose indices are given by the remaining uncontracted indices of the
original network. This allows to express the original tensor
$T_{i_1,\dots,i_N}$ with its exponential number $d^N$ of parameters
(for simplicity, $i_k=1,\dots,d$ for all $k$) through a much smaller number of
parameters (typically scaling linearly in $N$), while retaining highly
non-trivial correlations between the different indices of $T$.
Surprisingly, in a remarkably broad range of applications, this
exponential saving in the number of parameters effectively does not
restrict the expressive power of the ansatz.

A particularly successful instance of tensor networks is given by
one-dimensional (1D) tensor networks, most commonly known by the names of
Matrix Product States (MPS) \cite{Perez-Garcia2007,Fannes1992} or Tensor Trains (TT) \cite{Oseledets2011}, among others. These
ansatzes have been independently (re-)discovered in different fields and have found
applications in a multitude of areas throughout data science
and physics. This includes, for instance, their use in modeling
and compressing high-dimensional data, such as for image compression or
the simulation of PDEs, as well as in the design of deep learning
algorithms, such as for hidden Markov models or Born machines \cite{Stoudenmire2017,Novikov2017,Cichocki2016,Cichocki2017,
Latorre2005,McCord2022,Iblisdir2007,Garcia-Ripoll2021,Bachmayr2016,Carleo2019,Glasser2019}.
In physics, Tensor Networks -- both one-dimensional MPS and
their higher-dimensional generalization, termed Projected Entangled Pair
States (PEPS) -- have proven to be a powerful ansatz for the simulation
and modeling of quantum many-body systems in a wide range of scenarios in
condensed matter, atomic, and high energy physics, as well as in quantum
chemistry \cite{Schollwoeck2011, Sharma2012, McCulloch2002, Dalmonte2016, Paeckel2019, Zheng2017,  Orus2019,Verstraete2008,Schon2005,Banuls2020,Pastawski2015,Jahn2021,Szalay2015,Olivares-Amaya2015}.
Notably, their suitability for faithfully approximating low-energy
states of local Hamiltonians can be rigorously proven in a broad range of
settings \cite{Hastings2007b,Hastings2007a,Cirac2021}, providing a formal justification of their success in the study
of such systems; and the fact that MPS allow to represent physical symmetries
locally \cite{Perez-Garcia2008a} -- that is, as acting on the auxiliary degrees of freedom -- has
enabled a comprehensive classification of unconventional phases under symmetries \cite{Chen2011,Schuch2011,Cirac2021}.

A generalization of MPS which is natural in many contexts are Matrix
Product Operators (MPOs) \cite{Zwolak2004}. MPOs describe
operators acting on a multi-dimensional tensor which themselves can be
expressed as 1D tensor networks and thus possess -- and preserve -- the
underlying 1D locality structure. On the one hand, such MPOs can encode
operations or transformations which themselves have a 1D structure; they
thus form natural generalization of strictly local operations, in that
they preserve the local correlation structure when acting on an MPS \cite{Murg,Haegeman2017,Cirac2017a,Sahinoglu2018}.
On the other hand, MPOs can characterize symmetries acting on a system
which have an intrinsic 1D correlation structure \cite{Roose2019,Couvreur2022}, or by implication
projections onto subspaces invariant under the symmetry which can exhibit
highly non-trivial structures \cite{Garre-Rubio2022}. Notably, in both cases, despite their local
form the structure of such operators can be fundamentally different from
what can be achieved by low-depth local circuits alone.

MPOs appear most naturally in the context of quantum many-body problems.
They can, for instance, describe evolutions with a local structure (but
not necessarily locally generated) such as in driven quantum systems,
i.e., Floquet physics \cite{Po2016}; and they can describe the symmetries of chains of
particles with anyonic statistics \cite{Feiguin2007,Gils2013,Garre-Rubio2022}.  However, the use of MPOs in many-body
physics goes far beyond this and extends deeply into the study of 2D systems: Boundaries
-- and in particular the correlations across those boundaries, i.e., the
entangement spectrum --
play a key role in
understanding the physics of strongly correlated quantum matter \cite{Li2008,Laflorencie2016}.
 The boundary of a 2D quantum system has, yet again, a natural 1D structure,
which can be made explicit by cutting a 2D tensor network description of
the bulk system at the boundary \cite{Cirac2011}. The density matrix which carries the
entanglement spectrum can then be described by an MPO at the boundary.
On the other hand, just as physical symmetries in MPS can be represented
as acting at the auxiliary indices, symmetries in 2D PEPS can be
understood as acting as MPOs at the 1D boundary \cite{Chen2011a,Williamson2014,Molnar2017}, in which case they
naturally form a representation of the symmetry group; indeed, the
classification of MPO representations of groups enabled the classification
of symmetry protected phases in 2D. However, there is another way in which
such symmetries can act: Rather than acting as faithful representations
both on the physical and the auxiliary degree of freedom, MPO symmetries
can also appear as symmetries of the auxiliary degrees of freedom on their
own, that is, as pure entanglement symmetries. These entanglement
symmetries are precisely what underlies topological order in 2D, and
they allow to comprehensively understand topologically ordered systems,
encompassing
their ground space structure as well as anyons and their braiding \cite{Schuch2010,Bultinck2017,Cirac2021} (see
\cref{sec:intro:topo} for a detailed discussion).

In many of these cases, MPOs naturally form algebraic structures, such as
group representations in the case of symmetries or evolutions. A
particularly strong structure arises in the case where the
MPO on its own describes a symmetry, such as the entanglement symmetries
which appear in topologically ordered systems: As both products
and linear combinations of symmetries are again symmetries,
the MPOs which appear in
topologically ordered systems naturally form MPO algebras. Given the
widespread use of MPS and MPOs, the importance of understanding
topologically ordered systems, and the key role played by MPO algebras in
this context, it is thus highly desirable to formalize the representation
theory of MPO algebras and their underlying algebraic structure, and to
understand the way in which additional conditions imposed on the
corresponding algebraic structures are reflected in properties of their
MPO representation, and vice versa \cite{Bultinck2017, Kawahigashi2020,Kawahigashi2021,Kawahigashi2021a}.

\subsection{Main results of this work}

In this paper, we lay out a framework for describing MPOs as (faithful)
representations of algebraic structures $\mathcal A$ which are both
algebras and co-algebras (i.e.\ duals of an algebra), with certain
compatibility conditions: pre-bialgebras, weak bialgebras, weak Hopf
algebras, and variants thereof, and we establish a dictionary between the
properties of those algebraic structures $\mathcal A$ and those of their MPO
representations.  This framework allows one to use the algebraic structure of
$\mathcal A$ to reason about MPOs, and conversely to use MPOs to derive
statements about $\mathcal A$.   We then show that under suitable
conditions, those structures -- and thus their MPO representations --
satisfy a condition which we use to define a new algebraic structure:
\emph{pulling-through algebras}; those algebras play a key role in the
construction and study of topologically ordered models from weak Hopf
algebras. We conclude by demonstrating how pulling-through algebras show
up in the construction of a large class of topological models; the
study of further applications in topological
order is left for future work.

Let us describe more specifically the key results of the work.  We start
by showing that MPS naturally appear as representations of coalgebras
(i.e.\ duals of algebras) in \cref{sec:MPS_coalg}. MPOs are then representations of
pre-bialgebras $\mathcal A$  -- coalgebras which are also algebras, with a
minimal compatibility condition (\cref{sec:bialgebra}). In both cases, the coalgebra element
$a\in\mathcal A$ is encoded in the boundary condition of the MPO or MPS.
We then show that semisimplicity of $\mathcal{A}^*$, the dual of $\mathcal{A}$, amounts to an MPS/MPO
which decomposes into irreducible
(i.e., normal/injective) blocks corresponding to the irreps of $\mathcal{A}^*$ -- importantly, for MPOs these irreducible blocks give rise to a
notion of sectors and their fusion, which links to topological order -- and that cocentral elements  are
represented by translational invariant MPOs.

We then study the effect of introducing the additional conditions which
make $\mathcal A$ a weak bialgebra (WBA) or weak Hopf algebra (WHA), respectively
(\cref{sec:WBAWHA}).  We show how those conditions give rise to additional
properties of the sectors of the MPO and their fusion, giving them the structure of a monoidal category in the case of a WBA and multi-fusion category in the case of a WHA. In particular, we show how to
construct a ``vaccum'' sector, as well as an
integral whose representation is a special projector (used later to
construct the topological models).  By further restricting to pivotal,
spherical, and $C^*$ weak Hopf algebras, we then use said integral to show the existence of
a ``pulling through structure'' satisfying a sequence of
increasingly strong conditions.  This motivates the definition of
pulling-through algebras in \cref{sec:pulling_through_alg}, whose MPO representation possesses
such a pulling through structure which satisfies the corresponding
conditions.  In \cref{sec:peps}, we construct PEPS based on pulling-through
algebras, show that their symmetry structure is scale invariant -- central
for follow-up work -- and demonstrate that in the case of $C^*$ Hopf
algebras, the resulting models are equivalent to the class of generalized
Kitaev models, discussed in \cref{sec:intro:topo} below.

This work is the first of a series of papers in which we utilize MPO algebras
to understand and classify topological phases of matter in quantum many
body systems.
The second paper is concerned with classifying phases in 1D
open quantum systems as equivalence classes of shallow circuits. At the
heart of this problem is the classification of renormalization fixed
point density operators, and the  present work provides the tools required
to build examples for fixed point MPOs, as well as channels that map between such
states.
The third paper of the series studies the structure of PEPS
representations of topological models constructed from pulling-through
algebras, as explained above. This is in particular relevant for
characterizing the representation of physical symmetries on the auxiliary
degrees of freedom, and thus the classification of symmetry-enriched
topological phases.

\subsection{Connection to topological order\label{sec:intro:topo}}

As mentioned before, MPO algebras are especially important in the study of
topologically ordered phases in 2D tensor networks, i.e., PEPS. In particular, the
definition of pulling-through algebras in this work is directly motivated
by the study of topological order in PEPS. In the following, we explain
this context, and the way in which our results fit in, in more detail.

Topologically ordered phases are phases which exhibit order which cannot
be detected by any local order parameter. Instead, they are characterized
by a global ordering in their quantum correlations, also known as
entanglement.  Characteristic to these systems are their degenerate ground
states, which are locally indistinguishable and whose number depends on the
topology of the surface on which the system is defined -- both
incompatible with local order parameters -- as well as the presence of
excitations with non-trivial statistics in the system, termed ``anyons''.

In a seminal work, Kitaev~\cite{Kitaev2003} first proposed a Hamiltonian
model for a spin system with the aforementioned properties,
the Toric Code model,
as well as its generalization to finite groups $G$, the quantum double
models.  An alternative construction was provided by the string-net
models of Levin and Wen~\cite{Levin2005a}. While string-net models can be
understood as another way of generalizing the Toric Code model, they are in
fact motivated by topological quantum field theories.  Hence, they
use a category theoretical language (the construction utilizes
unitary fusion categories with some additional restrictions -- the
so-called tetrahedral symmetry of the $F$-symbols), as opposed to the
algebraic approach used for the quantum double models (which are
constructed as representations of the quantum double $D(G)$).
Over the years, generalizations of both classes of models have been
devised.  Kitaev already noted in his original work~\cite{Kitaev2003} that the same
construction also works for semisimple Hopf algebras. This has later been
worked out in detail~\cite{Balsam2012,Buerschaper2013}, and further
generalized to $C^*$-weak Hopf algebras~\cite{Chang2014}.  For string-net
models, it has been shown that the requirement of tetrahedral symmetry can
be dropped~\cite{Hahn2020}, and thus string-net models can be built from
arbitrary unitary fusion categories; further, the construction has been
generalized to build on bimodule categories instead of fusion categories
\cite{Lootens2020}.

Both Kitaev models and string-net models admit tensor network descriptions. Such
a description has been developed first for Kitaev's Toric Code
model~\cite{Verstraete2006a} and later generalized to Kitaev models based on
finite groups~\cite{Schuch2010} and Hopf algebras
\cite{Buerschaper2013}, and separately for string-net
models~\cite{Buerschaper2008,Gu2009a}.  A characteristic feature
of the PEPS description both of Kitaev models based on finite groups
and of string-net models is that the tensors which define the state
possess symmetries which act solely on the auxiliary degrees of freedom of
the tensor. This symmetry is intimately tied to the topological features
of the system, as it  allows to explain both its ground space degeneracy
and the presence of anyonic excitations
\cite{Schuch2010,Sahinoglu2014,Bultinck2017}, whereas breaking it
even slightly
leads to an immediate breakdown of topological order in
2D~\cite{Chen2010a} (but not in 3D~\cite{Williamson2021,Delcamp2021}).
A crucial property of these symmetries is their size-independence: They
are given either by tensor powers of a local symmetry generator (for the double
models) or, more generally, by homogenous MPOs (for string-net models),
such that every region in the PEPS possesses the same MPO symmetries.  The
specific properties of the underlying topological phase can then be
inferred by studying the algebraic properties of the corresponding MPO
algebra.  Remarkably, it is even possible to build topological models in
this very phase from nothing but the MPO symmetry itself: the PEPS tensor
is then constructed from the MPO by placing it on a fixed-size ring with suitable boundaries \cite{Schuch2010,Bultinck2017,Cirac2021}.

Despite the success of MPO symmetries in understanding, characterizing, and
simulating topological order in the phases of the Kitaev double models of
finite groups and of the string-net
models~\cite{Schuch2010,Bal2018,Lootens2021,Duivenvoorden2017}, the
picture is unfortunately not complete. First off, for the Kitaev model
based on Hopf algebras the known tensor network constructions
\cite{Buerschaper2013} do not evidently display any such symmetries; and
moving to an even broader setting, for the Kitaev models constructed from weak Hopf
algebras not even a tensor network description is known -- which, in turn, should
be possible to construct once the underlying MPO symmetries have been
identified. This
lack of knowledge is the more surprising given that weak Hopf algebras
correspond to multi-fusion categories~\cite{Etingof2002}, and thus,
they exhibit the same type of topological order as the corresponding
string-net models.
A key reason why, despite this connection, an understanding of the MPO
symmetries underlying Kitaev models for (weak) Hopf algebras is missing is
the fact that the MPO symmetries for string-net models are constructed in
a category theoretical language.  To understand the MPO symmetries
relevant for describing Kitaev models, however, an algebraic language is
clearly more natural. This is precisely what we achieve in this work:
We show that weak Hopf algebras correspond to MPO algebras with specific
properties, most importantly the pulling-through structure. Using these
MPO algebras, we can then construct a PEPS representation for Kitaev
models based on any weak Hopf algebra;  we explicitly show the connection
between both representations for the case where the MPOs are built from a
$C^*$-Hopf algebra.

On a more abstract level, the relation between the MPO symmetries of the
string-net models and the MPOs representations of semisimple weak Hopf
algebras is as follows.
The MPOs constructed in this work are representations of semisimple weak
Hopf algebras. In turn, representations of semisimple weak Hopf algebras
are known to be exactly multi-fusion categories \cite{Etingof2002}.
Vice versa, every (multi-)fusion category arises as the category of
representations of some weak Hopf algebra. The MPO symmetries of the
string-net model based on any given fusion category will thus form a
representation of the corresponding weak Hopf algebra. In fact, that weak
Hopf algebra can be constructed from the MPO symmetries themselves, when
closed with arbitrary boundary conditions. It is important to note that
while (ordinary) string-net models are built on a single fusion category,
general representations of weak Hopf algebras involve two different fusion
categories: one is the representation category of the weak Hopf algebra,
the other is the category (or a subcategory) of the representations of the
dual weak Hopf algebra. The correspondence between fusion categories,
MPOs, and weak Hopf algebras then suggests that string-net models based on
bimodule categories~\cite{Lootens2020} and PEPS constructed from weak Hopf
algebras are actually the same. The exact details on how to map these
models to each other is left for future work.

\subsection{Structure of the paper}

This paper is structured as follows. First, in \cref{sec:tensor_calculus}
we introduce the graphical notation of tensor calculus used throughout the
paper. In \cref{sec:MPS_coalg} we introduce coalgebras and matrix product
states (MPS), and establish the relation between them: one can think of an
MPS as a representation of a coalgebra. We specialize to the case where
the coalgebra is cosemisimple; in this case the representing MPS is in
canonical form, i.e.\ it is a sum of smaller bond dimensional injective
MPS.  In \cref{sec:bialgebra} we investigate pre-bialgebras: coalgebras
that are algebras as well such that the two structures satisfy a
compatibility condition. Correspondingly, the previously defined MPS
representations of pre-bialgebras will become MPOs and these MPOs are
closed under multiplication. In \cref{sec:WBAWHA} we introduce weak
bialgebras (pre-bialgebras such that their unit and counit satisfy certain
properties) and weak Hopf algebras (weak bialgebras with an extra
operation called the antipode). The main result of this section is
\cref{thm:wha_special_integral}, which proves the above mentioned special
integral in cosemisimple WHAs over $\mathbb{C}$. In \cref{sec:pivotal_WHA}
we specialize to pivotal WHAs, which, as mentioned above, allows us to
define a ``pulling-through'' structure on the WHA. In \cref{sec:C_star_WHA}
we further specialize to $C^*$-WHAs,  which guarantees extra properties of
this pulling-through structure. In \cref{sec:pulling_through_alg} we
investigate pulling-through algebras independently from the previous weak
Hopf algebra properties. We develop a graphical language suited for these
algebras that we then use in \cref{sec:peps} to define PEPS that have
stable entanglement symmetries described by this pulling-through algebra, and
which we show to precisely describe Kitaev's quantum double models in the
case of Hopf algebras.

\section{Graphical notation of tensor calculus}\label{sec:tensor_calculus}

In this section we introduce the graphical notation of tensor calculus that we use throughout the paper. This graphical notation is especially useful to visualize equations that involve the contraction of many higher rank tensors (i.e.\ tensors with more than two indices) and it is standard in the field of tensor networks. In this paper, however, we face extra challenges as we use the graphical language parallel to an algebraic one, and thus we have to modify the usual graphical language in order to be able to translate between the two languages.

From a computational point of view, tensors are just multi-dimensional arrays. In the usual graphical notation of tensor calculus, one denotes tensors by dots (and various shapes) with lines connected to them. The number of lines connected to the dot (or other shape) is the rank of the tensor, and each line corresponds to one of the vector spaces in the tensor product. For example, the following diagrams represent a scalar, a vector and a matrix, respectively:
\begin{equation*}
  s =
  \begin{tikzpicture}[baseline=-1mm]
    \node[tensor, label=below:$s$] {};
  \end{tikzpicture} \ , \quad
  v =
  \begin{tikzpicture}[baseline=-1mm]
    \node[tensor, label=below:$v$] (v) {};
    \draw (v)--++(-0.5,0);
  \end{tikzpicture} \ , \quad
  A =
  \begin{tikzpicture}[baseline=-1mm]
    \node[tensor, label=below:$A$] (A) {};
    \draw (-0.5,0) -- (A);
    \draw (A) -- (0.5,0);
  \end{tikzpicture} \ .
\end{equation*}
Tensor contraction is denoted by joining the lines corresponding to the contracted indices of the two tensors. For example the scalar product of two vectors, a matrix acting on a vector and the product of two matrices are denoted by the following diagrams, respectively:
\begin{equation*}
  \sum_i w_i v_i =
  \begin{tikzpicture}[baseline=-1mm]
    \node[tensor, label=below:$w$] at (0,0) (v) {};
    \node[tensor, label=below:$v$] at (1,0) (w) {};
    \draw (v)--(w);
  \end{tikzpicture} \ , \quad
  \sum_j A_{ij} v_j =
  \begin{tikzpicture}[baseline=-1mm]
    \node[tensor, label=below:$A$] at (0,0) (A) {};
    \node[tensor, label=below:$v$] at (0.7,0) (v) {};
    \draw (v)--(A);
    \draw (A)--+(-0.5,0);
  \end{tikzpicture} \ , \quad
 \sum_{j}A_{ij}B_{jk} =
 \begin{tikzpicture}[baseline=-1mm]
    \node[tensor, label=below:$A$] at (0,0) (A) {};
    \node[tensor, label=below:$B$] at (0.7,0) (B) {};
    \draw (B)--(A)--++(-0.5,0);
    \draw (B)--++(0.5,0);
 \end{tikzpicture} \ .
\end{equation*}
Here we have made the implicit assumption that the first index is the left line, the second one is the right one. For higher rank tensors and more complicated contraction schemes one has to keep track of which index belongs to which line (by e.g.\ fixing a convention such as in the previous figure). In the following we outline a notation that allows us to distinguish the different indices of the tensor without requiring them to be always at the same position. This notation is thus suitable to depict more complicated tensor constructions such as the definition of a PEPS.

To introduce our modification of the graphical notation, let us first formalize what tensors and tensor contractions are. Rank-$n$ tensors are elements of the tensor product $V_1 \otimes \dots \otimes V_n$ for some finite dimensional vector spaces $V_1 \dots V_n$. One can naturally take tensor products of tensors: for example, if $V_1$, $V_2$ and $V_3$ are vector spaces and  $r = v_1 \otimes v_2 \in V_1 \otimes V_2$ and $s=v_3\otimes f \in V_3\otimes V_2^*$, where $V_2^*$ denotes the space of linear functionals on $V_2$, then their tensor product $t=r\otimes s$ is $t=v_1 \otimes v_2 \otimes v_3 \otimes f$. Tensor contraction (without introducing a scalar product) is then the following operation: if amongst the $n$ components of the tensor product both a vector space $V$ and its dual $V^*$ appears, then one can form a rank-$(n-2)$ tensor by acting with the linear functional in $V^*$ on the vector in $V$. For example, the tensor $t=v_1 \otimes v_2 \otimes v_3 \otimes f$ defined above is an element of the space $V_1 \otimes V_2 \otimes V_3 \otimes V_2^*$, and thus one can contract its second and fourth components to obtain a rank-two tensor $\mathcal{C}_{24}(t) = f(v_2) \cdot v_1 \otimes v_3 \in V_1 \otimes V_3$.

As we have seen, to make sense of tensor contraction without a scalar product, it is important to differentiate between vector spaces and linear functionals. We will denote indices corresponding to vectors by outgoing arrows, while indices corresponding to linear functionals by incoming arrows. Moreover, to distinguish between the different indices ot the tensor, we will label the lines with vector spaces. For outgoing arrow, the label is the corresponding tensor component. For incoming arrow, the label is the vector space the given tensor component is the dual of. For example, a vector $v\in V$, a linear functional $f\in V^*$ and a rank-two tensor $A\in V\otimes V^*$ is denoted by
\begin{equation*}
  v =
 \ .
\end{equation*}
Tensor contraction is still denoted by joining lines. Note, however, that now only lines with the same label can be joined that also point in the same direction. For example,
\begin{equation*}
  f(v) =
  %
 \ .
\end{equation*}

The vector space $V$ is canonically isomorphic to $V^{**}$, and thus one can equally think of $V$ as the space of linear functionals on $V^*$. Correspondingly, in our graphical notation every arrow can be reversed by changing the label from $V$ to $V^*$. For example, there are four different ways to depict a rank-two tensor $A\in V\otimes V^*$:
\begin{equation*}
  %
 \ .
\end{equation*}
The first depiction of $A$ suggests to interpret it as a linear map $\hat{A}:V \to V$, while the last equation suggests to interpret it as a linear map $V^* \to V^*$; this linear map is $\hat{A}^T$. Such rank-two tensors are  sometimes labeled by the linear map $\hat{A}$; when this is the case, we will try to be consistent and label the vector spaces by $V$ and not by $V^*$.

As we will depict tensor networks in two dimensions,  sometimes  it will be convenient to rotate tensors. This means that vectors don't always point to the left, and thus we actually need the arrows to distinguish between the two indices of a rank-two tensor. Such a rotation is, for example, the following:
\begin{equation*}
 \ .
\end{equation*}

In the rest of the paper we will often deal with rank-three and rank-four tensors of the form $A \in V \otimes W
\otimes W^*$ and $O \in V\otimes V^* \otimes W \otimes W^*$. These tensors are denoted by
\begin{equation*}
  A =
  %
    \ .
\end{equation*}
For better visual distinction, we have introduced colors: the red colored edges are labeled by $W$ and the black colored edges by $V$. In fact, in cases where the vector spaces $V$ and $W$ are fixed and we only need to distinguish between $V,V^*,W$ and $W^*$, we will drop the labels and keep the colors only, denoting $A$ and $O$ by
\begin{equation*}
  A =
  %
    \ .
\end{equation*}

\section{Coalgebras and matrix product states}\label{sec:MPS_coalg}

In this section we define coalgebras and matrix product states and show that MPS can be thought of as representations of coalgebras. We also show that this correspondence holds the other way around as well: given an MPS tensor, one can construct a coalgebra such that the MPS forms a representation of the constructed coalgebra. This observation makes thus MPS and coalgebras completely equivalent. We elaborate on a special case: when the coalgebra is cosemisimple, the corresponding MPS tensor is a sum of injective tensors (see \cref{def:injectivity}), and vice versa, given an MPS that is a sum of injective tensors, the constructed coalgebra is cosemisimple. We also define the notion of cocentral and non-degenerate coalgebra elements and show how these properties are reflected in the MPS representation.

We start by defining coalgebras.

\begin{definition}[Coalgebra]
  The triple $(\mathcal{C},\Delta,\epsilon)$ is a \emph{coalgebra} if $\mathcal{C}$ is a finite dimensional vector space over $\mathbb{C}$, $\Delta:\mathcal{C}\to \mathcal{C}\otimes\mathcal{C}$ is a linear map called comultiplication such that it is associative:
  \begin{equation*}
      (\Delta\otimes\id ) \circ \Delta = (\id \otimes \Delta) \circ \Delta,
  \end{equation*}
  and $\epsilon\in\mathcal{C}^*$ is a linear functional called \emph{counit}, such that
    \begin{equation*}
      (\epsilon\otimes\id)\circ\Delta = (\id\otimes\epsilon)\circ \Delta = \id.
    \end{equation*}
\end{definition}
Coalgebras emerge as the  dual of algebras: Given a finite dimensional algebra $\mathcal{A}$ with product $\mu_{\mathcal{A}}:\mathcal{A}\otimes\mathcal{A}\to \mathcal{A}$ and unit $1$, we can define a coproduct on $\mathcal{A}^*$ by defining $\Delta_{\mathcal{A}^*}:\mathcal{A}^*\to\mathcal{A}^*\otimes \mathcal{A}^*$ as $\Delta_{\mathcal{A}^*} = \mu_{\mathcal{A}}^T$, that is, by
\begin{equation*}
  \Delta_{\mathcal{A}^*}(f)(x\otimes y) = f(xy),
\end{equation*}
where $x,y\in \mathcal{A}$ and $f\in \mathcal{A}^*$: Associativity of $\Delta_{\mathcal{A}^*} = \mu_{\mathcal{A}}^T$ is equivalent to the associativity of $\mu_{\mathcal{A}}$, and the map given by $f\mapsto f(1)$ defines the counit of $\mathcal{A}^*$. Vice versa, if $\mathcal{C}$ is a coalgebra with coproduct $\Delta_{\mathcal{C}}$ and counit $\epsilon$, then we can  naturally give $\mathcal{C}^*$ an algebra structure by defining the product via $\mu_{\mathcal{C}^*} = \Delta_{\mathcal{C}}^T$, i.e.\ by
\begin{equation*}
  (fg)(x) = (f\otimes g)\circ \Delta_{\mathcal{C}} (x),
\end{equation*}
where $f,g\in\mathcal{C}^*$ and $x\in \mathcal{C}$; the unit of $\mathcal{C}^*$ is then $\epsilon$.

Associativity of the coproduct allows us to write $\Delta^2(x)$ instead of $(\Delta\otimes\id ) \circ \Delta$, and  $\Delta^n(x)$ for $n$ repeated application of $\Delta$ on $x$. In the following we will use Sweedler's notation of the coproduct and write
\begin{equation*}
  \Delta^n(x) = \sum_{(x)} x_{(1)} \otimes x_{(2)} \otimes \dots \otimes x_{(n+1)}.
\end{equation*}
We will show below that this shorthand notation actually hides a more complicated sum that has a special structure called matrix product state.
\begin{definition}[Matrix product states]
  Let $(V_i)_{i=1}^n$ and $(W_j)_{j=0}^n$ be two collections of finite dimensional vector spaces over $\mathbb{C}$. An MPS is given by tensors $A_i \in V_i \otimes W_{i-1} \otimes W_i^*$ ($i=1,\ldots, n$), $A_i = \sum_k \ket{k} \otimes A_i^k$ and a matrix $X\in W_n \otimes W_0^*$; the state generated by the MPS is given by
  \begin{equation*}
    \ket{\Psi} = 	\sum_k \tr \left( X \cdot A_1^{k_1} \cdots A_n^{k_n} \right)\  \ket{k_1 \dots k_n} = \
    	  \begin{tikzpicture}
    	   \draw[virtual] (0.5,0) rectangle (5.5,-0.5);
         \node[tensor,label=below:$X$] (x) at (1,0) {};
    	    \foreach \x/\t in {2/1,3/2,5/n}{
    	      \node[tensor,label=below:$A_\t$] (t\x) at (\x,0) {};
    	      \draw[->-] (t\x) --++ (0,0.5);
    	    }
    	    \node[fill=white] (dots) at (4,0) {$\dots$};
          \draw[virtual,->-] (t2) -- (x);
          \draw[virtual,->-] (t3) -- (t2);
          \draw[virtual,->-] (dots) -- (t3);
          \draw[virtual,->-] (t5) -- (dots);
          \draw[virtual,->-] (0.5,-0.5) -- (5.5,-0.5);
    	  \end{tikzpicture}  \ .
  \end{equation*}
  We say that the MPS is translation invariant with open boundary condition if $V_1 = \dots = V_n$, $W_0=\dots = W_n$ and $A_1 = \dots = A_n$.
\end{definition}

Let us now show that in any coalgebra $\mathcal{C}$ the repeated coproduct of an element $x\in\mathcal{C}$, $\Delta^{n-1}(x)$, can be represented as an MPS on $n$ sites.
\begin{theorem}[MPS representation of coalgebras]\label{thm:coalg_to_mps}
  Let $\mathcal{C}$ be a coalgebra, $V_i$ be finite dimensional vector spaces over $\mathbb{C}$  and $\phi_i : \mathcal{C}\to V_i$  linear maps ($i=1,\ldots, n$). Let $W$ be a vector space and $\psi: \mathcal{C}^* \to \End(W)$ be an injective representation of the algebra $\mathcal{C}^*$. Let $A_i \in V_i \otimes \End(W)$ be defined by
  \begin{equation*}
    \begin{tikzpicture}
      \node[tensor, label = below: $A_i$] (t) at (0,0) {};
      \draw[virtual, -<-] (t) --++ (0.5,0);
      \draw[virtual, ->-] (t) --++ (-0.5,0);
      \draw[->-] (t) --++ (0,0.5);
    \end{tikzpicture} =
    \sum_{x\in B} \phi_i(x) \otimes \psi(\delta_x),
  \end{equation*}
  where $B$ is a basis of $\mathcal{C}$, and $\delta_x$ denotes the dual basis elements (i.e.\ $\delta_x\in \mathcal{C}^*$ and $\delta_x(y) = \delta_{x,y}$ for  any $x,y\in B$). Then for all $x\in\mathcal{C}$ there exists a matrix $b(x)\in \End(W)$ such that for all $n>0$,
  \begin{equation*}
    (\phi_1 \otimes \dots \otimes \phi_n) \circ \Delta^{n-1} (x) =
    	  \begin{tikzpicture}
    	    \draw[virtual] (0.5,0) rectangle (5.5,-0.5);
         \node[tensor,label=below:$b(x)$] (x) at (1,0) {};
    	    \foreach \x/\t in {2/1,3/2,5/n}{
    	      \node[tensor,label=below:$A_\t$] (t\x) at (\x,0) {};
    	      \draw[->-] (t\x) --++ (0,0.5);
    	    }
    	    \node[fill=white] (dots) at (4,0) {$\dots$};
          \draw[virtual,->-] (t2) -- (x);
          \draw[virtual,->-] (t3) -- (t2);
          \draw[virtual,->-] (dots) -- (t3);
          \draw[virtual,->-] (t5) -- (dots);
          \draw[virtual,->-] (0.5,-0.5) -- (5.5,-0.5);
    	  \end{tikzpicture}  \ .
  \end{equation*}
\end{theorem}

Let us remark that the tensor $A_i$ is independent of the concrete choice of the basis $B$ of $\mathcal{C}$, as the expression $\sum_{x \in B} x\otimes \delta_x$ is independent of $B$. Let us also remark that might be many different choices for the matrix $b(x)$ satisfying the required equation.

\begin{proof}
  Notice that
  \begin{equation*}
    \sum_{x,y} f(x) \cdot g(y) \cdot \delta_x \delta_y = fg = \sum_{z\in B} fg(z) \cdot \delta_z = \sum_{z\in B} \sum_{(z)} f(z_{(1)}) \cdot g(z_{(2)}) \cdot \delta_z,
  \end{equation*}
  where in both equations we have used that for any $f\in\mathcal{C}^*$, $ f = \sum_{x\in B} f(x)\cdot \delta_x$. As this is true for all $f,g\in\mathcal{C}^*$, we conclude that
  \begin{equation*}
    \sum_{x,y\in B} x\otimes y \otimes \delta_x \delta_y = \sum_{z\in B} \sum_{(z)} z_{(1)} \otimes z_{(2)} \otimes \delta_z = \sum_{z\in B} \Delta(z) \otimes \delta_z.
  \end{equation*}
  This means that
  \begin{equation*}
 =
    \sum_{x,y\in B} \phi_1(x) \otimes \phi_2(y) \otimes \psi(\delta_x \delta_y) =   \sum_{z\in B}  (\phi_1 \otimes \phi_2) \circ \Delta (z) \otimes \psi(\delta_z).
  \end{equation*}
  A similar equation holds for any number of consecutive tensors, i.e.\
  \begin{equation*}
    	  %
  \  =
        \sum_{x\in B} (\phi_1 \otimes \dots \otimes \phi_n) \circ \Delta^{n-1} (x) \otimes \psi(\delta_x).
  \end{equation*}
  As $\psi$ is an injective representation of $\mathcal{C}^*$, there exists a matrix $b(x)$ for all $x\in \mathcal{C}$  such that $f(x) = \tr \left( b(x) \cdot \psi(f)\right)$ holds for all $f\in\mathcal{C}^*$. Note that this equation does not uniquely determine $b(x)$; there may be many different choices of $b(x)$ satisfying this equation. Using any such choice of $b(x)$, the following holds:
  \begin{equation*}
    %
   = \sum_{y\in B} \tr \left(b(x) \cdot \psi(\delta_y)\right) \cdot \phi_i(y) = \sum_{y\in B} \delta_y(x) \cdot \phi_i(y)  = \phi_i(x).
  \end{equation*}
  Finally, combining this statement with the previous one results in the desired equation
  \begin{equation*}
    (\phi_1 \otimes \dots \otimes \phi_n) \circ \Delta^{n-1} (x) =
    	  %
  \ .
  \end{equation*}
\end{proof}
From now on, unless otherwise specified, we only consider translation invariant representations of coalgebras. This restriction is not essential and most statements obviously generalize to non-TI representations, but it eases the notation. For example, as all tensors are the same, we will drop the label $A$ denoting the MPS tensor  and simply write
\begin{equation*}
  (\phi \otimes \dots \otimes \phi) \circ \Delta^{n-1} (x) =
  	  %
  \ .
\end{equation*}

Let us now show that the construction of \cref{thm:coalg_to_mps} can be reversed in the translation invariant case: given a translation invariant MPS with open boundary condition, one can construct a coalgebra $\mathcal{C}$ such that the MPS becomes a representing MPS of $\mathcal{C}$: Indeed, let $A\in V\otimes W \otimes W^*$ be an MPS tensor and let us write $A = \sum_i \ket{i} \otimes A^i$ ($i=1,\ldots, d$ with $d=\mathrm{dim}(V)$) with $A^i \in \End(W)$. Let us define the algebra $\mathcal{A}$ as
\begin{equation*}
  \mathcal{A} = \operatorname{Span}  \bigcup_{n\in \mathbb{N}} \left\{ A^{i_1} A^{i_2} \cdots A^{i_n} \middle | (i_1, \ldots, i_n) \in \{ 1,\ldots, d\}^{\times n} \right\},
\end{equation*}
and let $\mathcal{C} = \mathcal{A}^*$. Let us fix now a basis $B$ of $\mathcal{C}$; elements of this basis are denoted by $x,y,\ldots$. This basis then also fixes a basis (the dual basis) on $\mathcal{C}^* = \mathcal{A}$. Elements of this basis are denoted by $\delta_x,\delta_y,\ldots$. By definition the matrices $A^i$ are elements of $\mathcal{A}$, and thus one can expand them in this (dual) basis. One can thus write
$A = \sum_{x\in B} \ket{v_x} \otimes  \delta_x$ for some vectors $\ket{v_x}\in V$. This tensor then can be interpreted as a map $\phi:\mathcal{C}\to V$ by $\phi(x) = \sum_{y\in B} \delta_y(x) \ket{v_y}$. Note that -- by the definition of the dual basis -- for any $x\in B$, $\phi(x) = \ket{v_x}$. This implies that the following equation also holds: $A  = \sum_{x\in B} \phi(x) \otimes \psi(\delta_x)$, with $\psi=\id$. Finally, as any element $x\in\mathcal{C}$ is a linear functional on $\mathcal{A}$ ($\mathcal{C} = \mathcal{A}^*$), one can find a matrix $b(x)$ such that $x(m) = \tr (b(x) m)$ for all matrices $m\in \mathcal{A}$. With this, we have obtained that for any $x\in \mathcal{C}$,
\begin{equation*}
  \phi^{\otimes n} \circ \Delta^{n-1} (x) =
  \begin{tikzpicture}
    \draw[virtual] (0.5,0) rectangle (5.5,-0.5);
    \node[tensor,label=below:$b(x)$] (x) at (1,0) {};
    \foreach \x/\t in {2/1,3/2,5/n}{
      \node[tensor,label=below:$A$] (t\x) at (\x,0) {};
      \draw[->-] (t\x) --++ (0,0.5);
    }
    \node[fill=white] (dots) at (4,0) {$\dots$};
    \draw[virtual,->-] (t2) -- (x);
    \draw[virtual,->-] (t3) -- (t2);
    \draw[virtual,->-] (dots) -- (t3);
    \draw[virtual,->-] (t5) -- (dots);
    \draw[virtual,->-] (0.5,-0.5) -- (5.5,-0.5);
  \end{tikzpicture}  \ ,
\end{equation*}
i.e.\ the MPS defined above forms a representation of $\mathcal{C}$ with the properties listed in \cref{thm:coalg_to_mps}.

\subsection{Cosemisimplicity and injectivity of the representing MPS}

In this section we introduce cosemisimple coalgebras as well as injective MPS and examine the connection between these two properties: the MPS representation of a cosemisimple coalgebra decomposes into a sum of injective MPS, and conversely, given an MPS that decomposes into a direct sum of injective MPS,  the corresponding coalgebra is cosemisimple.

Let $\mathcal{C}$ be a coalgebra (finite dimensional, over $\mathbb{C}$). As we have seen in the previous section, $\mathcal{C}^*$ has a natural algebra structure. In the following we will talk about representations of this algebra $\mathcal{C}^*$, as the MPS construction in the previous section uses the representations of that algebra. Recall that two representations $\psi_1: \mathcal{C}^*\to \End(W_1)$ and $\psi_2: \mathcal{C}^*\to \End(W_2)$ are called equivalent if there is an invertible linear map $Z:W_1\to W_2$ such that $\psi_1(f) = Z^{-1} \cdot \psi_2(f) \cdot Z$ for all $f\in\mathcal{C}^*$. In particular, the dimension of the two equivalent representations coincide, $\mathrm{dim}(W_1) = \mathrm{dim}(W_2)$. The set of irreducible representation (irrep) equivalence classes of the algebra $\mathcal{C}^*$ is denoted by $\irr(\mathcal{C}^*)$, and the elements of this set (i.e.\ the different irrep equivalence classes) will be denoted by small Roman letters $a,b,c,\ldots$. The dimension of (all) irreps from the class $a$ will be denoted by $D_a$. For convenience, let us fix a concrete representation $\psi_{a}$ on vector space $W_a$ from each irrep class $a$. Recall that by the density theorem \cite{Etingof}, $\psi_a(\mathcal{C}^*) = \End(W_a) \simeq \mathcal{M}_{D_a}$ for all irreps $\psi_a$, where $\mathcal{M}_{D_a}$ denotes the set of $D_a\times D_a$ matrices over $\mathbb{C}$; in fact $\phi(\mathcal{C}^*) = \bigoplus_{a\in I \subseteq \irr(\mathcal{C}^*)} \mathcal{M}_{D_a} \otimes \id_{m_a}$ for all representations of the form $\phi(x) = \bigoplus_{a\in I \subseteq \irr(\mathcal{C}^*)} \phi_a(x)\otimes \id_{m_a}$.

We say that the coalgebra $\mathcal{C}$ is \emph{cosemisimple} if the algebra $\mathcal{C}^*$ is semisimple, i.e.\ if $\mathcal{C}^*\simeq \bigoplus_{a\in \irr(\mathcal{C}^*)} \mathcal{M}_{D_a}$. In particular, $\mathcal{C}^*$ has finitely many irrep classes. If $\mathcal{C}^*$ is semisimple, then any representation $\psi$ of it, up to a basis transformation, is of the form
\begin{equation*}
  \psi(x) \simeq \bigoplus_{a\in I \subseteq \irr(\mathcal{C}^*)} \psi_a (x) \otimes \id_{m_a},
\end{equation*}
where $a$ runs over a subset $I$ of the equivalence classes of irreps of $\mathcal{C}^*$ and $\psi_a$ are the previously fixed representatives from the class $a$, and the numbers $m_a$ denote the multiplicity of the irrep $\psi_a$ in the decomposition of $\psi$. The representation $\psi$ is injective if and only if all irrep classes are present in this decomposition\footnote{Note that even if $\mathcal{C}^*$ arises as the group algebra $\mathbb{C}[G]$ of some finite group $G$, we are talking about the representation of the group algebra, and not the representation of the group; that is, we require the map $\phi: \mathbb{C}[G]\to \End(V)$ to be injective, not the map $\phi|_G: G\to \End(V)$. }, i.e.\ if $I=\irr(\mathcal{C}^*)$. As $\psi$ is a direct sum of irreps, the density theorem applies and thus $\psi(\mathcal{C}^*) = \bigoplus_{a\in I \subseteq \irr(\mathcal{C}^*)} \mathcal{M}_{D_a} \otimes \id_{m_a}$.

The MPS tensor $A$ in \cref{thm:coalg_to_mps} is constructed using a representation $\psi$ of $\mathcal{C}^*$ that is assumed to be injective. Therefore, if $\mathcal{C}$ is cosemisimple, $A$ decomposes as
\begin{equation*}
  A \simeq \sum_{x\in B} \phi(x) \otimes \bigoplus_{a\in \irr(\mathcal{C}^*)} \psi_a (\delta_x) \otimes \id_{m_a}.
\end{equation*}
As the defining property of $b(x)$ is that $\tr \left(b(x) \psi(f)\right) = f(x)$, the matrix $b(x)$  can also be chosen w.l.o.g.\  in this form, i.e.\ such that in the same basis as $\psi$, it reads
\begin{equation*}
   b(x) \simeq \bigoplus_{a\in \irr(\mathcal{C}^*)} b_a (x) \otimes \id_{m_a}.
\end{equation*}
If $b(x)$ is in this form, then it is uniquely defined by the equation  $\tr \left(b(x) \psi(f)\right) = f(x)$ and the map $x\mapsto b(x)$ is linear and a bijection between $\mathcal{C}$ and $\bigoplus_{a\in \irr(\mathcal{C}^*)} \mathcal{M}_{D_a}$.

We have thus obtained that the MPS representing a cosemisimple coalgebra decomposes into a sum of MPSs with smaller bond dimension,
\begin{equation*}
   \tikzsetnextfilename{61048b72-7ab1-44b0-a7b8-c3db25636324}
  \begin{tikzpicture}
    \draw[virtual] (0.5,0) rectangle (5.5,-0.5);
    \node[tensor,label=below:$b(x)$] (x) at (1,0) {};
    \foreach \x/\t in {2/1,3/2,5/n}{
      \node[tensor] (t\x) at (\x,0) {};
      \draw[->-] (t\x) --++ (0,0.5);
    }
    \node[fill=white] (dots) at (4,0) {$\dots$};
    \draw[virtual,->-] (t2) -- (x);
    \draw[virtual,->-] (t3) -- (t2);
    \draw[virtual,->-] (dots) -- (t3);
    \draw[virtual,->-] (t5) -- (dots);
    \draw[virtual,->-] (0.5,-0.5) -- (5.5,-0.5);
  \end{tikzpicture}  \ =
  \sum_{a\in \irr(\mathcal{C}^*)} m_a \cdot
  \tikzsetnextfilename{385cccbc-93b8-48cf-a742-436e2cc9f5fe}
  \begin{tikzpicture}
    \draw[virtual] (0.5,0) rectangle (5.5,-0.5);
    \node[tensor,label=below:$b_a(x)$] (x) at (1,0) {};
    \foreach \x/\t in {2/1,3/2,5/n}{
      \node[tensor] (t\x) at (\x,0) {};
      \draw[->-] (t\x) --++ (0,0.5);
    }
    \node[fill=white] (dots) at (4,0) {$\dots$};
    \draw[virtual,->-] (t2) -- (x) node[midway,black,irrep,above] {$W_a$};
    \draw[virtual,->-] (t3) -- (t2) node[midway,black,irrep,above] {$W_a$};
    \draw[virtual,->-] (dots) -- (t3) node[midway,black,irrep,above] {$W_a$};
    \draw[virtual,->-] (t5) -- (dots) node[midway,black,irrep,above] {$W_a$};
    \draw[virtual,->-] (0.5,-0.5) -- (5.5,-0.5);
  \end{tikzpicture}  \ ,
\end{equation*}
where
\begin{equation*}
  \begin{tikzpicture}
    \node[tensor] (t) at (0,0) {};
    \draw[virtual, -<-] (t) --++ (0.7,0) node[irrep] {$W_a$};
    \draw[virtual, ->-] (t) --++ (-0.7,0) node[midway,black,irrep,above] {$W_a$};
    \draw[->-] (t) --++ (0,0.5);
  \end{tikzpicture} =
  \sum_{x\in B} \phi(x) \otimes \psi_a(\delta_x),
\end{equation*}
for the previously fixed irrep representatives $\psi_a$ in the irrep class $a$. In the following, we will choose the representation $\psi$ such that it contains exactly one irrep from each irrep class (i.e.\ such that $m_a=1$). The previously fixed irrep representatives $\psi_a$ determine the vector spaces $W_a$ for each $a$, and thus from now on, we do not display the labels $W_a$ in graphical representation, only the label $a$. Therefore we will write
\begin{equation*}
   \tikzsetnextfilename{61048b72-7ab1-44b0-a7b8-c3db25636324}
  \begin{tikzpicture}
    \draw[virtual] (0.5,0) rectangle (5.5,-0.5);
    \node[tensor,label=below:$b(x)$] (x) at (1,0) {};
    \foreach \x/\t in {2/1,3/2,5/n}{
      \node[tensor] (t\x) at (\x,0) {};
      \draw[->-] (t\x) --++ (0,0.5);
    }
    \node[fill=white] (dots) at (4,0) {$\dots$};
    \draw[virtual,->-] (t2) -- (x);
    \draw[virtual,->-] (t3) -- (t2);
    \draw[virtual,->-] (dots) -- (t3);
    \draw[virtual,->-] (t5) -- (dots);
    \draw[virtual,->-] (0.5,-0.5) -- (5.5,-0.5);
  \end{tikzpicture}  \ =
  \sum_{a\in \irr(\mathcal{C}^*)} m_a \cdot
  \tikzsetnextfilename{6c743af8-405a-477f-9390-12a2d8c5064f}
  \begin{tikzpicture}
    \draw[virtual] (0.5,0) rectangle (5.5,-0.5);
    \node[tensor,label=below:$b_a(x)$] (x) at (1,0) {};
    \foreach \x/\t in {2/1,3/2,5/n}{
      \node[tensor] (t\x) at (\x,0) {};
      \draw[->-] (t\x) --++ (0,0.5);
    }
    \node[fill=white] (dots) at (4,0) {$\dots$};
    \draw[virtual,->-] (t2) -- (x) node[midway,black,irrep,above] {$a$};
    \draw[virtual,->-] (t3) -- (t2) node[midway,black,irrep,above] {$a$};
    \draw[virtual,->-] (dots) -- (t3) node[midway,black,irrep,above] {$a$};
    \draw[virtual,->-] (t5) -- (dots) node[midway,black,irrep,above] {$a$};
    \draw[virtual,->-] (0.5,-0.5) -- (5.5,-0.5);
  \end{tikzpicture}  \ ,
\end{equation*}
with
\begin{equation}\label{eq:small_mps_tensor}
  \begin{tikzpicture}
    \node[tensor] (t) at (0,0) {};
    \draw[virtual, -<-] (t) --++ (0.7,0) node[midway,black,irrep,above] {$a$};
    \draw[virtual, ->-] (t) --++ (-0.7,0) node[midway,black,irrep,above] {$a$};
    \draw[->-] (t) --++ (0,0.5);
  \end{tikzpicture} =
  \sum_{x\in B} \phi(x) \otimes \psi_a(\delta_x).
\end{equation}
The MPS tensors in \cref{eq:small_mps_tensor}, provided that $\phi$ has certain properties, are special:
\begin{definition}[Injective and normal MPS tensor]\label{def:injectivity}
  An MPS tensor $A\in V\otimes \End(W)$, $A = \sum_i \ket{i}\otimes A^i$ is \emph{normal} if there is an $n\in\mathbb{N}$ such that
  \begin{equation*}
    \operatorname{Span} \left\{ A^{i_1} \cdots A^{i_n} \middle| (i_1, \ldots, i_n) \in \{1, \ldots, \mathrm{dim}(V) \}^{\times n}  \right\} = \End(W).
  \end{equation*}
  The tensor is called \emph{injective} if it is normal with $n=1$.
\end{definition}
It is immediate to see that the MPS tensors defined in \cref{eq:small_mps_tensor} are injective if and only if $\phi$ is injective and they are normal if and only if there is an $n$ such that $\phi^{\otimes n} \circ \Delta^{n-1}$ is injective. We will call such a linear map \emph{normal}.

We have thus obtained that if $\mathcal{C}$ is a cosemisimple coalgebra and $\phi$ is a $\mathcal{C}\to V$ linear map such that $\phi^{\otimes n} \circ \Delta^{n-1}$ is injective for some $n$, then the coproduct of an element $x$ has a special MPS representation of the form
\begin{equation*}
  \phi^{\otimes n} \circ \Delta^{n-1} (x) =
  \sum_{a\in \irr(\mathcal{C}^*)}
  \begin{tikzpicture}
    \draw[virtual] (0.5,0) rectangle (5.5,-0.5);
    \node[tensor,label=below:$b_a(x)$] (x) at (1,0) {};
    \foreach \x/\t in {2/1,3/2,5/n}{
      \node[tensor] (t\x) at (\x,0) {};
      \draw[->-] (t\x) --++ (0,0.5);
    }
    \node[fill=white] (dots) at (4,0) {$\dots$};
    \draw[virtual,->-] (t2) -- (x) node[midway,black,irrep,above] {$a$};
    \draw[virtual,->-] (t3) -- (t2) node[midway,black,irrep,above] {$a$};
    \draw[virtual,->-] (dots) -- (t3) node[midway,black,irrep,above] {$a$};
    \draw[virtual,->-] (t5) -- (dots) node[midway,black,irrep,above] {$a$};
    \draw[virtual,->-] (0.5,-0.5) -- (5.5,-0.5);
  \end{tikzpicture}  \ ,
\end{equation*}
where each MPS tensor is normal. This statement now can be reversed: let us consider a set $\mathcal{S}$ of injective MPS tensors $A\in V \otimes \End(W_A)$ for each $A\in \mathcal{S}$ such that no two of them are related to each other with a basis transformation. One then can construct a \emph{cosemisimple} coalgebra $\mathcal{C}$ (using the construction from the previous section), a map $\phi: \mathcal{C}\to V$ and a bijection $B: \bigoplus \End(W_A) \to \mathcal{C}$, the inverse of the map $x\mapsto b(x)$, such that
\begin{equation*}
  \sum_{A\in\mathcal{S}}
  \begin{tikzpicture}
    \draw[virtual] (0.5,0) rectangle (5.5,-0.5);
    \node[tensor,label=below:$X_A$] (x) at (1,0) {};
    \foreach \x/\t in {2/1,3/2,5/n}{
      \node[tensor,label=below:$A$] (t\x) at (\x,0) {};
      \draw[->-] (t\x) --++ (0,0.5);
    }
    \node[fill=white] (dots) at (4,0) {$\dots$};
    \draw[virtual,->-] (t2) -- (x);
    \draw[virtual,->-] (t3) -- (t2);
    \draw[virtual,->-] (dots) -- (t3);
    \draw[virtual,->-] (t5) -- (dots);
    \draw[virtual,->-] (0.5,-0.5) -- (5.5,-0.5);
  \end{tikzpicture}  \ =
  \phi^{\otimes n} \circ \Delta^{n-1} \circ B(\oplus_A X_A).
\end{equation*}

\subsection{Cocentral and non-degenerate elements} \label{sec:MPS_cocentral}

In this section we define the co-center of a coalgebra and show that in a cosemisimple coalgebra the set of cocentral elements have special MPS representations -- they are translation invariant MPS with periodic boundary condition.

Let $\mathcal{C}$ be a coalgebra with coproduct $\Delta$. Then the map $\Delta_{\mathrm{op}}:\mathcal{C}\to \mathcal{C}\otimes\mathcal{C}$, $x\mapsto \Delta_{\mathrm{op}}(x)$ defined by swapping the two components of the tensor product in $\Delta(x)$,
\begin{equation*}
  \Delta_{\mathrm{op}}(x) = \sum_{(x)} x_{(2)} \otimes x_{(1)},
\end{equation*}
is also a coproduct. Using the opposite coproduct, we can define cocentral elements as:
\begin{definition}
  An element $x\in\mathcal{C}$ is called \emph{cocentral} or \emph{trace-like} if it satisfies
  $\Delta_{\mathrm{op}}(x) = \Delta(x)$.
\end{definition}

Due to the definition of the product in $\mathcal{C}^*$, if $x\in \mathcal{C}$ is cocentral, then for all $f,g\in\mathcal{C}^*$,
\begin{equation*}
  (fg)(x) = (f\otimes g)\circ \Delta(x) = (g\otimes f)\circ \Delta_{\mathrm{op}}(x) = (g\otimes f)\circ \Delta(x) = (gf)(x).
\end{equation*}
This means thus that $x: f\mapsto f(x)$ is a trace-like (cyclic) linear functional on $\mathcal{C}^*$, i.e.\ the set of cocentral elements of $\mathcal{C}$ is exactly the set of trace-like linear functionals of $\mathcal{C}^*$. Due to their cyclicity, repeated coproducts of these elements are translation invariant:
\begin{equation*}
  \sum_{(x)} x_{(n)} \otimes x_{(1)} \otimes \dots \otimes x_{(n-1)} = (\id \otimes \Delta^{n-2}) \circ \Delta_{\mathrm{op}} (x) = (\id \otimes \Delta^{n-2}) \circ \Delta (x)= \Delta^{n-1} (x) = \sum_{(x)} x_{(1)} \otimes x_{(2)} \otimes \dots \otimes x_{(n)}.
\end{equation*}
The corresponding MPS representation is also translation invariant, i.e.\ $b(x)$ is such that
\begin{equation*}
  \ .
\end{equation*}
Let us assume now that $\mathcal{C}$ is cosemisimple. In this case, cocentral elements of $\mathcal{C}$ have very simple MPS representations: if $x\in\mathcal{C}$ is cocentral and $\phi$ is normal, then
\begin{equation*}
  \phi^{\otimes n }\circ\Delta^{n-1}(x) =
  \sum_{a\in \irr(\mathcal{C}^*)} \
  %
 \ ;
\end{equation*}
indeed, the boundary matrix $b(x) = \oplus_a b_a(x)$ is defined by $\tr(b(x) \psi(f)) = f(x)$ for all $x$, and as $x$ is trace-like, $\tr(b_a(x) \psi_a(f)\psi_a(g)) = \tr(\psi_a(f) b_a(x) \psi_a(g))$ for all $f,g\in\mathcal{C}^*$ and $a\in \irr(\mathcal{C}^*)$. As $\psi_a$ is an irrep, both $\psi_a(f)$ and $\psi_a(g)$ span the full matrix algebra and thus $b_a(x)$ is necessarily proportional to the identity.

Let $\tau_a$ denote the character of the irrep class $a$. As $\tau_a$ is a linear functional on $\mathcal{C}^*$, it can also be viewed as an element of the coalgebra $\mathcal{C}$.  By definition, $\tau_a(f) = \tr(\psi_a(f))$, and thus it can be written as $\tau_a(x) = \tr(b(\tau_a) \psi(f))$ with $b_b(\tau_a) = \delta_{ab} \cdot \id_a$. That is, the MPS representation of $\tau_a$ is
\begin{equation*}
  \phi^{\otimes n }\circ\Delta^{n-1}(\tau_a) =
  \begin{tikzpicture}
    \draw[virtual] (1.5,0) rectangle (5.5,-0.5);
    \foreach \x/\t in {2/1,3/2,5/n}{
      \node[tensor] (t\x) at (\x,0) {};
      \draw[->-] (t\x) --++ (0,0.5);
    }
    \node[fill=white] (dots) at (4,0) {$\dots$};
    \draw[virtual,->-] (t3) -- (t2) node[midway,black,irrep,above] {$a$};
    \draw[virtual,->-] (dots) -- (t3) node[midway,black,irrep,above] {$a$};
    \draw[virtual,->-] (t5) -- (dots) node[midway,black,irrep,above] {$a$};
    \draw[virtual,->-] (1.5,-0.5) -- (5.5,-0.5);
  \end{tikzpicture} \ ;
\end{equation*}

Another set of special elements in the coalgebra are those that have full-rank coproduct:

\begin{definition}\label{def:non-degenerate}
  An element $x\in\mathcal{C}$ is called \emph{non-degenerate} if its coproduct $\Delta(x)$ is full rank.
\end{definition}
Equivalently, $x$ is non-degenerate if and only if
  \begin{equation*}
    (\id\otimes \mathcal{C}^*)\circ\Delta(x) = (\mathcal{C}^*\otimes \id) \circ \Delta (x) = \mathcal{C},
  \end{equation*}
or, with other words, if for all $y\in\mathcal{A}$ there exist linear functionals $f$ and $g$ such that
\begin{equation*}
  y = (\id\otimes f)\circ \Delta(x) = (g\otimes \id) \circ \Delta(x).
\end{equation*}
Let us now show that if $\mathcal{C}$ is cosemisimple, then $x$ is non-degenerate if and only if its MPS representation is of the form
\begin{equation*}
  \phi^{\otimes n }\circ\Delta^{n-1}(x) =
  \sum_{a\in \irr(\mathcal{C}^*)} \
  \begin{tikzpicture}
    \draw[virtual] (0.5,0) rectangle (5.5,-0.5);
    \foreach \x/\t in {2/1,3/2,5/n}{
      \node[tensor] (t\x) at (\x,0) {};
      \draw[->-] (t\x) --++ (0,0.5);
    }
    \node[tensor,label=below:$b_a(x)$] (b) at (1,0) {};
    \node[fill=white] (dots) at (4,0) {$\dots$};
    \draw[virtual,->-] (t2) -- (b) node[midway,black,irrep,above] {$a$};
    \draw[virtual,->-] (t3) -- (t2) node[midway,black,irrep,above] {$a$};
    \draw[virtual,->-] (dots) -- (t3) node[midway,black,irrep,above] {$a$};
    \draw[virtual,->-] (t5) -- (dots) node[midway,black,irrep,above] {$a$};
    \draw[virtual,->-] (1.5,-0.5) -- (5.5,-0.5);
  \end{tikzpicture}  \ ,
\end{equation*}
where all $b_a(x)$ are invertible. To prove this, note first that for any linear functional $f\in \mathcal{C}^*$ the element $y = (f \otimes \id)\circ \Delta(x)$ is described by the MPS
\begin{equation*}
  \phi^{\otimes n }\circ\Delta^{n-1}(y) =
  \sum_{a\in \irr(\mathcal{C}^*)} \
  \begin{tikzpicture}
    \draw[virtual] (0.0,0) rectangle (5.5,-0.5);
    \foreach \x/\t in {2/1,3/2,5/n}{
      \node[tensor] (t\x) at (\x,0) {};
      \draw[->-] (t\x) --++ (0,0.5);
    }
    \node[tensor,label=below:$b_a(x)\cdot \psi_a(f)$] (b) at (1,0) {};
    \node[fill=white] (dots) at (4,0) {$\dots$};
    \draw[virtual,->-] (t2) -- (b) node[midway,black,irrep,above] {$a$};
    \draw[virtual,->-] (t3) -- (t2) node[midway,black,irrep,above] {$a$};
    \draw[virtual,->-] (dots) -- (t3) node[midway,black,irrep,above] {$a$};
    \draw[virtual,->-] (t5) -- (dots) node[midway,black,irrep,above] {$a$};
    \draw[virtual,->-] (1.5,-0.5) -- (5.5,-0.5);
  \end{tikzpicture}  \ ,
\end{equation*}
i.e.\ the boundary describing $y=(f \otimes \id)\circ \Delta(x)$ is given by $b_a(y) = b_a(x) \cdot \psi_a(f)$ in every sector $a\in\irr(\mathcal{C}^*)$. Similarly, the boundary describing $y = (\id\otimes g)\circ\Delta(x)$ is given by $b_a(y) = \psi_a(g) \cdot b_a(x)$. Therefore $x$ is non-degenerate if and only if for all $y\in\mathcal{C}$ there are $f,g\in\mathcal{C}^*$ such that
\begin{equation*}
  b_a(y) = \psi_a(g) \cdot b_a(x) = b_a(x) \cdot \psi_a(f).
\end{equation*}
As $\psi_a(f)$ can be any matrix, this is equivalent with the invertibility of $b_a(x)$.

As a particular case of the previous statement, let us consider a cocentral coalgebra element $x  = \sum_a \lambda_a \tau_a$, where $\tau_a\in\mathcal{C}$ are the irrep characters of $\mathcal{C}^*$. Then, $x$ is non-degenerate if and only if $\lambda_a \neq 0$ for all $a\in\irr(\mathcal{C}^*)$. For example, the cocentral element $\Theta = \sum_a \tau_a$ with MPS representation
\begin{equation*}
  \phi^{\otimes n }\circ\Delta^{n-1}(\Theta) =
  \sum_{a\in \irr(\mathcal{C}^*)} \
  \begin{tikzpicture}
    \draw[virtual] (1.5,0) rectangle (5.5,-0.5);
    \foreach \x/\t in {2/1,3/2,5/n}{
      \node[tensor] (t\x) at (\x,0) {};
      \draw[->-] (t\x) --++ (0,0.5);
    }
    \node[fill=white] (dots) at (4,0) {$\dots$};
    \draw[virtual,->-] (t3) -- (t2) node[midway,black,irrep,above] {$a$};
    \draw[virtual,->-] (dots) -- (t3) node[midway,black,irrep,above] {$a$};
    \draw[virtual,->-] (t5) -- (dots) node[midway,black,irrep,above] {$a$};
    \draw[virtual,->-] (1.5,-0.5) -- (5.5,-0.5);
  \end{tikzpicture}  \ ,
\end{equation*}
is a non-degenerate cocentral element. Using this element $\Theta$, one can interpret $b(x)$ (more precisely, $\psi^{-1}(b(x))$) as the linear functional that satisfies $(\psi^{-1}(b(x))\otimes \id)\circ\Delta(\Theta) = x$.

\section{Pre-bialgebras and their matrix product operator representations}\label{sec:bialgebra}

In this section we define pre-bialgebras (bialgebras without any property imposed on the unit and counit) as well as matrix product operators. We show that the MPS representations of coalgebras from the previous section naturally generalize to pre-bialgebras providing an MPO representation for them. We further investigate this MPO representation for cosemisimple pre-bialgebras.

\begin{definition}[Pre-bialgebra]
  $\mathcal{A}$ is a \emph{pre-bialgebra} if it is both an algebra and a coalgebra such that the coproduct $\Delta$ is multiplicative, i.e.\ for all $x,y\in\mathcal{A}$
  \begin{equation*}
    \Delta(xy) = \Delta(x) \Delta(y).
  \end{equation*}
\end{definition}
In this equation the multiplication in $\mathcal{A}\otimes \mathcal{A}$ is taken component-wise, i.e.\ $(x\otimes y ) \cdot (z\otimes w) = xz \otimes yw$. If $\mathcal{A}$ is a pre-bialgebra with product $\mu_{\mathcal{A}}$ and coproduct $\Delta_{\mathcal{A}}$, then $\mathcal{A}^*$ is also a pre-bialgebra with product $\mu_{\mathcal{A}^*} =   \Delta_{\mathcal{A}}^T$ and coproduct $\Delta_{\mathcal{A}^*} = \mu_{\mathcal{A}}^T$. It is clear from context whether we refer to the coproduct of $\mathcal{A}$ or that of $\mathcal{A}^*$, and thus in the following we drop the subscript $\mathcal{A}$ and $\mathcal{A}^*$, and simply write $\Delta$ for both $\Delta_{\mathcal{A}}$ and $\Delta_{\mathcal{A}^*}$. That is, the product and coproduct in $\mathcal{A}^*$ are such that for all $f,g\in\mathcal{A}^*$ and $x,y\in\mathcal{A}$,
\begin{equation*}
  (fg)(x) = (f\otimes g) \circ \Delta (x) \quad \text{and} \quad \left[\Delta(f)\right](x\otimes y) = f(xy).
\end{equation*}
The unit of $\mathcal{A}^*$ is the counit of $\mathcal{A}$ and the counit of $\mathcal{A}^*$ is the unit of $\mathcal{A}$.

In the following, we will talk about representations of the pre-bialgebra $\mathcal{A}$. These representations should be understood as representations of the algebraic structure of $\mathcal{A}$ (i.e.\ disregarding the coalgebra structure). The extra structure given by the coproduct allows us to define the tensor product of representations. Let $\phi_1:\mathcal{A}\to \End(V_1)$ and $\phi_2:\mathcal{A}\to \End(V_2)$ be two representations of $\mathcal{A}$. Then, as the coproduct $\Delta$ is multiplicative, $(\phi_1\otimes \phi_2) \circ \Delta: \mathcal{A} \to \End(V_1\otimes V_2)$ is also multiplicative. This map is not a representation on $V_1\otimes V_2$, however, since $(\phi_1\otimes \phi_2) \circ \Delta(1)$ is not the identity, unless $\Delta(1)= 1\otimes 1$. The element $(\phi_1\otimes \phi_2) \circ \Delta(1)$ is a projector and is absorbed by any element from $(\phi_1\otimes \phi_2) \circ \Delta(\mathcal{A})$, as $a1 = 1a = a$ for all $a\in\mathcal{A}$. This means that all operators in $(\phi_1\otimes \phi_2) \circ \Delta(\mathcal{A})$ can be restricted to the range of $(\phi_1\otimes \phi_2) \circ \Delta(1)$. Let $V_1\boxtimes V_2$ denote this subspace of $V_1\otimes V_2$. By definition, $(\phi_1\otimes \phi_2) \circ \Delta(1)$ restricted to $V_1\boxtimes V_2$, $(\phi_1\otimes \phi_2) \circ \Delta(1)|_{V_1\boxtimes V_2}$, is the identity, and thus the map defined by $(\phi_1 \boxtimes \phi_2)(x) = (\phi_1\otimes \phi_2) \circ \Delta(x)|_{V_1\boxtimes V_2}$ is a representation of $\mathcal{A}$.  This representation is then called the tensor product of the representations $\phi_1$ and $\phi_2$. Using associativity of the coproduct, one can define the $n$-fold tensor product of representations for any $n>2$ integer as well: given $\phi_i:\mathcal{A}\to \End(V_i)$ representation for $i=1\dots n$, the tensor product representation $\phi_1 \boxtimes \dots \boxtimes\phi_n$ is given by the restriction of $(\phi_1\otimes \dots \otimes\phi_n)\circ \Delta^{n-1}(x)$ onto the range of $(\phi_1\otimes \dots \otimes\phi_n)\circ \Delta^{n-1}(1)$ in the vector space $V_1 \otimes \dots \otimes V_n$.

Just as in the previous section, using the coalgebra structure of  $\mathcal{A}$, one can form MPS representations of $\mathcal{A}$. As $\mathcal{A}$ has an algebra structure as well, it is natural to choose the linear map $\phi$ used at the construction of the MPS to be a representation of the algebra. The resulting MPS is then interpreted as an operator, and in fact, this structure is called a matrix product operator (MPO):

\begin{definition}[Matrix product operators]
  Let $(V_i)_{i=1}^n$ and $(W_j)_{j=0}^n$ be two collections of finite dimensional vector spaces over $\mathbb{C}$. An MPO is given by tensors $A_i \in V_i \otimes V_i^* \otimes W_{i-1} \otimes W_i^*$ $(i=1, \ldots, n)$, $A_i = \sum_{kl} \ket{k}\bra{l} \otimes A_i^{kl}$ and a matrix $X\in W_n \otimes W_0^*$; the operator generated by the MPO is given by
  \begin{equation*}
    O = 	\sum_{kl} \tr \left( X \cdot A_1^{k_1 l_1} \cdots A_n^{k_n l_n} \right)\  \ket{k_1 \ldots k_n}\bra{l_1 \ldots l_n} = \
  \ .
  \end{equation*}
\end{definition}

Fixing a representation $\phi: \mathcal{A} \to \End(V)$ of $\mathcal{A}$ and an injective representation $\psi:\mathcal{A}^*\to \End(W)$ of $\mathcal{A}^*$, one can repeat the procedure described in the previous section to form MPOs representing the coalgebra structure of $\mathcal{A}$:
\begin{equation*}
  \phi^{\otimes n}\circ \Delta^{n-1} (x)  = \
  	  %
 =
    \sum_{x\in B} \phi(x) \otimes \psi(\delta_x),
  \end{equation*}
where $B$ is a basis of $\mathcal{A}$, and $\phi$ is a representation of $\mathcal{A}$, and the matrices $b(x)$ are such that $\tr\left( b(x) \psi(f) \right) = f(x)$. As described above, the map $\phi^{\otimes n}\circ \Delta^{n-1}$, in general, is not a representation of $\mathcal{A}$ as $1$ is not mapped to the identity operator on the space $V^{\otimes n}$:
\begin{equation*}
 \phi^{\otimes n} \circ \Delta^{n-1} (1) = \
 %
 \ \neq \id_V^{\otimes n}.
\end{equation*}
These MPOs, nevertheless, are multiplicative as $\phi^{\otimes n}\circ\Delta^{n-1}$ is multiplicative:
\begin{equation}\label{eq:mpo_product}
 %
 \ .
\end{equation}
In particular, these MPOs are invariant under multiplying with $\phi^{\otimes n} \circ \Delta^{n-1}(1)$ from either side,
\begin{equation*}
 %
 \ .
\end{equation*}
One can thus restrict $V^{\otimes n}$ to the range of $\phi^{\otimes n} \circ \Delta^{n-1}(1)$ (notice that this is a projector as $1\cdot 1 = 1$), and on this space the MPOs form a representation of $\mathcal{A}$. Note that this restriction is only necessary for $n>1$; for $n=1$ the representing MPOs are simply $x\mapsto \phi(x)$:
\begin{equation*}
  \phi (x)  =
  %
 \ .
\end{equation*}
In particular, on a single site, unlike for $n>1$ sites, the unit is represented by the identity operator:
\begin{equation*}
  \id = \phi (1)  =
  %
 \ .
\end{equation*}

\subsection{Cosemisimplicity}

Let $\mathcal{A}$ be a pre-bialgebra, then $\mathcal{A}^*$ is a pre-bialgebra as well.  As such, one can form tensor products of its representations, i.e.\ if $\psi_1:\mathcal{A}^*\to \End(W_1)$ and $\psi_2:\mathcal{A}^*\to \End(W_2)$ are representations of $\mathcal{A}^*$, then $(\psi_1\otimes\psi_2)\circ \Delta_{\mathcal{A}^*}$ restricted to the range of $(\psi_1\otimes\psi_2)\circ \Delta_{\mathcal{A}^*}(\epsilon)$ is a representation\footnote{Remember that the unit of $\mathcal{A}^*$ is $\epsilon$} of $\mathcal{A}^*$ denoted by $\psi_1\boxtimes \psi_2$.

Just as for coalgebras, we say that a pre-bialgebra $\mathcal{A}$ is cosemisimple if the algebra $\mathcal{A}^*$ is semisimple, i.e.\ if every representation of it decomposes into a direct sum of irreps. In particular, let us fix irreps $\psi_a:\mathcal{A}^*\to \End(W_a)$ for each irrep class $a\in\irr(\mathcal{A}^*)$; then given two of these irreps, $\psi_a$ and $\psi_b$, their tensor product $\psi_a\boxtimes \psi_b$ decomposes into a direct sum of irreps as follows:
\begin{equation*}
  \psi_a\boxtimes \psi_b \simeq \bigoplus_{c\in \irr(\mathcal{A}^*)} \psi_c \otimes \id_{N_{ab}^c},
\end{equation*}
i.e.\ in the decomposition of $\psi_a\boxtimes \psi_b$ the irrep $\psi_c$ appears $N_{ab}^c$ times. These $N_{ab}^c$  non-negative integers are called \emph{fusion multiplicities}. Note that $N_{ab}^c$ might be $0$; in that case the irrep $\psi_c$ does not appear in the decomposition. The above equation holds up to a basis transformation, i.e.\ for all $a,b,c\in \irr(\mathcal{A}^*)$ there are  invertible operators $Z_{ab}: \bigoplus_c W_c \otimes \mathbb{C}^{N_{ab}^c} \to W_a \boxtimes W_b$ such that for all $f\in\mathcal{A}^*$
\begin{equation}\label{eq:irrep_product_decomposition}
  \left(\psi_a\boxtimes \psi_b\right) (f) =  Z_{ab} \left( \bigoplus_{c\in \irr(\mathcal{A}^*)} \psi_c(f) \otimes \id_{N_{ab}^c}\right) \left(Z_{ab}\right)^{-1},
\end{equation}
or equivalently, for all $a,b,c\in \irr(\mathcal{A}^*)$ there are $N_{ab}^c$ operators $Z_{ab}^{c\mu}: V_c \to V_{a}\boxtimes V_b$ and $Y_{ab}^{c\mu}: V_a\boxtimes V_b\to V_c$ (here $\mu=1\dots N_{ab}^c$ are integers) such that
\begin{equation*}
  \left(\psi_a\boxtimes \psi_b\right) (f) = \sum_{c\in \irr(\mathcal{A}^*)} \sum_{\mu=1}^{N_{ab}^c} Z_{ab}^{c\mu} \cdot \psi_c(f) \cdot Y_{ab}^{c\mu} \quad \text{and} \quad Y_{ab}^{c\mu} \cdot Z_{ab}^{d\nu} = \delta_{cd}\cdot \delta_{\mu\nu} \cdot \id_{V_c}.
\end{equation*}
These operators then can be extended\footnote{Note that as $\epsilon$ is a projector, $V_a\otimes V_b = \mathrm{Ker}\{(\psi_a \boxtimes \psi_b)(\epsilon)\} \oplus \mathrm{Im}(\psi_{a}\boxtimes\psi_{b})(\epsilon)$. The extension of the operators $Y_{ab}^{c\mu}$ is such that they act as zero on the space $\mathrm{Ker}\{(\psi_a \boxtimes \psi_b)(\epsilon)\}$.} to map from (and to) $W_a \otimes W_b$ instead of $W_{a}\boxtimes W_b$, i.e.\ there are linear maps $V_{ab}^{c\mu}: W_c \to W_{a}\otimes W_b$ and $W_{ab}^{c\mu}: W_a \otimes W_b\to W_c$, called \emph{fusion tensors}, such that
\begin{equation}\label{eq:splitting_algebraic}
  (\psi_a \otimes \psi_b)\circ\Delta(f) = \sum_{c} \sum_{\mu=1}^{N_{ab}^c} V_{ab}^{c\mu} \cdot \psi_c(f) \cdot W_{ab}^{c\mu} \quad \text{and} \quad W_{ab}^{c\mu} \cdot V_{ab}^{d\nu} = \delta_{cd} \cdot\delta_{\mu\nu} \cdot \id_{W_c},
\end{equation}
for all $a,b,c\in \irr(\mathcal{A}^*)$ and  $f\in\mathcal{A}^*$. The fusion tensors are rank-four tensors; the index $\mu$, however, plays a very different role than the rest of its indices. This allows us to think of the fusion tensors as a set for rank-three tensors instead, and write
\begin{equation*}
  V_{ab}^{c\mu} =
 \ .
\end{equation*}
Using this notation, the graphical representation of \cref{eq:splitting_algebraic} is the following:
\begin{equation}\label{eq:splitting}
  (\psi_a \otimes \psi_b) \circ \Delta(f) =
  \sum_{c\mu}
  %
 = \delta_{cd} \cdot \delta_{\mu\nu} \cdot \id_{c} \ .
\end{equation}
Let us stress again that the fusion tensors $V_{ab}^{c\mu}$ (and $W_{ab}^{c\mu}$) do not map to (and from) the whole space $W_a\otimes W_b$, instead only to (and from) the subspace $W_a\boxtimes W_b$. A projector onto this subspace is $(\psi_a\otimes\psi_{b})\circ\Delta(\epsilon) \neq \id_a\otimes \id_b$, that has the graphical representation
\begin{equation}\label{eq:delta_epsilon}
  (\psi_a \otimes \psi_b) \circ \Delta(\epsilon) =
  \sum_{c\mu}
  \tikzsetnextfilename{90905da3-7b87-4642-934e-14bd68f2d1b2}
  \begin{tikzpicture}[baseline=-1mm]
    \node[fusion tensor] (v) at (0,0) {};
      \node[anchor=south] at (v.north) {$\mu$};
    \node[fusion tensor] (w) at (0.5,0) {};
      \node[anchor=south] at (w.north) {$\mu$};
    \draw[virtual,-<-] (v)--(w) node[midway,above,black,irrep] {$c$};
    \draw[virtual,->-] let \p1=(v.west) in (\x1,0.25)--++(-0.5,0) node[midway,above,black,irrep] {$a$};
    \draw[virtual,->-] let \p1=(v.west) in (\x1,-0.25)--++(-0.5,0) node[midway,above,black,irrep] {$b$};
    \draw[virtual,-<-] let \p1=(w.east) in (\x1,0.25)--++(0.5,0) node[midway,above,black,irrep] {$a$};
    \draw[virtual,-<-] let \p1=(w.east) in (\x1,-0.25)--++(0.5,0) node[midway,above,black,irrep] {$b$};
  \end{tikzpicture} \ .
\end{equation}

The operators $V_{ab}^{c\mu}$ and $W_{ab}^{c\mu}$ are not unique, instead only their tensor product $\sum_\mu V_{ab}^{c\mu} \otimes W_{ab}^{c\mu}$ is fixed. This allows for a basis change: $\sum_\mu V_{ab}^{c\mu} \otimes W_{ab}^{c\mu} = \sum_{\mu} \hat{V}_{ab}^{c\mu} \otimes \hat{W}_{ab}^{c\mu}$ if and only if $\hat{V}_{ab}^{c\mu} = \sum_\nu K_{\mu\nu} V_{ab}^{c\nu}$ and $\hat{W}_{ab}^{c\mu} = \sum_\nu (K^{-1})_{\nu\mu} W_{ab}^{c\nu}$ for some $N_{ab}^c \times N_{ab}^c$ complex invertible  matrix $K$. Note as well that the value of $\sum_\mu V_{ab}^{c\mu} \otimes W_{ab}^{c\mu}$ depends on the choices of the irrep representatives $\psi_a$ that we have made at the beginning of this section; different choices will lead to a gauge transformation of $V_{ab}^{c\mu}$ and $W_{ab}^{c\mu}$, i.e.\ if the irrep representatives are chosen to be $\hat{\psi}_a = U_a \psi_a U_a^{-1}$ instead of $\psi_a$ for each $a\in \irr(\mathcal{A}^*)$ (here $U_a$ are general invertible matrices), then the corresponding $\hat{V}_{ab}^{c\mu}$ and $\hat{W}_{ab}^{c\mu}$ are given by
\begin{equation*}
  \hat{V}_{ab}^{c\mu} = (U_a \otimes U_b) \cdot V_{ab}^{c\mu} \cdot U_c^{-1}
  \quad \text{and} \quad
  \hat{W}_{ab}^{c\mu} = U_c \cdot W_{ab}^{c\mu}  \cdot (U_a^{-1} \otimes U_b^{-1}).
\end{equation*}

Let us now show using this graphical language how the so-called $F$-symbols of the representation category of $\mathcal{A}^*$ emerge. Using \cref{eq:splitting}, associativity of the coproduct of $\mathcal{A}^*$ implies that for all $f\in\mathcal{A}^*$ and $a,b,c,e\in \irr(\mathcal{A}^*)$,
\begin{equation*}
  \sum_{d\mu\nu}
  \tikzsetnextfilename{3fb83239-fee9-499e-9d4f-9ab56f90b882}
 \ .
\end{equation*}
As this equation holds for all $f\in\mathcal{A}^*$, and $\psi_e(\mathcal{A}^*)=\mathcal{M}_{D_e}$ for each irrep block $e$, i.e.\ the matrix $\psi_e(f)$ can take any value, we conclude that
\begin{equation}\label{eq:associative}
  \sum_{d\mu\nu}
  \tikzsetnextfilename{be340654-ec24-40b7-a79b-0b1b2bc056e0}
  %
 \ ,
\end{equation}
or equivalently, that there is an invertible (square) matrix $F^{abc}_e$ of dimension $\sum_{d} N_{ab}^d N_{dc}^e = \sum_d N_{bc}^d N_{ad}^e$ such that
\begin{equation*}
  \tikzsetnextfilename{56cf2d11-42a5-4cf2-9745-8e09e55251a0}
  %
 \ .
\end{equation*}
Standard arguments show that these $F$-symbols satisfy a consistency condition called the pentagon equation.

Let us now investigate the MPO representations of a cosemisimple pre-bialgebra $\mathcal{A}$. Just as in the MPS case, the MPOs representing $\mathcal{A}$ decompose into a sum of MPOs with smaller bond dimension  corresponding to the irreps of $\mathcal{A}^*$. We denote these MPOs by writing the label $a\in \irr(\mathcal{A}^*)$ on the virtual indices of the tensor:
\begin{equation*}
  \phi^{\otimes n} \circ \Delta^{n-1} (x)  = \sum_{a\in \irr(\mathcal{A}^*)}
      \tikzsetnextfilename{4e3993a8-1c3f-4491-9515-97fe2d99d1ed}
  	  \begin{tikzpicture}
  	    \draw[virtual] (0.5,0) rectangle (5.5,-0.7);
       \node[tensor,label=below:$b_a(x)$] (x) at (1,0) {};
  	    \foreach \x/\t in {2/1,3/2,5/n}{
  	      \node[tensor] (t\x) at (\x,0) {};
  	      \draw[->-] (t\x) --++ (0,0.5);
  	      \draw[-<-] (t\x) --++ (0,-0.5);
  	    }
  	    \node[fill=white] (dots) at (4,0) {$\dots$};
        \draw[virtual,->-] (t2) -- (x) node[midway,black,irrep,above] {$a$};
        \draw[virtual,->-] (t3) -- (t2) node[midway,black,irrep,above] {$a$};
        \draw[virtual,->-] (dots) -- (t3) node[midway,black,irrep,above] {$a$};
        \draw[virtual,->-] (t5) -- (dots) node[midway,black,irrep,above] {$a$};
        \draw[virtual,->-] (0.5,-0.7) -- (5.5,-0.7);
  	  \end{tikzpicture}  \quad \text{with} \quad
      \begin{tikzpicture}
        \node[tensor] (t) at (0,0) {};
        \draw[virtual, -<-] (t) --++ (0.7,0) node[midway,black,irrep,above] {$a$};
        \draw[virtual, ->-] (t) --++ (-0.7,0) node[midway,black,irrep,above] {$a$};
        \draw[->-] (t) --++ (0,0.5);
        \draw[-<-] (t) --++ (0,-0.5);
      \end{tikzpicture} =
      \sum_{x\in B} \phi(x) \otimes \psi_a(\delta_x).
\end{equation*}
If $\phi$ is injective, then each smaller MPO tensor is injective and if there is a positive integer $n$ such that $\phi^{\otimes n} \circ \Delta^{n-1}$ is injective, then they are normal. The above decomposition is thus nothing but the decomposition of the MPO into its normal components.

Let us now consider the product of two MPOs each representing an element of $\mathcal{A}$, such as in \cref{eq:mpo_product}. The l.h.s.\ of this equation then decomposes to injective blocks as described above. On the r.h.s.\ both MPOs can be decomposed into injective blocks, therefore their product can also be decomposed into a sum:
\begin{equation}\label{eq:MPO_product_blocks}
  \ .
\end{equation}
It turns out, however, that these MPOs are still not injective, i.e.\ they can be further decomposed. To understand why, let us investigate the tensor describing such an MPO:
  \begin{equation*}
    \tikzsetnextfilename{444084af-d010-40de-8f6b-51d890aa6f01}
    %
 =
    \sum_{x,y\in B} \phi(x)\cdot \phi(y) \otimes (\psi_a(\delta_x) \otimes \psi_b(\delta_y)) =
    \sum_{x\in B} \phi(x) \otimes (\psi_a\otimes \psi_b) \circ \Delta(\delta_x) ,
  \end{equation*}
where the last equality, analogous to the proof of \cref{thm:coalg_to_mps}, holds because for every $f\in\mathcal{A}^*$,
\begin{equation*}
  \sum_{x,y\in B} f(xy) \cdot \delta_x \otimes \delta_y = \sum_{x,y\in B, (f)} f_{(1)}(x) \cdot f_{(2)}(y) \cdot \delta_x\otimes \delta_y  = \sum_{(f)} f_{(1)} \otimes f_{(2)}  = \sum_{x\in B} f(x) \cdot \Delta(\delta_x),
\end{equation*}
and thus $\sum_{x,y\in B} xy \otimes \delta_x \otimes \delta_y = \sum_{x\in B} x\otimes \Delta(\delta_x)$.

This form  of the MPO tensor involves thus the tensor product of the irreps $\psi_a$ and $\psi_b$. In \cref{eq:splitting} we have seen how to decompose such a representation into irreps. Using that equation, the product of the MPO tensors satisfy
  \begin{equation}\label{eq:MPO_tensor_product}
    \tikzsetnextfilename{8d3b82ce-416a-42a6-9339-608dbf3db7fb}
 = \delta_{cd} \cdot \delta_{\mu\nu} \cdot \id_{c} \ .
  \end{equation}
Using \cref{eq:MPO_tensor_product} in \cref{eq:MPO_product_blocks}, we can finally decompose the r.h.s.\ of \cref{eq:mpo_product} into the sum of injective MPOs:
\begin{equation}\label{eq:MPO_product_result}
  \sum_{a,b\in \irr(\mathcal{A}^*)}
  %
  \ .
\end{equation}
We have thus obtained that the product of two algebra elements $x$ and $y$ is described by the boundary
\begin{equation}\label{eq:product_boundary}
   b_{c}(xy) = \sum_{ab\mu}
   %
  \ .
\end{equation}

Let us note here that so far we have not investigated how the unit of the pre-bialgebra is represented. The existence of the unit in $\mathcal{A}$, in fact, imposes further restrictions on the structure of the MPOs representing pre-bialgebras. Remember that $\phi^{\otimes n}\circ\Delta^{n-1}(1)$ is a non-trivial projector that is represented by:
\begin{equation*}
  \phi^{\otimes n} \circ \Delta^{n-1} (1)  = \sum_{a\in \irr(\mathcal{A}^*)}
  	  %
 \neq \id^{\otimes n}.
\end{equation*}
The relation $1x = x1 = x$, using \cref{eq:product_boundary}, then implies that the matrices $b_a(1)$ have the following property:
\begin{equation}
   b_{c}(x) = \sum_{ab\mu}
   %
  \ .
\end{equation}

\subsection{Cocentral elements and the Grothendieck ring}

Let us consider a cosemisimple pre-bialgebra $\mathcal{A}$. Then the cocenter of $\mathcal{A}$, as we have seen in \cref{sec:MPS_cocentral}, consists of elements $x$ of the form
\begin{equation*}
  \phi^{\otimes n} \circ \Delta^{n-1} (x)  = \sum_{a\in \irr(\mathcal{A}^*)} \lambda_a \cdot
      \tikzsetnextfilename{e2b54ea6-bd1b-49b1-a8cf-e542c2fea473}
  	  %
 \ .
\end{equation*}
Let $\tau_a\in\mathcal{A}$ be the character of the irrep class $a\in \irr(\mathcal{A}^*)$, i.e.\ the cocentral element with MPO representation
\begin{equation*}
  \phi^{\otimes n} \circ \Delta^{n-1} (\tau_a)  = \
      \tikzsetnextfilename{678f2fd1-c392-49e2-bcd7-e83aff322062}
  	  %
 \ .
\end{equation*}
Let us evaluate the product of two irrep characters $\tau_a$ and $\tau_b$. Using the previously derived results, it is immediate to see from their MPO representations that
\begin{equation*}
  %
 \ ,
\end{equation*}
where the first equation is \cref{eq:MPO_product_result} using that $b_a(\tau_c) = \delta_{ac} \id_c$ for all $a,c\in \irr(\mathcal{A}^*)$, and the second is just the orthogonality relations from \cref{eq:MPO_tensor_product} together with the fact that $\mu$ runs from 1 to $N_{ab}^c$. This equation then reads as
\begin{equation*}
  \phi^{\otimes n} \circ \Delta^{n-1} (\tau_a) \cdot \phi^{\otimes n} \circ \Delta^{n-1} (\tau_a) =  \sum_c N_{ab}^c \cdot \phi^{\otimes n} \circ \Delta^{n-1} (\tau_c),
\end{equation*}
and thus it implies, as $\phi^{\otimes n} \circ \Delta^{n-1}$ is a homomorphism and it is w.l.o.g.\ injective, that
\begin{equation}\label{eq:character_multiplication}
  \tau_a \cdot \tau_b = \sum_c N_{ab}^c \cdot \tau_c.
\end{equation}
Let us remark here that to obtain this well-known result, one does not have to consider MPO representations. Instead, the same result can be obtained directly from the decomposition of the tensor product representation into irreps: let $\psi_a$ be an irrep on a vector space $V_a$ such that it is from the irrep class $a$ (and thus its character is $\tau_a$), and $\psi_b$  an irrep on a vector space $V_b$ such that it is from the irrep class $b$ (and thus its character is $\tau_b$). Then, by definition of the product in $\mathcal{A}^*$,
\begin{equation*}
  (\tau_a \cdot \tau_b)(f) = (\tau_a\otimes \tau_b) \circ \Delta (f) = (\tr_{V_a} \otimes \tr_{V_b}) \circ (\psi_a\otimes \psi_b) \circ \Delta (f) = \tr \left( (\psi_a\otimes \psi_b) \circ \Delta (f) \right).
\end{equation*}
In this last trace the operator $(\psi_a\otimes \psi_b) \circ \Delta(f)$ is supported on $V_a \boxtimes V_b$ instead of the whole tensor product space $V_a\otimes V_b$, and thus restricting $(\psi_a\otimes \psi_b) \circ \Delta(f)$ to $V_a \boxtimes V_b$ does not change its trace:
\begin{equation*}
(\tau_a \cdot \tau_b)(f) = \tr  (\psi_a\otimes \psi_b) \circ \Delta (f) = \tr  (\psi_a\boxtimes \psi_b)(f).
\end{equation*}
Finally, this representation decomposes into irreps (see \cref{eq:irrep_product_decomposition}), and thus the trace can be evaluated:
\begin{equation*}
  (\tau_a \cdot \tau_b)(f) = \tr  (\psi_a\boxtimes \psi_b)(f) = \sum_c N_{ab}^c \cdot \tr \psi_c (f) = \sum_c N_{ab}^c \cdot \tau_c (f),
\end{equation*}
that is equivalent to \cref{eq:character_multiplication}.
We have thus seen that
\begin{proposition}
  In a finite dimensional cosemisimple pre-bialgebra $\mathcal{A}$ over $\mathbb{C}$ the irrep characters of $\mathcal{A}^*$ correspond to the injective blocks in the MPO representation of $\mathcal{A}$. For $a\in \irr(\mathcal{A}^*)$, the irrep character $\tau_a\in\mathcal{A}$ has the following MPO representation:
  \begin{equation*}
    \phi^{\otimes n} \circ \Delta^{n-1} (\tau_a)  = \
        \tikzsetnextfilename{678f2fd1-c392-49e2-bcd7-e83aff322062}
 \ .
  \end{equation*}
  These MPOs form a closed ring\footnote{Not necessarily unital; note that the unit of the algebra  $1\in\mathcal{A}$ might have a non-trivial boundary, and thus, in general, it is not in the Grothendieck ring} over $\mathbb{Z}$, i.e.\ for all $a,b,c$ there are non-negative integer numbers $N_{ab}^c$ such that $\tau_a \tau_b = \sum_c N_{ab}^c \tau_c$, or graphically,
\begin{equation*}
  \tikzsetnextfilename{21ff2474-e24a-49d6-a222-10afbbea700b}
  %
 \ .
\end{equation*}
This ring is then called the Grothendieck ring of $\mathcal{A}$.
\end{proposition}

\subsection{Duality}

As we have seen before, the dual $\mathcal{A}^*$ of a pre-bialgebra $\mathcal{A}$ is a pre-bialgebra as well. As such, it also has MPO representations. Let us fix a representation $\psi$ of $\mathcal{A}^*$ on a vector space $W$, and a representation $\phi$ of $\mathcal{A}^{**} = \mathcal{A}$ on a vector space $V$. Then the previous construction leads to the following MPO representation:
\begin{equation}\label{eq:dual_mpo}
  \psi^{\otimes n} \circ \Delta^{n-1} (f)  =
  \tikzsetnextfilename{0ceb3e17-d811-4514-ba57-97948fe43443}
 =
  \sum_{x\in B} \psi(\delta_x) \otimes \phi(x),
\end{equation}
where $B$ is a basis of $\mathcal{A}$. Notice that this MPO tensor is exactly 90 degree rotation of the MPO tensor describing the coproduct of elements in $\mathcal{A}$. As above, these MPOs form a representation of $\mathcal{A}^*$,
\begin{equation*}
 \tikzsetnextfilename{917b011c-47b0-4985-a57b-8c0a556e50a0}
 %
  \ .
\end{equation*}
If $\mathcal{A}$ is semisimple (in general this does not follow from semisimplicity of $\mathcal{A}^*$), then this MPO representation of $\mathcal{A}^*$ decomposes into a sum of smaller bond dimensional MPOs corresponding to the characters of $\mathcal{A}$:
\begin{equation*}
  \psi^{\otimes n} \circ \Delta^{n-1} (f)  = \sum_{\alpha\in \irr(\mathcal{A})}
  \tikzsetnextfilename{c6fb2570-cada-457c-bb6c-43d21bb2d022}
  %
 =
  \sum_{x\in B} \psi(\delta_x) \otimes \phi_\alpha(x),
\end{equation*}
where $B$, as above, is a basis of $\mathcal{A}$. We will denote irrep classes of $\mathcal{A}$ with Greek letters, and irrep classes of $\mathcal{A}^*$ with Latin letters.  Just as above, there are linear maps $\hat{V}_{\alpha\beta}^{\gamma m}$ and $\hat{W}_{\alpha \beta}^{\gamma m}$ ($m=1\dots \hat{N}_{\alpha\beta}^\gamma$), such that
  \begin{equation*}
    \tikzsetnextfilename{19ce0d3a-0242-4116-8099-f5874753326c}
    %
 = \delta_{\gamma\delta} \cdot \delta_{mn} \cdot \id_{\gamma} \ .
  \end{equation*}
These linear maps give rise to a set of $F$-symbols (also satisfying the pentagon equations) that are, in general, different from the ones in the previous sections.

\section{Weak bialgebras and weak Hopf algebras}\label{sec:WBAWHA}

In this section we introduce weak bialgebras, weak Hopf algebras, pivotal weak Hopf algebras, spherical weak Hopf algebras and $C^*$-weak Hopf algebras. All these structures are all special pre-bialgebras, and as such, we use their MPO representations to reason about the structure of these objects; this is possible as the MPO representations are w.l.o.g.\ injective. Our main result is \cref{thm:wha_special_integral}, where we construct a special normalized integral $\Lambda$ in any cosemisimple weak Hopf algebra over $\mathbb{C}$. This integral $\Lambda$ is then used to prove that in a cosemisimple co-pivotal weak Hopf algebra over $\mathbb{C}$, there is a cocommutative projector with a property reminiscent to the definition of an integral that will allow us to define MPO-injective PEPS. We then further specialize these results to spherical- and $C^*$-weak Hopf algebras.

\subsection{Weak bialgebras}

In this section we define weak bialgebras (WBA) and show that the MPO representation of the unit of a cosemisimple WBA $\mathcal{A}$ has the following property: for all irreps $a$ of $\mathcal{A}^*$, $b_a(1)$ is either $0$ or rank-one. We show, moreover, that the Grothendieck ring of $\mathcal{A}$ has a unit. This unit can be written in the form $\tau_E:= \sum_{e\in E} \tau_e$, where $E\subseteq \irr(\mathcal{A}^*)$ consists of the irreps $a$ for which  $b_a(1)\neq 0$.

\begin{definition}[Weak bialgebra]
  A weak bialgebra (WBA) is a pre-bialgebra $\mathcal{A}$ such that the unit $1\in \mathcal{A}$ and counit $\epsilon\in \mathcal{A}^*$ satisfy
    \begin{align}
         &\sum_{(1)} 1_{(1)} \otimes 1_{(2)} \otimes 1_{(3)} = \sum_{(1)(1')} 1_{(1)} \otimes 1_{(2)} 1'_{(1)} \otimes 1'_{(2)} = \sum_{(1)(1')} 1_{(1)} \otimes  1'_{(1)} 1_{(2)} \otimes 1'_{(2)}, \label{eq:unit_axiom} \\
          & \sum_{(\epsilon)} \epsilon_{(1)} \otimes \epsilon_{(2)} \otimes \epsilon_{(3)} = \sum_{(\epsilon)(\epsilon')} \epsilon_{(1)} \otimes \epsilon_{(2)} \epsilon'_{(1)} \otimes \epsilon'_{(2)} = \sum_{(\epsilon)(\epsilon')} \epsilon_{(1)} \otimes  \epsilon'_{(1)} \epsilon_{(2)} \otimes \epsilon'_{(2)}. \label{eq:counit_axiom}
    \end{align}
    We will refer to \cref{eq:unit_axiom} as the unit axiom and to \cref{eq:counit_axiom} as the counit axiom.
\end{definition}
In the equations above $1$ and $\epsilon$ appears twice in the same Sweedler notation. To distinguish between the two appearances, we added a prime to one of them. One can also write the unit axiom, \cref{eq:unit_axiom}, as
\begin{equation*}
  \Delta^2(1) = (1\otimes \Delta(1)) \cdot (\Delta(1) \otimes 1) = (\Delta(1) \otimes 1) \cdot (1\otimes \Delta(1)),
\end{equation*}
while the counit axiom, \cref{eq:counit_axiom}, is more convenient to think of as
\begin{equation*}
  \epsilon(xyz) = \sum_{(y)} \epsilon(xy_{(1)}) \epsilon(y_{(2)}z) = \sum_{(y)} \epsilon(xy_{(2)}) \epsilon(y_{(1)}z).
\end{equation*}
As pre-bialgebras are self-dual and the above two axioms are the dual of one another, weak bialgebras are also self-dual: if $\mathcal{A}$ is a WBA, then $\mathcal{A}^*$ is also a WBA.

Let $\mathcal{A}$ be a cosemisimple WBA. The graphical representation of the unit axiom, \cref{eq:unit_axiom}, is
\begin{equation*}
  \sum_c \
 \ .
\end{equation*}
Decomposing the product of the two MPO tensors using the fusion tensors (\cref{eq:MPO_tensor_product}), we arrive at
  \begin{equation*}
		\sum_c \
		\tikzsetnextfilename{6310f334-b5db-4327-9b29-c2fcab1be723}
	  %
 \ .
  \end{equation*}
Applying linear functionals $f,g$ and $h$ on the three components, we arrive at the equations
  \begin{equation*}
		\sum_c \
    \tikzsetnextfilename{e20ae762-4021-4916-8db2-d508bccfb756}
	  %
 \ ,
  \end{equation*}
where $f_c$ denotes the part of $f$ supported on the irrep sector $c$, i.e.\ $f_c = p_c f = f p_c$, where $p_c\in Z(\mathcal{A}^*)$ is the irrep projector onto the irrep sector $c$. This equation is true for all $f,g$ and $h$, and thus, since $\psi_c(\mathcal{A}^*)$ is the full matrix algebra $\mathcal{M}_{D_c}$ for all irrep sectors,  we obtain
  \begin{align}
    \delta_{a,c} \delta_{b,c} \
    \tikzsetnextfilename{40357cef-4645-4936-a629-955e6a794e30}
    %
 \ . \label{eq:unit_rank_1_intermediate_2}
  \end{align}
  Combining these two equations, we arrive at
  \begin{equation*} \sum_\mu \
    \tikzsetnextfilename{d6e74d24-af7a-43b1-b527-25064cd458a5}
    %
 \ .
  \end{equation*}
  Comparing the last two expressions, we obtain that $b_c(1)$ is rank-1 for all $c$ where it is non-zero, i.e.\  there is a vector $\tikz{\draw[virtual,->-] (0,0) node[t,black] {} --++(-0.5,0) node[midway,black,irrep] {$c$};} \in V_c$ and a linear functional $\tikz{\draw[virtual,-<-] (0,0) node[t,black] {} --++(0.5,0) node[midway,black,irrep] {$c$};} \in V_c^*$ such that
  \begin{equation}\label{eq:unit_rank_1}
    b_c(1) =
    %
 \ .
  \end{equation}
  Let $E$ be the set of irreps such that $b_c(1)\neq 0$. Let us now consider \cref{eq:unit_rank_1_intermediate_1,eq:unit_rank_1_intermediate_2} such that the irreps $a$ and $b$ are from this set $E$, while $c$ is an arbitrary irrep. Using \cref{eq:unit_rank_1}, we obtain
	\begin{align}\label{eq:unit_axiom_intermediate}
	 \sum_{\mu }
    \tikzsetnextfilename{65cd82f0-7cd9-4d7e-8196-daf11cf9e54b}
    %
 \ =  \delta_{a,c} \cdot \delta_{b,c} \cdot (\id_c\otimes \id_c) .
  \end{align}
  In particular, $N_{ab}^c\neq 0$ for $a,b\in E$ and $c\in \irr(\mathcal{A}^*)$ if and only if $a=b=c$. As $b(1)$ is rank-one in every sector where it is non-zero, the MPO representation of $1\in\mathcal{A}$ can be written as
  \begin{equation}\label{eq:unit_mpo}
  \phi^{\otimes n} \circ \Delta^{n-1} (1) =  \sum_{e\in E} \
  %
  \  .
  \end{equation}
  Note that the fact that $1\in\mathcal{A}$ is the unit of the algebra implies, using \cref{eq:product_boundary_unit}, that for all $a,c\in \irr{\mathcal{A}^*}$,
	\begin{align}\label{eq:unit_boundary_fusion_intermediate}
    \sum_{e\in E,\mu}
	  %
 \ =  \delta_{a,c} \cdot (\id_c\otimes \id_c) .
	\end{align}
  Let us now consider the element $\sum_{(1)} 1_{(1)}1_{(2)}$. The MPS representation of this element is given by
  \begin{equation*}
    \sum_{(1)} \phi(1_{(1)}1_{(2)}) = \sum_{(1)} \phi(1_{(1)}) \cdot \phi(1_{(2)}) = \sum_{e\in E}
    \tikzsetnextfilename{58625754-2a8d-4d7c-8198-3e0d0e17a833}
    %
    = \sum_e \phi(\tau_e),
  \end{equation*}
  where in the fourth equality we have used \cref{eq:unit_axiom_intermediate}. As $\phi$ is injective w.l.o.g.,  this implies that $\sum_{(1)} 1_{(1)}1_{(2)} = \sum_{e\in E} \tau_e$. Let us denote this element of the algebra as $\tau_E$. Notice now that $\Delta(1) \cdot \Delta_{\mathrm{op}}(1)$ is represented by the product
  \begin{equation*}
    \sum_{ef\in E}
    %
  \ ,
  \end{equation*}
  where in the last equation we have used \cref{eq:unit_axiom_intermediate}. Again, as $\phi\otimes\phi$ is w.l.o.g.\ injective, we conclude that $\Delta(1) \cdot \Delta_{\mathrm{op}}(1) = \Delta(\tau_E)$. Similar calculation shows that $\Delta_{\mathrm{op}}(1) \cdot \Delta(1) = \Delta(\tau_E)$ as well, and thus we have proven\footnote{The same result can be obtained by direct calculation as well \cite{Nikshych2003}: using the unit axiom \cref{eq:unit_axiom} twice, $\Delta(\tau_E) = \sum_{(1)}1_{(1)}1_{(3)} \otimes 1_{(2)}1_{(4)} =\sum_{(1)} \Delta(1) \cdot (1_{(2)} \otimes 1_{(1)} 1_{(3)}) = \Delta(1) \cdot \Delta_{\mathrm{op}}(1) \cdot \Delta(1)$. Finally, again due to the unit axiom, $\Delta(1)\cdot \Delta_{\mathrm{op}}(1) = \Delta_{\mathrm{op}}(1)\cdot \Delta(1)$, that leads to the desired result.
  } that
  \begin{equation*}
    \Delta(\tau_E) = \Delta(1) \cdot \Delta_{\mathrm{op}}(1) = \Delta_{\mathrm{op}}(1) \cdot \Delta(1).
  \end{equation*}
  This, in fact, implies that $\tau_E$ is the unit of the Grothendieck ring of $\mathcal{A}^*$ (and of the character algebra as well), as for any co-central $\tau\in \mathcal{A}$,
  \begin{equation*}
    \Delta(\tau_E \cdot \tau) = \Delta_{\mathrm{op}}(1) \cdot \Delta(1) \cdot \Delta(\tau) = \Delta_{\mathrm{op}}(1) \cdot \Delta(\tau) = \Delta_{\mathrm{op}}(1)\cdot \Delta_{\mathrm{op}}(\tau) =\Delta_{\mathrm{op}}(\tau) = \Delta(\tau),
  \end{equation*}
  i.e.\ $\tau_E \cdot \tau = \tau$, and similarly $\tau\cdot \tau_E = \tau$ as well. Notice that as $\tau_E$ decomposes into a sum of irrep characters, given any irrep character $\tau_a$ ($a\in \irr(\mathcal{A}^*)$) the following equations hold:
  \begin{equation*}
    \tau_a = \tau_E \tau_a = \sum_{e\in E,b} N_{ea}^b \tau_b \quad \text{and} \quad \tau_a =  \tau_a \tau_E = \sum_{e\in E,b} N_{ae}^b \tau_b.
  \end{equation*}
  Here both $N_{ea}^b\geq 0$ and $N_{ae}^b\geq 0$ are integers, and thus, as the equations above are equivalent to $\sum_{e\in E} N_{ea}^b =\delta_{ab}$ and $\sum_{e\in E} N_{ae}^b =\delta_{ab}$, there are unique labels $l_a,r_a\in E$ such that for all $e\in E$ and $a,b\in \irr(\mathcal{A}^*)$
  \begin{equation}\label{eq:N_unit}
    N_{ea}^b = \begin{cases}
      0 & \text{if $e \neq l_a$,} \\
      \delta_{ab} & \text{if $e=l_a$,}
    \end{cases}
    \quad \text{and} \quad
    N_{ae}^b = \begin{cases}
      0 & \text{if $e \neq r_a$,} \\
      \delta_{ab} & \text{if $e=r_a$.}
    \end{cases}
  \end{equation}
  Using this property of the fusion multiplicities $N_{ab}^c$, we notice that in \cref{eq:unit_boundary_fusion_intermediate} in the sum over $e\in E$ the summand is non-zero only for $e=l_a$ on the l.h.s.\ and $e=r_a$ on the r.h.s.\ and in these cases the sum over $\mu$ is trivial (because $N_{l_aa}^c = N_{ar_a}^c = \delta_{ac}$). We can thus simplify \cref{eq:unit_boundary_fusion_intermediate} to
	\begin{align}\label{eq:WBA_fusion_tensor_characteristic}
    \tikzsetnextfilename{fe0af022-9fdd-4b7e-a4ed-c0ffd367c271}
	  \begin{tikzpicture}[baseline=-1mm]
      \node[fusion tensor] (v) at (0,0) {};
      \node[fusion tensor] (w) at (1.5,0) {};
      \draw[virtual,->-] let \p1=(v.west) in (\x1,-0.25)--++(-0.5,0) node[midway,black,irrep] {$a$};
      \draw[virtual,->-] let \p1=(v.west) in (\x1,0.25)--++(-0.5,0) node[black,t] {} node[midway,black,irrep] {$e$};
      \draw[virtual,-<-] let \p1=(w.east) in (\x1,-0.25)--++(0.5,0) node[midway,black,irrep] {$a$};
      \draw[virtual,-<-=0.5] let \p1=(w.east) in (\x1,0.25)--++(0.5,0) node[black,t] {} node[midway,black,irrep] {$e$};
      \draw[virtual,-<-] let \p1=(v.east) in (\x1,0)--++(0.5,0) node[midway,black,irrep] {$c$};
      \draw[virtual,->-] let \p1=(w.west) in (\x1,0)--++(-0.5,0) node[midway,black,irrep] {$c$};
	  \end{tikzpicture} \ =  \delta_{e,l_a} \delta_{a,c} \cdot (\id_c\otimes \id_c) \quad \text{and} \quad
    \tikzsetnextfilename{ce733b2c-ca05-403c-ad4e-782bec4d348e}
	  \begin{tikzpicture}[baseline=-1mm]
      \node[fusion tensor] (v) at (0,0) {};
      \node[fusion tensor] (w) at (1.5,0) {};
      \draw[virtual,->-] let \p1=(v.west) in (\x1,0.25)--++(-0.5,0) node[midway,black,irrep] {$a$};
      \draw[virtual,->-] let \p1=(v.west) in (\x1,-0.25)--++(-0.5,0) node[black,t] {} node[midway,black,irrep] {$e$};
      \draw[virtual,-<-] let \p1=(w.east) in (\x1,0.25)--++(0.5,0) node[midway,black,irrep] {$a$};
      \draw[virtual,-<-=0.5] let \p1=(w.east) in (\x1,-0.25)--++(0.5,0) node[black,t] {} node[midway,black,irrep] {$e$};
      \draw[virtual,-<-] let \p1=(v.east) in (\x1,0)--++(0.5,0) node[midway,black,irrep] {$c$};
      \draw[virtual,->-] let \p1=(w.west) in (\x1,0)--++(-0.5,0) node[midway,black,irrep] {$c$};
	  \end{tikzpicture} \ =  \delta_{e,r_a}\delta_{a,c} \cdot (\id_c\otimes \id_c) ,
	\end{align}
that hold for all $e\in E$, and $a,c\in \irr(\mathcal{A}^*)$. Let us note that \cref{eq:WBA_fusion_tensor_characteristic} together with \cref{eq:N_unit} implies \cref{eq:unit_axiom_intermediate}, and thus it is easy to check that the element defined by \cref{eq:unit_mpo} is indeed the unit of $\mathcal{A}$ and that it satisfies the unit axiom. We have thus seen that
\begin{proposition}
  In a finite dimensional cosemisimple pre-bialgebra $\mathcal{A}$ over $\mathbb{C}$ the unit axiom \cref{eq:unit_axiom} is equivalent to the following:
  \begin{itemize}
    \item there is a set $E\subseteq \irr(\mathcal{A}^*)$ such that $\tau_E = \sum_{e\in E} \tau_e$ is the unit of the Grothendieck ring of $\mathcal{A}^*$, or equivalently, for all $a\in\mathcal{A}$ there are unique labels $l_a,r_a\in E$ such that for all $e\in E$ and $b\in \irr(\mathcal{A}^*)$
      \begin{equation*}
        N_{ea}^b = \begin{cases}
          0 & \text{if $e \neq l_a$,} \\
          \delta_{ab} & \text{if $e=l_a$,}
        \end{cases}
        \quad \text{and} \quad
        N_{ae}^b = \begin{cases}
          0 & \text{if $e \neq r_a$,} \\
          \delta_{ab} & \text{if $e=r_a$,}
        \end{cases}
      \end{equation*}
    \item  and for all $e\in E$ there are vectors $\tikz{\draw[virtual,->-] (0,0) node[t,black] {} --++(-0.5,0) node[midway,black,irrep] {$e$};} \in V_e$ and linear functionals $\tikz{\draw[virtual,-<-] (0,0) node[t,black] {} --++(0.5,0) node[midway,black,irrep] {$e$};} \in V_e^*$ such that for all $e\in E$ and $a,c\in \irr(\mathcal{A}^*)$,
  	\begin{align*}
      \tikzsetnextfilename{0f657a05-04ed-4be7-9760-61135f81441e}
  	  \begin{tikzpicture}[baseline=-1mm]
        \node[fusion tensor] (v) at (0,0) {};
        \node[fusion tensor] (w) at (1.5,0) {};
        \draw[virtual,->-] let \p1=(v.west) in (\x1,-0.25)--++(-0.5,0) node[midway,black,irrep] {$a$};
        \draw[virtual,->-] let \p1=(v.west) in (\x1,0.25)--++(-0.5,0)
          node[black,t] {} node[midway,black,irrep] {$e$};
        \draw[virtual,-<-] let \p1=(w.east) in (\x1,-0.25)--++(0.5,0) node[midway,black,irrep] {$a$};
        \draw[virtual,-<-=0.5] let \p1=(w.east) in (\x1,0.25)--++(0.5,0)
          node[black,t] {} node[midway,black,irrep] {$e$};
        \draw[virtual,-<-] let \p1=(v.east) in (\x1,0)--++(0.5,0) node[midway,black,irrep] {$c$};
        \draw[virtual,->-] let \p1=(w.west) in (\x1,0)--++(-0.5,0) node[midway,black,irrep] {$c$};
  	  \end{tikzpicture} \ =  \delta_{e,l_a} \delta_{a,c} \cdot (\id_c\otimes \id_c) \quad \text{and} \quad
      \tikzsetnextfilename{01304eff-262b-4af5-85d4-cceb809f917a}
  	  \begin{tikzpicture}[baseline=-1mm]
        \node[fusion tensor] (v) at (0,0) {};
        \node[fusion tensor] (w) at (1.5,0) {};
        \draw[virtual,->-] let \p1=(v.west) in (\x1,0.25)--++(-0.5,0) node[midway,black,irrep] {$a$};
        \draw[virtual,->-] let \p1=(v.west) in (\x1,-0.25)--++(-0.5,0)
          node[black,t] {} node[midway,black,irrep] {$e$};
        \draw[virtual,-<-] let \p1=(w.east) in (\x1,0.25)--++(0.5,0) node[midway,black,irrep] {$a$};
        \draw[virtual,-<-=0.5] let \p1=(w.east) in (\x1,-0.25)--++(0.5,0)
          node[black,t] {} node[midway,black,irrep] {$e$};
        \draw[virtual,-<-] let \p1=(v.east) in (\x1,0)--++(0.5,0) node[midway,black,irrep] {$c$};
        \draw[virtual,->-] let \p1=(w.west) in (\x1,0)--++(-0.5,0) node[midway,black,irrep] {$c$};
  	  \end{tikzpicture} \ =  \delta_{e,r_a}\delta_{a,c} \cdot (\id_c\otimes \id_c) .
  	\end{align*}
  \end{itemize}
\end{proposition}

We will refer to the subset $E$ of $\irr(\mathcal{A}^*)$ as \emph{vacuum} and the irreps in $E$ as \emph{vacuum irrep}.

Let us now apply \cref{eq:N_unit} for the case that $a\in E$. As in this case both equations in \cref{eq:N_unit} apply, we obtain that for all $a\in\irr(\mathcal{A})$ and $e,f\in E$,
\begin{equation}\label{eq:unit_unit_fusion}
  N_{ef}^a = \begin{cases}
    1 & \text{if $e=f=a$,}\\
    0 & \text{otherwise.}
  \end{cases}
\end{equation}
Stating otherwise, we have obtained that for any $e\in E$, $l_e=r_e=e$ holds. Note the difference between the MPO representation of the unit of the algebra and the element $\tau_E$.
  The MPO representation of the unit of the algebra $1\in \mathcal{A}$ is
  \begin{equation*}
  \phi^{\otimes n} \circ \Delta^{n-1} (1) =  \sum_{e\in E} \
  \begin{tikzpicture}
    \node[t] (v) at (0.5,0) {};
    \node[t] (w) at (4.5,0) {};
    \foreach \x in {1,2,4}{
      \node[tensor] (t\x) at (\x,0) {};
      \draw[->-] (t\x)--++(0,0.4);
      \draw[-<-] (t\x)--++(0,-0.4);
    }
    \node[fill=white] (dots) at (3,0) {$\dots$};
    \draw[virtual,->-] (w) -- (t4) node[irrep] {$e$};
    \draw[virtual,->-] (t4) -- (dots) node[irrep] {$e$};
    \draw[virtual,->-] (dots) -- (t2) node[irrep] {$e$};
    \draw[virtual,->-] (t2) -- (t1) node[irrep] {$e$};
    \draw[virtual,->-] (t1) -- (v) node[irrep] {$e$};
  \end{tikzpicture}  \  ,
  \end{equation*}
  and the MPO representation of the unit of the Grothendieck ring, $\tau_E$, is
  \begin{equation*}
  \phi^{\otimes n} \circ \Delta^{n-1} (\tau_E) =  \sum_{e\in E} \
  \tikzsetnextfilename{83dfe25e-5482-4794-a347-5d4aea0cab93}
  \begin{tikzpicture}
    \coordinate (v) at (0.5,0);
    \coordinate (w) at (4.5,0);
    \foreach \x in {1,2,4}{
      \node[tensor] (t\x) at (\x,0) {};
      \draw[->-] (t\x)--++(0,0.4);
      \draw[-<-] (t\x)--++(0,-0.4);
    }
    \node[fill=white] (dots) at (3,0) {$\dots$};
    \draw[virtual,->-] (w) -- (t4) node[irrep] {$e$};
    \draw[virtual,->-] (t4) -- (dots) node[irrep] {$e$};
    \draw[virtual,->-] (dots) -- (t2) node[irrep] {$e$};
    \draw[virtual,->-] (t2) -- (t1) node[irrep] {$e$};
    \draw[virtual,->-] (t1) -- (v) node[irrep] {$e$};
    \draw[virtual] (v)--++(0,-0.5) coordinate (c);
    \draw[virtual] (w)--++(0,-0.5) coordinate (d);
    \draw[virtual,->-] (c)--(d);
  \end{tikzpicture}  \  .
  \end{equation*}
These two MPOs coincide if and only if the irrep $e$ is one-dimensional for all $e\in E$.

Until now, we have only used the unit axiom, and not the counit axiom: the unit $\tau_E\in \mathcal{A}$ of the Grothendieck ring exists and the MPO representation of the unit $1\in\mathcal{A}$ is rank-one in each injective block even for pre-bialgebras that satisfy the unit axiom but not the counit axiom. Let us now derive a consequence of the counit axiom that we will use later on. First, note that the counit axiom can equally be written as
\begin{equation*}
  \Delta^2(\epsilon) = (1\otimes \Delta(\epsilon)) \cdot (\Delta(\epsilon) \otimes \epsilon) = (\Delta(\epsilon) \otimes \epsilon) \cdot (\epsilon\otimes \Delta(\epsilon)),
\end{equation*}
and therefore the graphical representation of the counit axiom is
\begin{equation*}
  \sum_{ed\mu\nu}
  \tikzsetnextfilename{3288b985-658d-4b95-b3da-5dbe72580f7b}
 \ .
\end{equation*}
Using now the orthogonality relations in the first equation, we conclude that the following equation holds:
\begin{equation}\label{eq:counit_consequence}
  \sum_{e\mu\nu}
  \tikzsetnextfilename{ca7f10eb-8125-4ffe-b50d-186052e44089}
  %
 \ , & \text{if $N_{ab}^d\neq 0$,} \\
    0,                     & \text{if $N_{ab}^d = 0$}.
  \end{cases}
\end{equation}
This equation, however, is not a sufficient condition for the unit of the pre-bialgebra to satisfy the counit axiom \cref{eq:counit_axiom}.

\subsection{Weak Hopf algebras}

In this section we define weak Hopf algebras. In these algebras we define as set of elements called integrals, and show that a cosemisimple weak Hopf algebra over $\mathbb{C}$ has a special integral $\Lambda$ such that is a projector and is such that $\Lambda(fg) = \Lambda(S^2(g)f)$ holds for all $f,g\in\mathcal{A}^*$.

\begin{definition}[Weak Hopf algebra]\label{def:wha}
  A weak Hopf algebra (WHA) is a WBA $\mathcal{A}$ together with a linear map $S:\mathcal{A}\to\mathcal{A}$, called antipode, such that
  \begin{gather*}
     \sum_{(x)} S(x_{(1)})x_{(2)}= \sum_{(1)}  1_{(1)} \epsilon(x1_{(2)}), \\
     \sum_{(x)} x_{(1)}S(x_{(2)})= \sum_{(1)} \epsilon ( 1_{(1)} x)  1_{(2)}, \\
     \sum_{(x)} S(x_{(1)})x_{(2)}S(x_{(3)}) = S(x).
	\end{gather*}
\end{definition}

Given a WHA $\mathcal{A}$ with antipode $S_{\mathcal{A}}$, it is easy to check that the map $S_{\mathcal{A}^*}:\mathcal{A}^* \to \mathcal{A}^*$ defined by
\begin{equation}\label{eq:antipode_dual}
  S_{\mathcal{A}^*}(f) = f\circ S_{\mathcal{A}},
\end{equation}
for all $f\in \mathcal{A}^*$ (i.e.\ $S_{\mathcal{A}^*} = S_{\mathcal{A}}^T$), satisfies the antipode axioms as well, and thus $\mathcal{A}^*$ is a  WHA too. From now on, we do not differentiate between the antipode of $\mathcal{A}$ and that of $\mathcal{A}^*$, and denote both by $S$. The antipode $S$ of a WHA is an anti-homomorphism ($S(xy) = S(y)S(x)$), an anti-cohomomorphism ($\Delta\circ S = (S\otimes S) \circ \Delta_{\mathrm{op}}$) and a bijection (see \cite{Bohm1999} for a proof). In fact, an equivalent characterization of the antipode (of $\mathcal{A}^*$) is that it is a bijective anti-homomorphism of $\mathcal{A}^*$ such that for every $f\in\mathcal{A}^*$
\begin{align}
  \sum_{(f)} S(f_{(1)}) f_{(2)} \otimes f_{(3)} = \epsilon_{(1)} \otimes f \epsilon_{(2)}, \label{eq:antipode_ax_1} \\
  \sum_{(f)} f_{(1)} \otimes f_{(2)} S(f_{(3)})  = \epsilon_{(1)} f \otimes \epsilon_{(2)}. \label{eq:antipode_ax_2}
\end{align}

Let $\mathcal{A}$ be a cosemisimple WHA and $\psi:\mathcal{A}^*\to \End(W)$ be a representation of $\mathcal{A}^*$ on a vector space $W$. The antipode $S$ of $\mathcal{A}^*$ is an anti-homomorphism, and thus the linear map $\bar{\psi}:\mathcal{A}^* \to \End(W^*)$ defined by
\begin{equation*}
  \bar{\psi}(f) =(\psi\circ S(f))^T
\end{equation*}
is a representation of $\mathcal{A}^*$ on the vector space $W^*$. As $S$ is a bijection, $\bar{\psi}$ is an irrep if $\psi$ is an irrep. Let $\irr(\mathcal{A}^*)$, as above, denote the irrep equivalence classes of $\mathcal{A}^*$ and for every $a\in \irr(\mathcal{A}^*)$ let us fix a representation $\psi_a$ on the vector space $V_a$ from the irrep class $a$. For each irrep $\psi_a$ the representation $\bar{\psi}_a$ is an irrep in another class, that we denote by $\bar{a}$ (note that $\bar{a}$ might coincide with $a$). As $\bar{\psi}_a$ and $\psi_{\bar{a}}$ are in the same irrep class, they are related to each other by a basis transformation, i.e.\ there are linear maps $Z_{a}: W_{\bar{a}} \to W_{a}^*$ such that
\begin{equation*}
  \overline{\psi_a}(f) = Z_{a} \cdot \psi_{\bar{a}}(f)\cdot Z_{a}^{-1}.
\end{equation*}
Let us denote this equation using the graphical notation of tensor networks. We have to be careful with the notation, as $\overline{\psi_a}(f)$ is not one of the representations that we have fixed previously. We thus have to use a more verbose notation and display the representation itself and that it is a linear map from $W_a^*$ to $W_a^*$:
  \begin{equation*}
 \ .
  \end{equation*}
Let us now reverse the two arrows at the two ends of the figure. This changes $W_a^*$ to $W_a$, and thus, as it is one of the previously fixed vector spaces, we can simply write $a$ on the line. Similarly, $\overline{\psi_a}(f)$ changes to $\overline{\psi_a}(f)^T = \psi_a\circ S(f)$, and thus, as $\psi_a$ is one of the previously fixed representations, the rank-two tensor on the l.h.s.\ of the equation can simply be labeled by $S(f)$:
  \begin{equation}\label{eq:SZ}
    \tikzsetnextfilename{7b111b35-e547-40d3-a4f1-2b0a2621cbb4}
    %
 \ .
  \end{equation}
The indices of $S(f)$ are now oriented in the opposite direction as that of $f$. Because of this, we often bend back the indices on the r.h.s.\ of this equation such that we obtain
  \begin{equation}\label{eq:SZ_2}
    %
 \ .
  \end{equation}
Let us warn here the reader to the subtleties of our notation. First, $Z_a$ has two incoming indices, and thus the only way to differentiate between the two indices is via the labels $a$ and $\bar{a}$.  Note that even for irrep labels for which $a=\bar{a}$, this forces us to write $a$ and $\bar{a}$ on the two legs of the tensor. Second, there are two different operators with very similar graphical representations:
\begin{equation*}
  %
 \ .
\end{equation*}
These two tensors should be read in the opposite direction: the first index of $Z_a$ is the one labeled by $\bar{a}$ on its left, while the first index of $Z_{\bar{a}}$ is the one labeled by $a$ on its right.

As we have mentioned above, the antipode of $\mathcal{A}$ is the transpose of the antipode of $\mathcal{A}^*$, see \cref{eq:antipode_dual}. Therefore for all $f\in\mathcal{A}^*$ and $x\in\mathcal{A}$ the following holds:
\begin{equation*}
  \sum_{a\in \irr(\mathcal{A}^*)}
  %
 \ ,
\end{equation*}
where in the second equation we have used \cref{eq:SZ}. As this equation holds for all $f\in\mathcal{A}^*$, the boundary condition $b(S(x))$ describing $S(x)$ can be expressed as
  \begin{equation}\label{eq:boundary_S}
    \tikzsetnextfilename{5e45a41f-63ad-4401-a87f-c057381502ad}
    %
 \ .
  \end{equation}
This allows us to write the repeated coproduct of $S(x)$ as
\begin{equation}\label{eq:antipode_mpo_1}
  \phi^{\otimes n} \circ \Delta^{n-1} \circ S(x) =  \sum_{a} \
  \tikzsetnextfilename{2bc799e8-47ff-4d4b-bf68-86a4aa472902}
  %
 \ .
\end{equation}
Let us consider now a new rank-4 tensor, denoted by gray dot, defined by
\begin{equation*}
  %
 \ =
  \sum_y \phi\circ S(y) \otimes \psi_a(\delta_y) =
  \sum_x \phi(x) \otimes \psi_a \circ S(\delta_x),
\end{equation*}
where in the second equation we have used that if $B$ is a basis with basis elements $y$, then the set $x=S(y) \ (y\in B)$ forms a basis as well and its dual basis is $\delta_x = S^{-1}(\delta_y) \ (y\in B)$. Note that we have oriented the virtual indices of this tensor in the opposite direction as in the MPO tensor given by black dots. This is because $S$ is an anti-homomorphism; more precisely, we can relate the two tensors using \cref{eq:SZ}:
\begin{equation}\label{eq:S_MPO_tensor}
  \tikzsetnextfilename{43d1a570-2798-45f6-b732-30d466b95000}
  \begin{tikzpicture}[baseline=-1mm]
    \node[tensor,gray] (t) at (0,0) {};
    \draw[virtual,->-] (t)--++(0.5,0);
    \draw[virtual,-<-] (t)--++(-0.5,0);
    \node[irrep,anchor=south] at (-0.3,0) {$a$};
    \node[irrep,anchor=south] at (0.3,0) {$a$};
    \draw[->-] (t) --++ (0,0.4);
    \draw[-<-] (t) --++ (0,-0.4);
  \end{tikzpicture} \ =
  \tikzsetnextfilename{79b7bc84-2727-409d-b5eb-8cd86ccf7514}
  \begin{tikzpicture}[baseline=-1mm]
    \node[tensor] (t) at (0,0) {};
    \node[tensor,label=below:$Z_{a}^{-1}$] (zinv) at (0.6,0) {};
    \node[tensor,label=below:$Z_{a}$] (z) at (-0.6,0) {};
    \draw[virtual,->-] (zinv)--++(0.5,0);
    \draw[virtual,->-] (zinv)--(t);
    \draw[virtual,->-] (t)--(z);
    \draw[virtual,-<-] (z)--++(-0.5,0);
    \node[irrep,anchor=south] at (-0.3,0) {$\bar{a}$};
    \node[irrep,anchor=south] at (-0.85,0) {$a$};
    \node[irrep,anchor=south] at (0.3,0) {$\bar{a}$};
    \node[irrep,anchor=south] at (0.85,0) {$a$};
    \draw[->-] (t) --++ (0,0.4);
    \draw[-<-] (t) --++ (0,-0.4);
  \end{tikzpicture} \ .
\end{equation}
Using this gray tensor, one can rewrite \cref{eq:antipode_mpo_1} as
\begin{equation}\label{eq:S_x_MPO}
    \phi^{\otimes n} \circ \Delta^{n-1} \circ S(x) = \sum_a \
  \begin{tikzpicture}
    \foreach \x in {1,2,4}{
      \node[tensor,gray] (t\x) at (\x,0) {};
      \draw[->-] (t\x)--++(0,0.4);
      \draw[-<-] (t\x)--++(0,-0.4);
    }
    \node[fill=white] (dots) at (3,0) {$\dots$};
    \node[tensor,label=below:$b_a(x)$] (b) at (0,0) {};
    \draw[virtual,-<-] (t4)--(dots) node[midway,black,irrep] {$a$};
    \draw[virtual,-<-] (dots)--(t2) node[midway,black,irrep] {$a$};
    \draw[virtual,-<-] (t2)--(t1) node[midway,black,irrep] {$a$};
    \draw[virtual,-<-] (t1)--(b) node[midway,black,irrep] {$a$};
    \draw[virtual] (b)--++(-0.5,0)--++(0,-0.6) coordinate (d);
    \draw[virtual] (t4)--++(0.5,0) --++(0,-0.6) coordinate (e);
    \draw[virtual,-<-] (d)--(e);
  \end{tikzpicture} \ .
\end{equation}
This MPO, by definition, is $\phi^{\otimes n} \circ S^{\otimes n} \circ \Delta_{\mathrm{op}}^{n-1} (x)$ (read from left to right), as the virtual index of the MPO tensor is written in the opposite direction than as usual. As $\phi$ is w.l.o.g.\ injective, we have obtained the relation
\begin{equation*}
  \Delta^{n-1} \circ S (x) = S^{\otimes n}\circ \Delta_{\mathrm{op}}^{n-1} (x),
\end{equation*}
that is, in fact, a simple consequence of the fact that $S$ is an anti-cohomomorphism ($\Delta\circ S = (S\otimes S) \circ \Delta_{\mathrm{op}}$).

Using the tensors $Z_a$ that we defined in \cref{eq:SZ}, the antipode axioms of $\mathcal{A}^*$, \cref{eq:antipode_ax_1,eq:antipode_ax_2},  can be expressed as
  \begin{align*}
    \sum_{b\mu \, d\nu} \
    \tikzsetnextfilename{b9d1b674-e215-455e-af5d-85c7f8197922}
 \ .
  \end{align*}
Let us explain briefly how to arrive at these graphical representations of \cref{eq:antipode_ax_1} and \cref{eq:antipode_ax_2}. In \cref{eq:antipode_ax_1}, on the l.h.s.\ we first need to take two repeated coproduct of $f$: this is achieved by the two pairs of fusion tensors (see \cref{eq:splitting}). Then the antipode of the first component is taken (see \cref{eq:SZ_2}) and it is multiplied with the second component of the coproduct. On the r.h.s.\ of the same equation, we multiply $f$ with the second component of $\Delta(\epsilon)$ (see \cref{eq:delta_epsilon}). The second equation is obtained similarly.
As these equations hold for any $f \in\mathcal{A}^*$, we conclude that
  \begin{align*}
    \sum_{d\mu\nu} \
  	\tikzsetnextfilename{34ec966c-594a-4b78-9ed8-b6ff6f552de0}
 \ .
  \end{align*}
  Using the orthogonality relations of the fusion tensors and multiplying the first equation by $Z_{a}$, and the second equation by $Z_{b}^{-1}$, we obtain
  \begin{align*}
    \sum_{\mu} \
    \tikzsetnextfilename{d738dc86-1424-49a2-8fda-bc606ef52a5a}
    %
 \ .
  \end{align*}
In particular, there are complex matrices $C_{ab}^d$ of size $N_{ab}^d \times N_{\bar{a}d}^b$ and $\hat{C}_{ab}^d$ of size $N_{ab}^d\times N_{d\bar{b}}^a$ such that
  	\begin{equation}\label{eq:fold-back}
    \sum_{\mu} \left(C_{ab}^d\right)_{\nu\mu} \
    \tikzsetnextfilename{b9c92520-e888-4472-af70-1d8b8ad4893e}
    %
 \ .
  	\end{equation}
Let us now show that $C_{ab}^d$ and $\hat{C}_{ab}^d$ are both invertible. In both equations the tensors on both sides are linearly independent, therefore $N_{ab}^d\leq N_{\bar{a}d}^b$ and $N_{ab}^d \leq N_{d\bar{b}}^a$. Applying these inequalities twice, we obtain that $N_{ab}^d\leq N_{\bar{a}d}^b \leq N_{\bar{\bar{a}}b}^d$, and that $N_{ab}^d \leq N_{d\bar{b}}^d \leq N_{a\bar{\bar{b}}}^d$. As $a\mapsto \bar{\bar{a}}$ is a permutation of $\irr(\mathcal{A}^*)$, repeating these equations a few\footnote{the order of the permutation $a\mapsto \bar{\bar{a}}$} times we obtain that $N_{ab}^d\leq N_{\bar{\bar{a}}b}^d\leq \dots \leq N_{ab}^d$ and that $N_{ab}^d \leq N_{a\bar{\bar{b}}}^d \leq \dots \leq N_{ab}^d$, i.e.\ in all of the inequalities above equality holds. In particular, for all $a,b,c\in \irr(\mathcal{A}^*)$, the numbers $N_{ab}^d$ possess the following symmetries:
\begin{equation*}
  N_{ab}^d =  N_{\bar{a}d}^b \quad \text{and} \quad N_{ab}^d = N_{d\bar{b}}^a.
\end{equation*}
This implies that in both equation in \cref{eq:fold-back} we transform $N_{ab}^d$ linearly independent vectors to $N_{ab}^d$ linearly independent vectors, i.e.\ both $C_{ab}^d$ and $\hat{C}_{ab}^d$ are invertible. It is easy to check that the argument above can be repeated backwards, i.e.\ \cref{eq:fold-back} implies the antipode axioms. We have thus seen that
\begin{proposition}\label{thm:foldback}
  Let $\mathcal{A}$ be a finite dimensional cosemisimple WBA. Then the map $S:\mathcal{A}^*\to \mathcal{A}^*$ given by
  \begin{equation*}
    \tikzsetnextfilename{eb24378a-362f-4935-b8fa-74019ac142db}

  \end{equation*}
  is an antipode making $\mathcal{A}^*$ a WHA if and only if there are invertible matrices $C_{ab}^d$ of size $N_{ab}^d \times N_{\bar{a}d}^b$ and $\hat{C}_{ab}^d$ of size $N_{ab}^d\times N_{d\bar{b}}^a$ such that
  	\begin{equation*}
    \sum_{\mu} \left(C_{ab}^d\right)_{\nu\mu} \
    \tikzsetnextfilename{1c0b1a53-554e-4179-91e7-8d768de64cb9}
    %
 \ .
  	\end{equation*}
\end{proposition}
As it is interesting on its own, let us restate here the symmetry properties of the fusion multiplicities $N_{ab}^c$ together with some of the consequences of these symmetries:
\begin{proposition}\label{thm:N_symmetry_consequence}
  Let $\mathcal{A}$ be a finite dimensional cosemisimple weak Hopf algebra, and $a,b,c\in \irr(\mathcal{A}^*)$. Then
  \begin{equation}\label{eq:N_symmetries}
    N_{ab}^d =  N_{\bar{a}d}^b \quad \text{and} \quad N_{ab}^d = N_{d\bar{b}}^a.
  \end{equation}
  Moreover, the following equations hold:
  \begin{itemize}
  \item  $\bar{\bar{a}} = a$ for all $a\in\irr(\mathcal{A}^*)$,
  \item  $l_a = r_{\bar{a}}$ and $r_a=l_{\bar{a}}$ for all $a\in \irr(\mathcal{A}^*)$, where $l_a,r_a\in E$ are as in \cref{eq:N_unit},
  \item  $e = \bar{e}$ for all $e\in E$.
  \end{itemize}
\end{proposition}
\begin{proof}
We have already proven \cref{eq:N_symmetries} above. Let us now show that $a=\bar{\bar{a}}$ for all irrep labels $a\in \irr(\mathcal{A}^*)$. Applying the left side of \cref{eq:N_symmetries} twice with $d=a$ and $b=r_a$, we obtain that
\begin{equation*}
  1 = N_{ar_a}^a = N_{\bar{a}a}^{r_a} = N_{\bar{\bar{a}}r_a}^a.
\end{equation*}
Using now \cref{eq:N_unit}, this latter fusion multiplicity can only be non-zero if $a = \bar{\bar{a}}$.

Let us now show that for all $a\in\irr(\mathcal{A})$, $l_{\bar{a}} = r_a$ holds. Using the left  equation in \cref{eq:N_symmetries} with $b=r_a$ and $d=a$, we obtain that
\begin{equation*}
  1 = N_{ar_a}^a = N_{\bar{a}a}^{r_a} = N_{r_a\bar{a}}^{\bar{a}},
\end{equation*}
where the last equation is the right equation in \cref{eq:N_symmetries}. As $r_a\in E$ and the above fusion number is non-zero, \cref{eq:N_unit} implies that for all $a\in \irr(\mathcal{A})$, $r_a = l_{\bar{a}}$. The equation $r_{\bar{a}} = l_{a}$ can be shown in a similar way.

Let us finally prove that the permutation $a\mapsto \bar{a}$ leaves the set $E$ invariant, and in fact for every $e\in E$, $\bar{e} = e$. Let us use \cref{eq:N_symmetries} with $a=b=d=e\in E$. We obtain that
\begin{equation*}
  1 = N_{ee}^e = N_{\bar{e}e}^e,
\end{equation*}
and thus, as $e\in E$, using \cref{eq:N_unit}, this latter fusion number can only be non-zero if $e= \bar{e}$.
\end{proof}

Let us show now that in a cosemisimple WHA the matrices $Z_a$ can be expressed with the help of the fusion tensors and the unit of the underlying WBA. Therefore, if in a cosemisimple WBA there exists an antipode making it a WHA, this antipode is uniquely determined by the WBA structure. This statement, in fact, is more general than what we show here, and also holds in the non-cosemisimple case (see e.g. \cite{Bohm1999}).

Let $c\in\irr(\mathcal{A}^*)$ and let us set $\bar{a}=\bar{c}$, $d=c$ and $b=r_c\in E$ in the left equation in \cref{eq:fold-back}, and $d=\bar{c}$, $b = \bar{c}$ and $a=r_c\in E$ in the right equation in \cref{eq:fold-back} and use that $\bar{\bar{c}} = c$. As $N_{\bar{c}c}^{r_{c}} = 1$, we obtain that there are non-zero complex numbers $C_{cr_c}^c$ and $\hat{C}_{r_c\bar{c}}^{\bar{c}}$ such that
  	\begin{equation*}
    C_{cr_c}^c \
 \ .
  	\end{equation*}
Taking the tensor product of these two equations, applying the vectors describing $b_{r_c}(1)$ and using \cref{eq:WBA_fusion_tensor_characteristic},  we conclude that there are non-zero numbers $w_c$ such that
\begin{equation}\label{eq:w_a}
  \tikzsetnextfilename{ae4f8303-1934-42e0-9218-15b551a50d6c}
  %
  \ .
\end{equation}
We have thus obtained that the matrices $Z_c$ can be expressed with the help of the fusion tensors and the vectors describing the unit of the WBA, and thus the  antipode  of a cosemisimple WHA is unique and completely determined by the WBA structure.

Let us now define a linear functional that will play a central role in the rest of the paper. Before the definition, note that while the number $w_c$ defined in \cref{eq:w_a} depend on the concrete choice\footnote{We could have chosen $\lambda Z_c$ to describe the antipode instead of $Z_c$, for any $\lambda\in\mathbb{C}$. This is in fact the only freedom that we have in the choice of $Z_c$.} of the matrices $Z_c$ and $Z_{\bar{c}}^{-1}$, the number $w_c w_{\bar{c}}$ is independent from this choice, and it is also independent from all the choices we have made before (the fusion tensors and the vectors describing the unit of the algebra). This is because all these objects appear in pairs in the following equation:
\begin{equation*}
  \tikzsetnextfilename{2b4f3e77-e90c-4325-ba64-d85f812ece7b}
  \ .
\end{equation*}
Let $d_a$ be complex numbers such that $d_a^2 = w_a w_{\bar{a}}$ and $d_a = d_{\bar{a}}$. Later we will prove that $w_aw_{\bar{a}}>0$ and thus one can choose $d_a$ to be positive. Until then, $d_a=d_{\bar{a}}$ is an arbitrary square root of $w_aw_{\bar{a}}$. Let us now define the linear functional $g\in \mathcal{A}^*$ through
\begin{equation}\label{eq:x_def}
  %
 \ ,
\end{equation}
where we have used that $f$ for any matrix $m$ of the form $\bigoplus_a m_a$, $m_a\in \End(V_a)$, one can find a linear functional such that $\psi_a(f) = m_a$. The linear functional $g$ defined in this way is independent of all previous choices that we have made, except for the choice of the square root $d_a=d_{\bar{a}}$ of $d_a^2$. The importance of this linear functional is that the conjugation by it describes the action of $S^2$. Indeed, applying \cref{eq:SZ} twice, we obtain that
\begin{equation}\label{eq:S2}
  %
 \ .
\end{equation}
Let us now show a few additional properties of $g$. First, note that $w_{\bar{a}}/d_a = d_a/w_a = d_{\bar{a}}/w_a$, and thus
\begin{equation*}
  %
 \ ,
\end{equation*}
i.e.\ $g$ satisfies $g^{-1} = S(g)$. Second, the traces of $\psi_a(g)$ and $\psi_a(g^{-1})$ evaluate to the following:
\begin{gather}\label{eq:tr_g_inv}
  \tr \psi_a(g^{-1})  =  \frac{d_a}{w_{\bar{a}}}
  %
  = d_{a} \cdot \epsilon_{r_{\bar{a}}} (1),
\end{gather}
where $\epsilon_{r_a}\in\mathcal{A}^*$ is the irrep projector onto the irrep sector $r_a$ and $\epsilon_{r_{\bar{a}}}\in\mathcal{A}^*$ is the irrep projector onto the irrep sector $r_{\bar{a}}$.

In the following we will encounter algebra elements of the form $x = \sum_a (\id \otimes g) \circ \Delta(\tau_a)$. These algebra elements have the following MPO representation:
\begin{equation*}
  \phi^{\otimes n} \circ \Delta^{n-1} (x)  = \sum_a \lambda_a \cdot
  	  %
 \ .
\end{equation*}
The algebra elements are exactly those that satisfy $\Delta_{\mathrm{op}}(x) = (S^{-1}\otimes \id)\circ \Delta(x)$. Indeed,
\begin{equation}\label{eq:q-trace}
  \phi^{\otimes 2} \circ\Delta_{\mathrm{op}}(x) = \sum_a \lambda_a \cdot
  %
  \   =
  \phi^{\otimes 2}\circ (S^{-2}\otimes\id) \circ \Delta(x).
\end{equation}
One can easily see that all elements that satisfy the above cocommutation relation are in fact of this form. Such algebra elements will be called \emph{q-traces}:
\begin{definition}[Q-trace]
  Let $\mathcal{A}$ be a weak Hopf algebra. The algebra element $x\in\mathcal{A}$ is called a q-trace if
  \begin{equation*}
    \Delta_{\mathrm{op}}(x) = (S^{-2} \otimes \id) \circ \Delta(x).
  \end{equation*}
\end{definition}
In the following we will construct an algebra element that is not only a q-trace, but it is also a \emph{left integral}:
\begin{definition}[Integrals]
  A left integral of a WHA $\mathcal{A}$ is an element $\Lambda\in \mathcal{A}$ such that for all $x\in \mathcal{A}$
  \begin{equation*}
    (1\otimes x) \cdot\Delta(\Lambda) = (S(x)\otimes 1) \cdot\Delta(\Lambda).
  \end{equation*}
  The element $\Lambda\in \mathcal{A}$ is called a right integral if $S(\Lambda)$ is a left integral.
  An integral that is both left and right is a two-sided integral. A non-degenerate (see \cref{def:non-degenerate}) integral $\Lambda$ is normalized if $\Lambda^2 = \Lambda$. Finally, an integral is called a Haar integral if it is two-sided, non-degenerate  and normalized.
\end{definition}

We will now prove the main result of this section: that in a cosemisimple WHA there is a special integral that is non-degenerate, normalized and a q-trace:

\begin{theorem}\label{thm:wha_special_integral}
In a finite dimensional cosemisimple WHA $\mathcal{A}$ over $\mathbb{C}$ the element $\Lambda\in\mathcal{A}$ defined\footnote{Remember that $\phi$ is w.l.o.g.\ injective and thus the value of $\phi^{\otimes n} \circ \Delta^{n-1} (\Lambda)$ defines $\Lambda$ uniquely.} by
\begin{equation}\label{eq:spec_integral}
  \phi^{\otimes n} \circ \Delta^{n-1} (\Lambda)  = \sum_{a\in \irr(\mathcal{A}^*)} \frac{d_a}{\sum_{x:\, l_x = l_a} d^2_x }\cdot
      \tikzsetnextfilename{c3c5b78d-5bca-4541-ba3c-a7e24f65608e}
 \ ,
\end{equation}
is a non-degenerate normalized left integral that is also a q-trace.
\end{theorem}

In the proof of \cref{thm:wha_special_integral} we will use the following two lemmas that we prove in \cref{sec:foldback_lemma_proof}:
\begin{restatable}{lemma}{Blemma}\label{lem:foldback}
   Let $\mathcal{A}$ be a cosemisimple WHA, and $g$ be the linear functional defined in \cref{eq:x_def}. Then there exists an $N_{ab}^c\times N_{ab}^c$ matrix $B_{ab}^c$ such that the linear functional $g\in\mathcal{A}^*$ satisfies
     \begin{equation}\label{eq:B}
       %
 \ .
     \end{equation}
     Moreover, $\left(B_{ab}^c\right)^2 = \id_{N_{ab}^c}$ and the following equation holds as well:
     \begin{equation*}
       \sum_{\mu}  w_a \
       %
    \ .
     \end{equation*}
\end{restatable}

Investigating the trace of $B_{ab}^c$ we can prove that the numbers $d_a^2$ are positive and thus that the denominator in the r.h.s.\ of \cref{eq:spec_integral} is non-zero. We postpone the proof of this lemma to \cref{sec:foldback_lemma_proof} as well.

\begin{restatable}{lemma}{Tlemma}\label{cor:T}
  For all $a\in \irr(\mathcal{A}^*)$, $d_a^2 = w_a w_{\bar{a}}>0$. Let moreover $T_{ab}^c$ be defined by $T_{ab}^c = \sum_{\mu} \left(B_{ab}^c\right)_{\mu\mu}$. Then the following equations hold:
  \begin{equation*}
    \sum_b T_{ab}^c \cdot d_b = d_a \cdot \delta_{l_a l_c} \cdot d_c \quad \text{and} \quad \sum_{x:l_x = l_a} d_x^2 = \sum_{x:l_x =r_a} d_x^2.
  \end{equation*}
\end{restatable}

With these two  lemmas in hand, we can proceed to the proof of \cref{thm:wha_special_integral}.

\begin{proof}[Proof of \cref{thm:wha_special_integral}]
  Let us define $L$ as
  \begin{equation*}
    \phi^{\otimes n} \circ \Delta^{n-1} (L)  = \sum_{a\in \irr(\mathcal{A}^*)} d_a \cdot
 \ .
  \end{equation*}
  We will prove that $L$ a non-degenerate left integral that is related to $\Lambda$ as follows. Let us define $N$ via
  \begin{equation}\label{eq:normalization}
    \phi^{\otimes n} \circ \Delta^{n-1} (N) =  \sum_{e\in E} \sum_{x:\, l_x = e} d^2_x \
    %
  \  .
  \end{equation}
  This algebra element $N$ is central, as $Nx$ is described by the MPO
  \begin{equation*}
    \phi^{\otimes n} \circ \Delta^{n-1} (Nx)  = \sum_{a\in \irr(\mathcal{A}^*)} \sum_{x:\, l_x = l_a} d^2_x \cdot
    	  %
 \ ,
  \end{equation*}
  while $xN$ is described by the MPO
  \begin{equation*}
    \phi^{\otimes n} \circ \Delta^{n-1} (xN)  = \sum_{a\in \irr(\mathcal{A}^*)} \sum_{x:\, l_x = r_a} d^2_x \cdot
    	  %
 \ ,
  \end{equation*}
  and, by \cref{cor:T}, $\sum_{x:l_x = l_a} d_x^2 = \sum_{x: l_x = r_a} d_x^2$. Note as well that, as $d^2_x>0$ for all $x$ (see again \cref{cor:T}), $N$ is invertible with inverse
  \begin{equation*}
    \phi^{\otimes n} \circ \Delta^{n-1} (N^{-1}) =  \sum_{e\in E} \frac{1}{\sum_{x:\, l_x = e} d_x^2} \
    %
  \  .
  \end{equation*}
  In fact, $\Lambda = N^{-1} L = LN^{-1}$.  Therefore, in order to prove that $\Lambda$ is a normalized left integral, we only have to prove that $L$ is a left integral, and that $L^2 = N L$. Note here that $\Lambda$ is of the form \cref{eq:q-trace}, and thus it is a q-trace. It is also non-degenerate as $g$ is invertible and the normalization constant is non-zero.

  Let us prove now that $L$ is a left integral.   Let us start with the following simple corollary of \cref{lem:foldback}:
  \begin{equation*}
    \sum_{\mu} \frac{d_ad_b}{d_c}\
    \tikzsetnextfilename{7ad583f6-f93f-4dfe-a121-6354b60de6ac}
 \ .
  \end{equation*}
  Let us rearrange this equation to obtain
  \begin{equation*}
    \sum_{\mu} d_a \
    \tikzsetnextfilename{0688a31c-9b9b-4e4a-b6de-7c9ae3be2ab6}
    %
 \ ,
  \end{equation*}
  where in the second equality we have used that
  \begin{equation*}
     \frac{w_b}{d_b}
    %
 \ .
  \end{equation*}
  This implies that for all $x\in \mathcal{A}$,
  \begin{equation*}
    \sum_{abc\mu}  \ d_a \cdot
    \tikzsetnextfilename{a557b571-d523-453b-8693-75dfa7f838d8}
    %
  \ .
  \end{equation*}
  We can now use \cref{eq:MPO_tensor_product} to obtain
  \begin{equation*}
    \sum_{a,b} \ d_a \cdot
    \tikzsetnextfilename{d1638fc3-99b4-413d-a023-a2529f733827}
    %
  \ ,
  \end{equation*}
  or equivalently, using that the representation $\phi$ is w.l.o.g.\ injective, that  $(1\otimes x) \cdot \Delta(L) = (S(x)\otimes 1) \cdot \Delta(L)$, i.e.\ that $L$ is a left integral. This integral is automatically non-degenerate (see \cref{sec:MPS_cocentral}) as $b_a(L) = d_a\psi_a(g)$ is invertible for all $a\in \irr(\mathcal{A}^*)$.

  Let us now show that $L^2 = N L$, and thus that $\Lambda$ is a normalized integral. The matrices $b_c(L^2)$ describing $L^2$ are of the form (see \cref{eq:product_boundary})
  \begin{equation*}
    b_c(L^2) =
    \sum_{ab\mu} d_a d_b\cdot
    \tikzsetnextfilename{6ddc9766-1632-45dd-8592-95d914016254}
    \begin{tikzpicture}[baseline=-1mm]
      \node[fusion tensor] (v) at (0,0) {};
      \node[fusion tensor] (w) at (1,0) {};
      \node[tensor, label = below: $g$] (xb) at (0.5,-0.25) {};
      \node[tensor, label = above: $g$] (xa) at (0.5,0.25) {};
      \node[anchor=north,inner sep=3pt,font=\tiny] at (v.south) {$\mu$};
      \node[anchor=north,inner sep=3pt,font=\tiny] at (w.south) {$\mu$};
      \draw[virtual,->-] let \p1=(w.west) in (\x1,-0.25)--(xb) node[black,irrep,midway] {$b$};
      \draw[virtual,-<-] let \p1=(v.east) in (\x1,-0.25)--(xb) node[black,irrep,midway] {$b$};
      \draw[virtual,->-] let \p1=(w.west) in (\x1,0.25)--(xa) node[black,irrep,midway] {$a$};
      \draw[virtual,-<-] let \p1=(v.east) in (\x1,0.25)--(xa) node[black,irrep,midway] {$a$};
      \draw[virtual,-<-] (w.east)  --++ (0.5,0) node[black,irrep,midway] {$c$};
      \draw[virtual,->-] (v.west)  --++ (-0.5,0) node[black,irrep,midway] {$c$};
    \end{tikzpicture} =
    \sum_{ab} d_a T_{ab}^c d_b \cdot
    \begin{tikzpicture}[baseline=-1mm]
      \node[tensor,label=below:$g$] (x) at (0,0) {};
      \draw[virtual,-<-] (x)--++(0.6,0) node[black,irrep,midway] {$c$};
      \draw[virtual,->-] (x)--++(-0.6,0) node[black,irrep,midway] {$c$};
    \end{tikzpicture} = \sum_{a:\; l_a = l_c} d_a^2\cdot
    \begin{tikzpicture}[baseline=-1mm]
      \node[tensor,label=below:$g_c$] (x) at (0,0) {};
      \draw[virtual,-<-] (x)--++(0.6,0) node[black,irrep,midway] {$c$};
      \draw[virtual,->-] (x)--++(-0.6,0) node[black,irrep,midway] {$c$};
    \end{tikzpicture}  = b_c(NL)\ ,
  \end{equation*}
  where in the second equality we have used the definition of $T_{ab}^c$, and in the third, the last point of \cref{cor:T}.
\end{proof}

A few remarks are in place. First, it is known that every cosemisimple WHA is finite dimensional (see \cite{Nikshych2003} and Theorem 3.13 in \cite{Bohm1999}), and thus the assumption on $\mathcal{A}$ being finite dimensional is redundant. Second, the normalized integral $\Lambda$ provides a separability element $I\in \mathcal{A}\otimes \mathcal{A}$ via $I = (S^{-1}\otimes\id) \circ \Delta(\Lambda)$ for the algebra $\mathcal{A}$. Therefore $\mathcal{A}$ is separable, and in particular, it is semisimple, i.e.\ we have obtained Theorem 2.26.\ of \cite{Etingof2002}. Finally, let us note that $\Lambda$ is closely related to the canonical left integral $\hat{L}$ defined in \cite{Nikshych2003} (see also \cite{Bohm1999} and \cite{Etingof2002}). In \cite{Nikshych2003}, the canonical left integral is defined as
\begin{equation*}
  \hat{L}(f) = \sum_{x\in B} x( f S^2(\delta_x)),
\end{equation*}
where $B$ is a basis of $\mathcal{A}$ and $\delta_x$ denotes the dual basis. In a finite dimensional  cosemisimple weak Hopf algebra, this expression can be re-expressed using the matrix $b(x)$:
\begin{equation*}
  \hat{L}(f) = \sum_{x\in B} \sum_{a\in \irr(\mathcal{A}^*)} \tr \{ b_a(x) \cdot \psi_{a}(f) \cdot \psi_a ( S^2(\delta_x))\}
\end{equation*}
Using now \cref{eq:S2}, we can also write
\begin{equation*}
  \hat{L}(f) =  \sum_{a\in \irr(\mathcal{A}^*)} \sum_{x\in B} \tr \{ b_a(x) \cdot \psi_{a}(f) \cdot \psi_{a}(g)  \cdot \psi_a (\delta_x) \cdot \psi_a(g^{-1})\}
\end{equation*}
As we sum over $x$, the matrix $\psi_a (\delta_x)$ takes all possible values in $\End(W_a)$. Conversely, iterating over elements of $\bigoplus_{a\in \irr(\mathcal{A}^*)} \End(V_a)$ defines a basis in $\mathcal{A}^*$ through the representation $\psi = \bigoplus_a \psi_a$. Let us fix a basis $\ket{ai}$ in each $W_a$ ($i=1, \ldots, \mathrm{dim}(W_a)$), and choose the basis $B$ such that for all $x$, there is $a\in \irr(\mathcal{A}^*)$ and $i,j=1, \ldots, \mathrm{dim}(W_a)$ such that $\psi(\delta_x) = \ket{ai}\bra{aj}$. It is easy to see that then $b_a(x) = \ket{aj}\bra{ai}$. Therefore one can also write
\begin{equation*}
  \hat{L}(f) =  \sum_{a\in \irr(\mathcal{A}^*)} \sum_{ij} \tr \{ \ket{aj}\bra{ai} \cdot \psi_{a}(fg) \cdot \ket{ai}\bra{aj} \cdot \psi_a(g^{-1})\} = \sum_{a\in \irr(\mathcal{A}^*)} \sum_i \bra{ai} \psi_{a}(fg)  \ket{ai} \cdot \sum_j \bra{aj} \psi_a(g^{-1}) \ket{aj}.
\end{equation*}
Using the trace formula \cref{eq:tr_g_inv}, $\tr \psi_a(g^{-1}) = d_a \epsilon_{r_a}(1)$, $\hat{L}(f)$ can be further written as
\begin{equation*}
  \hat{L}(f) =  \sum_{a\in \irr(\mathcal{A}^*)} d_a \epsilon_{r_a}(1) \cdot \tr \{\psi_{a}(gf) \},
\end{equation*}
i.e.\ $\hat{L}$ is described by the boundary $\epsilon_{r_a}(1) \cdot d_a \psi_{a}(g)$, or equivalently, the MPO representation of $\hat{L}$ is given by
\begin{equation*}
  \phi^{\otimes n} \circ \Delta^{n-1} (\hat{L})  = \sum_{a\in \irr(\mathcal{A}^*)} d_{a} \cdot \epsilon_{r_a}(1) \cdot
  	  \begin{tikzpicture}
  	    \draw[virtual] (0.5,0) rectangle (5.5,-0.7);
       \node[tensor,label=below:$g$] (x) at (1,0) {};
  	    \foreach \x/\t in {2/1,3/2,5/n}{
  	      \node[tensor] (t\x) at (\x,0) {};
  	      \draw[->-] (t\x) --++ (0,0.5);
  	      \draw[-<-] (t\x) --++ (0,-0.5);
  	    }
  	    \node[fill=white] (dots) at (4,0) {$\dots$};
        \draw[virtual,->-] (t2) -- (x) node[midway,black,irrep,above] {$a$};
        \draw[virtual,->-] (t3) -- (t2) node[midway,black,irrep,above] {$a$};
        \draw[virtual,->-] (dots) -- (t3) node[midway,black,irrep,above] {$a$};
        \draw[virtual,->-] (t5) -- (dots) node[midway,black,irrep,above] {$a$};
        \draw[virtual,->-] (0.5,-0.7) -- (5.5,-0.7);
  	  \end{tikzpicture}\ .
\end{equation*}

\subsection{Pivotal weak Hopf algebras} \label{sec:pivotal_WHA}

In certain WHAs one can define a special cocentral element using the integral $\Lambda$ defined in the previous section. This element satisfies an equation reminiscent to the definition of an integral, but here, instead of the antipode, another operation appears. The WHAs where the construction works are called \emph{pivotal} WHAs, and their defining property is that the antipode is special: the square of the antipode is an inner automorphism of the algebra realized by a \emph{group-like} element (see the definition below). While such a structure seems to be restrictive, it is conjectured that any semisimple WHA is actually pivotal \cite{Etingof2002}.

\begin{definition}[Group-like elements]\label{def:grouplike}
  Let $\mathcal{A}$ be a pre-bialgebra. $G\in\mathcal{A}$ is group-like if it is invertible and
  \begin{equation*}
    \Delta(G) = (G\otimes G) \cdot \Delta(1) =\Delta(1) \cdot (G\otimes G).
  \end{equation*}
\end{definition}
Group-like elements of a WHA $\mathcal{A}$ form a group: if $G,H\in \mathcal{A}$ are group-like, then $GH$ is also group-like; $1\in \mathcal{A}$ is also group-like, and finally, if $G\in\mathcal{A}$, then $G^{-1}$ is group-like as well. In a WHA $S(G) = G^{-1}$ for any group-like $G$ \cite{Bohm1999}.

$\mathcal{A}^*$ is a pre-bialgebra if and only if $\mathcal{A}$ is a pre-bialgebra; the unit of $\mathcal{A}^*$ is $\epsilon$, and thus a linear functional $k\in \mathcal{A}^*$ is group-like if it is invertible and
\begin{equation*}
  \Delta(k) = (k\otimes k) \cdot \Delta(\epsilon) =\Delta(\epsilon) \cdot (k\otimes k).
\end{equation*}
In a cosemisimple WHA $\mathcal{A}$, group-like elements of $\mathcal{A}^*$ have a nice characterization using the graphical representation of the fusion tensors. An element $k\in\mathcal{A}^*$ is group-like if and only if it is invertible and
\begin{equation*}
  \sum_{c\mu}
  \tikzsetnextfilename{c2fa75b6-6964-4113-aca7-53a522a6ea80}
 \ .
\end{equation*}
Using now the orthogonality relations of the fusion tensors, we conclude that $k$ is group-like if and only if  for all $a,b,c\in \irr(\mathcal{A}^*)$ and $\mu=1\dots N_{ab}^c$,
\begin{equation*}
  %
 \ .
\end{equation*}
In the previous section (\cref{cor:T}) we have seen a group-like element: the element $g^2\in \mathcal{A}^*$ is group-like, where $g$ is defined by \cref{eq:x_def}. Conjugation by this element describes $S^4$. A \emph{pivotal} WHA is one where not only the fourth power, but also the square of the antipode is realized as a conjugation by a group-like element:
\begin{definition}[Pivotal WHA]
  A WHA $\mathcal{A}$ is pivotal if there is a group-like element $G$ such that for all $x\in\mathcal{A}$,
  \begin{equation*}
    S^2(x) = G \cdot x \cdot G^{-1}.
  \end{equation*}
  Such a group-like element $G$ is called a pivotal element of $\mathcal{A}$.
\end{definition}

If a WHA $\mathcal{A}$ is such that $\mathcal{A}^*$ is pivotal, then, as $S^2$ is an inner automorphism of $\mathcal{A}^*$  described by the linear functional $g$ defined in \cref{eq:x_def}, all pivotal elements $k$ of $\mathcal{A}^*$ are of the form $k=\xi g$ for some $\xi\in\mathcal{A}^*$ central element. That is, there are numbers $\xi_a$ such that for all $a$,
\begin{equation}
    \psi_a(k) = \xi_a
  \ .
\end{equation}
As $S(g) = g^{-1}$ and $S(k) = k^{-1}$, $\xi_{\bar{a}} = 1/\xi_a$. As $g^2$ is group-like (see \cref{lem:foldback}), the numbers $\xi_a$ satisfy $\left(\xi_a \xi_b/\xi_c\right)^2 = 1$ for all $a,b,c$ such that $N_{ab}^c\neq 0$.

Let us consider a WHA $\mathcal{A}$ such that $\mathcal{A}^*$ is pivotal and let us fix the numbers $\xi_a$ such that the element $k\in \mathcal{A}^*$ defined by $k_a = \xi_a g_a$ is pivotal. In this case, as
\begin{equation*}
    %
,
\end{equation*}
\cref{lem:foldback} takes the form
\begin{equation}\label{eq:foldback_pivotal}
  \sum_{\mu} \frac{\frac{d_a}{\xi_a} \cdot \frac{d_b}{\xi_b}}{\frac{d_c}{\xi_c}}
  \tikzsetnextfilename{9f2542f3-60a9-4ee3-a79e-0e8478ab2459}
  %
 \ .
\end{equation}
Let us now notice that the matrix $B_{ab}^c$ defined in \cref{cor:T} is proportional to the identity, more precisely, $B_{ab}^c = \xi_c/(\xi_a \xi_b) \cdot \id_{N_{ab}^c}$, and thus that $T_{ab}^c = N_{ab}^c\cdot \xi_c/(\xi_a \xi_b)$. We thus obtain that
\begin{equation}\label{eq:N_eigenvector}
  \sum_b N_{ab}^c \cdot \frac{d_b}{\xi_b} = \delta_{l_a l_c} \cdot \frac{d_{\bar{a}}}{\xi_{\bar{a}}} \cdot \frac{d_c}{\xi_c}.
\end{equation}
With these statements we can now prove the following theorem:
\begin{theorem}\label{thm:main}
  Let $\mathcal{A}$ be a finite dimensional weak Hopf algebra over $\mathbb{C}$ such that $\mathcal{A}^*$ is semisimple and pivotal with pivotal element $k\in\mathcal{A}^*$ such that $k_a = \xi_a g_a$. Then the element $\Omega\in \mathcal{A}$ defined by
  \begin{equation*}
    \phi^{\otimes n} \circ \Delta^{n-1} (\Omega)  = \sum_{a\in \irr(\mathcal{A}^*)} \frac{1}{\sum_{x:\, l_x = l_a} d_x^2 }\cdot \frac{d_a}{\xi_a}  \cdot
    \tikzsetnextfilename{fd27623e-cad2-4d0a-835a-14cc67ee452e}
    \begin{tikzpicture}
      \draw[virtual] (1.5,0) rectangle (5.5,-0.7);
      \foreach \x/\t in {2/1,3/2,5/n}{
        \node[tensor] (t\x) at (\x,0) {};
        \draw[->-] (t\x) --++ (0,0.5);
        \draw[-<-] (t\x) --++ (0,-0.5);
      }
      \node[fill=white] (dots) at (4,0) {$\dots$};
      \draw[virtual,->-] (t3) -- (t2) node[midway,black,irrep,above] {$a$};
      \draw[virtual,->-] (dots) -- (t3) node[midway,black,irrep,above] {$a$};
      \draw[virtual,->-] (t5) -- (dots) node[midway,black,irrep,above] {$a$};
      \draw[virtual,->-] (1.5,-0.7) -- (5.5,-0.7);
    \end{tikzpicture}\
  \end{equation*}
  is non-degenerate and cocommutative, it is a projector, and there exists a linear map $T:\mathcal{A}\to \mathcal{A}$ such that
  \begin{equation*}
    (1\otimes x) \cdot \Delta(\Omega) = (T(x)\otimes 1) \cdot \Delta(\Omega).
  \end{equation*}
  This\footnote{As $\Omega$ is non-degenerate, $T$ is uniquely defined, see \cref{sec:pulling_through_alg}.} map $T$ moreover satisfies
  \begin{equation*}
    \Delta\circ T = (T\otimes k^{-1} \otimes T) \circ \Delta_{\mathrm{op}}^2.
  \end{equation*}
\end{theorem}

The number $d_a/\xi_a$ is the \emph{quantum dimension} of the irrep class $a$. This number depends on the concrete choice of $k$ (equivalently, on the choice of the numbers $\xi_a$).  Before proceeding to the proof, note that $\Omega = (g^{-1}\otimes\id)\circ\Delta(\Lambda)$, where $\Lambda$ is defined in \cref{eq:spec_integral}. From this, using that $g$ is group-like and the properties of $\Lambda$, simple (but tedious) algebraic calculations show all the desired properties. In particular, one can obtain that $T = (S\otimes g) \circ \Delta$. Instead of this algebraic calculation, below we provide a proof based on the graphical notation.

\begin{proof}[Proof of \cref{thm:main}]
  By definition, $\Omega$ is cocommutative and as for all $a\in \irr(\mathcal{A}^*)$, $d_a\neq 0$, it is also non-degenerate. Let us now show that $\Omega$ is a projection. As in the proof of \cref{thm:wha_special_integral}, note that $\Omega$ can be written as $N^{-1} \hat{\Omega} = \hat{\Omega} N^{-1}$, where $\hat{\Omega}$ is given by
  \begin{equation}
    \hat{\Omega} = \sum_{a\in \irr(\mathcal{A}^*)}  \frac{d_a}{\xi_a}  \cdot \tau_a,
  \end{equation}
  and $N$ is given by \cref{eq:normalization}. Then $\Omega$ is a projector if and only if $\hat{\Omega}$ satisfies $\hat{\Omega}^2 = N \hat{\Omega}$. Let us calculate $\hat{\Omega}^2$:
  \begin{equation*}
    \hat{\Omega}^2 = \sum_{ab} \frac{d_ad_b}{\xi_a\xi_b}  \cdot \tau_a\tau_b = \sum_{abc} \frac{d_ad_b}{\xi_a\xi_b}  \cdot N_{ab}^c \tau_c =\sum_{ac} \frac{d_a}{\xi_a} \cdot \left(\sum_b N_{ab}^c \frac{d_b}{\xi_b}\right)  \cdot \tau_c ,
  \end{equation*}
  and thus using \cref{eq:N_eigenvector}, we obtain that
  \begin{equation*}
    \hat{\Omega}^2 = \sum_{ac} \delta_{l_al_c} \frac{d_ad_{\bar{a}}}{\xi_a\xi_{\bar{a}}} \cdot \frac{d_c}{\xi_c}\cdot \tau_c = \sum_{c} \left(\sum_{a:l_a = l_c} d_a^2\right) \cdot \frac{d_c}{\xi_c}\cdot \tau_c = N \hat{\Omega},
  \end{equation*}
  where in the second equation we have used that $\xi_{\bar{a}} = \xi_a^{-1}$ and that $d_a = d_{\bar{a}}$. We have thus obtained that $\Omega = N^{-1} \hat{\Omega}$ is a projector.

  Let us now check that $ (1\otimes x) \cdot \Delta(\Omega) = (T(x)\otimes 1) \cdot \Delta(\Omega) $ holds. Similar to the proof in \cref{thm:wha_special_integral}, we start with \cref{eq:foldback_pivotal}:
  \begin{equation*}
    \sum_{\mu} \frac{d_a/\xi_a \cdot d_b/\xi_b}{d_c/\xi_c}
    \tikzsetnextfilename{642f516f-dd27-418b-a495-1e7c5a16d2bd}
 \ .
  \end{equation*}
  This implies that for all $x\in \mathcal{A}$,
  \begin{equation*}
    \sum_{abc\mu} \ \frac{d_a}{\xi_a} \
    \tikzsetnextfilename{67ec711a-252e-4375-b88a-0613d141d939}
    %
  \ .
  \end{equation*}
  We can now use \cref{eq:MPO_tensor_product} to obtain
  \begin{equation*}
    \sum_{a,b} \ \frac{d_a}{\xi_a}\
    \tikzsetnextfilename{ca9ec046-c268-4343-ae99-dfd4182291e6}
    %
  \ ,
  \end{equation*}
or $ (1\otimes x) \cdot \Delta(\Omega) = (T(x)\otimes 1) \cdot \Delta(\Omega)$, where $T$ is defined by the formula
\begin{equation}\label{eq:T}
  \phi(T(x)) = \sum_a \xi_a \cdot \frac{w_a}{d_a} \
  %
 \ .
\end{equation}
This means, in particular, that
\begin{equation*}
  %
 \ =
  \sum_x \phi\circ T(x) \otimes \psi_a (\delta_x) \ .
\end{equation*}
Let us now consider the MPO representation of $\Delta\circ T(x)$ and, using \cref{eq:x_def}, observe that
\begin{equation*}
  \sum_a \xi_a \cdot \frac{w_a}{d_a} \cdot
  %
 \ ,
\end{equation*}
or, as $\phi$ is w.l.o.g.\ injective (notice that on the r.h.s.\ the arrows on the virtual index  are reversed),
\begin{equation*}
  T\circ \Delta(x) = (T\otimes k^{-1} \otimes T) \circ \Delta_{\mathrm{op}}^{2} (x).
\end{equation*}
\end{proof}

Note that the graphical representation of the action of $T$, \cref{eq:T}, is very similar to that of $S$:
\begin{equation*}
  \phi(T(x)) = \sum_a \xi_a \
  %
 \ .
\end{equation*}
The difference between $T(x)$ and $S(x)$ is that in $T(x)$ the matrix $\xi_a Z_{\bar{a}}$ appears instead of $Z_a$. As
\begin{equation*}
  \xi_a \cdot
  %
 \
\end{equation*}
we directly obtain that $T(x) = (k^{-1} \otimes \id)\circ \Delta(S(x))$,  or, by using that $S$ is an anti-homomorphism and that $k^{-1}\circ S = k$, that $T(x) = (S\otimes k) \circ \Delta (x)$. The inverse of this relation, as it can be see from \cref{eq:T}, is $S(x) = (k^{-1}\otimes \id)\circ \Delta\circ T(x)$, or $S(x) = (T\otimes k^{-1})\circ \Delta(x)$.

\subsection{Spherical weak Hopf algebras}

In this section we define spherical weak Hopf algebras as pivotal weak Hopf algebras satisfying an additional property. Semisimple spherical weak Hopf algebras are such that their representation category is a spherical multi-fusion category. We specialize \cref{thm:main} to the case where $\mathcal{A}$ is not only pivotal, but also spherical.

\begin{definition}[Spherical WHA]
  A pivotal WHA  $\mathcal{A}$ is called spherical if the following two conditions hold. First, that for any irrep $\delta$ from the vacuum, $\epsilon(1_\delta) \neq 0$, where $1_\delta$ is the projector onto the irrep sector $\delta$. Second,  it has a pivotal element $G$ such that for all $\alpha\in \irr(\mathcal{A})$,
  \begin{equation*}
    \frac{t_\alpha(G)}{\epsilon(1_{\lambda_\alpha})}   = \frac{t_\alpha(G^{-1})}{\epsilon(1_{\rho_\alpha})},
  \end{equation*}
  where $t_\alpha$ is the irrep character of the irrep class $\alpha$, $\lambda_\alpha$ is the unique irrep from the vacuum such that $t_{\lambda_\alpha} t_\alpha \neq 0$, and $\rho_\alpha$ is the unique irrep from the vacuum such that $ t_\alpha t_{\rho_\alpha} \neq 0$. Such a pivotal element is called a spherical element of $\mathcal{A}$.
\end{definition}

Let us now consider a WHA over $\mathbb{C}$ such that $\mathcal{A}^*$ is semisimple and spherical. Let $k$ be a spherical element of $\mathcal{A}^*$ with
\begin{equation*}
    \psi_a(k) = \xi_a
  \begin{tikzpicture}[baseline=-1mm]
    \node[tensor,midway,label=below:$g_a$] (t) {};
    \draw[virtual,-<-] (t)--++(0.5,0) node[irrep] {$a$};
    \draw[virtual,->-] (t)--++(-0.5,0) node[irrep] {$a$};
  \end{tikzpicture}  \quad \text{and} \quad
  \psi_a(k^{-1}) = \frac{1}{\xi_a}
  \begin{tikzpicture}[baseline=-1mm]
    \node[tensor,midway,label=below:$g_a^{-1}$] (t) {};
    \draw[virtual,-<-] (t)--++(0.5,0) node[irrep] {$a$};
    \draw[virtual,->-] (t)--++(-0.5,0) node[irrep] {$a$};
  \end{tikzpicture}  \ .
\end{equation*}
The traces of $\psi_{a}(k)$ and $\psi_a(k^{-1})$ can be then evaluated using \cref{eq:tr_g,eq:tr_g_inv}:
\begin{gather*}
  \tau_a (k^{-1}) = \tr \{\psi_a(k^{-1})\} = \frac{1}{\xi_a} \cdot \tr\{\psi_a(g^{-1})\} = \frac{d_{a}}{\xi_a} \cdot \epsilon_{r_{a}} (1), \\
  \tau_a(k ) = \tr \{\psi_a (k)\} = \xi_a \cdot \tr\{\psi_a(g)\} = \xi_a d_{a} \cdot \epsilon_{r_{\bar{a}}} (1).
\end{gather*}
Sphericity of $k$ implies that $\xi_a^2=1$ for all $a\in\irr(\mathcal{A}^*)$, i.e.\ that $\xi_a = \pm 1$. Together with the fact that $\xi_{\bar{a}} = 1/\xi_a$, this implies that $\xi_a = \xi_{\bar{a}}$.

Let us define now $\Omega$ as in \cref{thm:main}, assuming that the pivotal element $k\in\mathcal{A}^*$ is  spherical. Then $\Omega$, using that $\xi_a = \xi_{\bar{a}}$, satisfies additionally that
\begin{equation*}
  (T\otimes k^{-1})\circ \Delta(\Omega) = S(\Omega) = \Omega.
\end{equation*}
We have thus seen that
\begin{theorem}\label{thm:main_spherical}
  Let $\mathcal{A}$ be a finite dimensional weak Hopf algebra over $\mathbb{C}$ such that $\mathcal{A}^*$ is semisimple and spherical with spherical element $k\in\mathcal{A}^*$ such that $\psi_a(k) = \xi_a \psi_{a}(g)$. Then the element $\Omega\in \mathcal{A}$ defined by
  \begin{equation*}
    \phi^{\otimes n} \circ \Delta^{n-1} (\Omega)  = \sum_{a\in \irr(\mathcal{A}^*)} \frac{1}{\sum_{x:\, l_x = l_a} d_x^2 }\cdot \frac{d_a}{\xi_a}  \cdot
    \tikzsetnextfilename{fd27623e-cad2-4d0a-835a-14cc67ee452e}
    \begin{tikzpicture}
      \draw[virtual] (1.5,0) rectangle (5.5,-0.7);
      \foreach \x/\t in {2/1,3/2,5/n}{
        \node[tensor] (t\x) at (\x,0) {};
        \draw[->-] (t\x) --++ (0,0.5);
        \draw[-<-] (t\x) --++ (0,-0.5);
      }
      \node[fill=white] (dots) at (4,0) {$\dots$};
      \draw[virtual,->-] (t3) -- (t2) node[midway,black,irrep,above] {$a$};
      \draw[virtual,->-] (dots) -- (t3) node[midway,black,irrep,above] {$a$};
      \draw[virtual,->-] (t5) -- (dots) node[midway,black,irrep,above] {$a$};
      \draw[virtual,->-] (1.5,-0.7) -- (5.5,-0.7);
    \end{tikzpicture}\
  \end{equation*}
  is non-degenerate and cocommutative, it is a projector, and there exists a linear map $T:\mathcal{A}\to \mathcal{A}$ such that
  \begin{gather*}
    (1\otimes x) \cdot \Delta(\Omega) = (T(x)\otimes 1) \cdot \Delta(\Omega), \\
    \Delta\circ T = (T\otimes k^{-1}\otimes T) \circ \Delta_{\mathrm{op}}^2.
  \end{gather*}
  Moreover, $\Omega$ satisfies
  \begin{equation*}
    (T\otimes k^{-1})\circ \Delta(\Omega) = S(\Omega) = \Omega.
  \end{equation*}
\end{theorem}

\subsection{\texorpdfstring{$C^*$}{C*}-weak Hopf algebras} \label{sec:C_star_WHA}

In this section we define $C^*$-weak Hopf algebras. We show how the $*$-operation acts on the representing MPOs. We then show that a $C^*$-weak Hopf algebra is spherical, and thus \cref{thm:main_spherical} applies.

\begin{definition}[$C^*$-WHA]
  A finite dimensional pre-bialgebra $\mathcal{A}$ over $\mathbb{C}$ is a $*$-pre-bialgebra if there is an anti-linear map $*:\mathcal{A}\to\mathcal{A}$ such that it is an involution ($x^{**} = x$), anti-homomorphism ($(xy)^* = y^* x^*$) and cohomomorphism ($\Delta \circ * = (*\otimes *)\circ \Delta$). It is a $C^*$-pre-bialgebra, if it is a $*$-pre-bialgebra and $\mathcal{A}$ has a faithful $*$-representation. A pre-bialgebra $\mathcal{A}$ is a $*$-WHA if it is both a $*$-pre-bialgebra and a WHA, and it is a $C^*$-WHA if it is both a $C^*$-pre-bialgebra and a WHA.
\end{definition}

An equivalent characterization of a finite dimensional $C^*$-WHA is that it is a semisimple $*$-WHA such that it has a complete set of irreps $\phi_\alpha\; (\alpha\in \irr(\mathcal{A}))$ that are $*$-representations, i.e.\ $\phi_\alpha(x^*) = \phi_\alpha(x)^\dagger$. The dual of a $C^*$-WHA is also a $C^*$-WHA with $*$-operation defined by \cite{Bohm1999}
\begin{equation}\label{eq:dual_star}
  f^*(x) = \overline{f\circ *\circ S(x)},
\end{equation}
where $\overline{\lambda}$ denotes the complex conjugate of $\lambda$ for all $\lambda\in\mathbb{C}$. As it is a $C^*$-WHA, $\mathcal{A}^*$ also possesses a complete set of $*$-representations. In this section, $\psi_a\; (a\in \irr(\mathcal{A}^*))$ will always denote such a complete set of irreps.

Let us now observe that the $*$-operation of $\mathcal{A}$ brings the character of $\psi_a$ into the character of $\psi_{\bar{a}}$:
\begin{proposition}
  Let $\mathcal{A}$ be a $C^*$-WHA, and $\tau_a\in \mathcal{A}$ $(a\in \irr(\mathcal{A}^*))$ be the irrep characters of $\mathcal{A}^*$. Then $\tau_a^* = \tau_{\bar{a}}$ for all $a$.
\end{proposition}

\begin{proof}
  Let us evaluate $\tau_a^*$ on a linear functional $f$:
  \begin{equation*}
    \tau_a^*(f) = f(\tau_a^*) = \overline{f^* (S^{-1}(\tau_a))} = \overline{f^*(\tau_{\bar{a}})} =  \overline{\tau_{\bar{a}} (f^*)},
  \end{equation*}
  where in the second equation we have used \cref{eq:dual_star} with $x=S^{-1}(\tau_a)$. Finally note that, as $\psi_a$ is a $*$-representation,
  \begin{equation*}
    \tau_{\bar{a}}(f^*) = \tr \psi_{\bar{a}}(f^*) = \tr \psi_{\bar{a}}  (f)^\dagger = \overline{\tr \psi_{\bar{a}}(f)} =  \overline{\tau_{\bar{a}}(f)},
  \end{equation*}
  and thus for all $f\in\mathcal{A}^*$, $\tau_a^*(f) = \tau_{\bar{a}}(f)$.
\end{proof}

As the MPO representation of $\tau_a$ is the TI MPO defined by the injective block $a$ of the MPO tensor representing the $C^*$-WHA, this proposition states that the injective MPO blocks are permuted under the $*$-operation:
\begin{equation*}
  \left(\
  \ .
\end{equation*}
On the l.h.s.\ of this equation the dagger can be taken component-wise. More precisely, let us define a gray MPO tensor as
\begin{equation*}
  %
 \ =
  \sum_x \phi(x^*) \otimes \psi_a(\delta_x^*) =
  \sum_x \phi(x)^\dagger \otimes \psi_a(\delta_x)^\dagger.
\end{equation*}
Notice that this MPO tensor is oriented in the opposite direction of the original (black) MPO tensor. By construction, the MPO defined by this MPO tensor is the Hermitian conjugate of the MPO describing $\tau_a$:
\begin{equation*}
  %
  \ .
\end{equation*}
In fact, this equation holds with arbitrary boundary condition as well, i.e.\ the MPO representation of $x^*$ can be written as
\begin{equation*}
  \sum_a
  %
\ .
\end{equation*}
Let us finally note that the gray and the black tensors can be related to each other using \cref{eq:dual_star}. First note that there are two ways of expressing $y^*$:
\begin{equation*}
  \sum_x x \cdot \delta_x(y^*) = y^* = \sum_x x^* \cdot \overline{\delta_x(y)} = \sum_x x^* \cdot \delta_x^*\circ S^{-1}(y^*),
\end{equation*}
where in the last equation we have used the fact that   the conjugation can be expressed with the $*$ operation of $\mathcal{A}$ and $\mathcal{A}^*$ as follows:
\begin{equation*}
  \overline{f(x)} = \overline{f\circ * \circ S \circ S^{-1} \circ * (x)} = f^*\circ S^{-1}(x^*).
\end{equation*}
As this equation holds for all $y$, we have obtained that
\begin{equation*}
  \sum_ x x^* \otimes \delta_x^* = \sum_x x \otimes S(\delta_x),
\end{equation*}
or graphically,
\begin{equation*}
  \begin{tikzpicture}[baseline=-1mm]
    \node[tensor,gray!50!white] (t) at (0,0) {};
    \draw[virtual,->-] (t)--++(0.7,0)  node[midway,irrep,black] {$a$};
    \draw[virtual,-<-] (t)--++(-0.7,0) node[midway,irrep,black] {$a$};
    \draw[->-] (t) --++ (0,0.5);
    \draw[-<-] (t) --++ (0,-0.5);
  \end{tikzpicture} \ =
  \begin{tikzpicture}[baseline=-1mm]
    \node[tensor] (t) at (0,0) {};
    \node[tensor,label=below:$Z_{a}^{-1}$] (zinv) at (0.6,0) {};
    \node[tensor,label=below:$Z_{a}$] (z) at (-0.6,0) {};
    \draw[virtual,->-] (zinv)--++(0.6,0) node[midway,irrep,black] {$a$};
    \draw[virtual,->-] (zinv)--(t) node[midway,irrep,black] {$\bar{a}$};
    \draw[virtual,->-] (t)--(z) node[midway,irrep,black] {$\bar{a}$};
    \draw[virtual,-<-] (z)--++(-0.6,0) node[midway,irrep,black] {$a$};
    \draw[->-] (t) --++ (0,0.5);
    \draw[-<-] (t) --++ (0,-0.5);
  \end{tikzpicture} \ .
\end{equation*}
Note finally that then $b(x)$ and $b(x^*)$ is also related:
\begin{equation*}
  \begin{tikzpicture}[baseline=-1mm]
    \node[tensor,midway,label=below:$b_{\bar{a}}(x^*)$] (t) {};
    \draw[virtual,-<-] (t)--++(0.7,0) node[midway,irrep,black] {$\bar{a}$};
    \draw[virtual,->-] (t)--++(-0.7,0) node[midway,irrep,black] {$\bar{a}$};
  \end{tikzpicture} \ = \
  \begin{tikzpicture}[baseline=-1mm]
    \node[tensor,label=below:$b_{a}(x)^\dagger$] (t) at (0,0) {};
    \node[tensor,label=below:$Z_{a}$] (z) at (0.9,0) {};
    \node[tensor,label=below:$Z_{a}^{-1}$] (zinv) at (-0.9,0) {};
    \draw[virtual,-<-] (z)--++(0.7,0) node[midway,irrep,black] {$\bar{a}$};
    \draw[virtual,-<-] (z)--(t) node[midway,irrep,black] {$a$};
    \draw[virtual,-<-] (t)--(zinv) node[midway,irrep,black] {$a$};
    \draw[virtual,->-] (zinv)--++(-0.7,0) node[midway,irrep,black] {$\bar{a}$};
  \end{tikzpicture} \ .
\end{equation*}

Let us now prove the well-known result that a $C^*$-WHA is pivotal, and in fact, spherical \cite{Bohm1999,Etingof2002}:

\begin{proposition}
  Let $\mathcal{A}$ be a $C^*$-WHA. Then $\mathcal{A}^*$ is also a $C^*$-WHA, it is spherical and the linear functional $g\in\mathcal{A}^*$ defined in \cref{eq:x_def} is a spherical element of it that is also positive.
\end{proposition}

\begin{proof}
  As mentioned above, we can fix the irrep representatives $\psi_a$ of $\mathcal{A}^*$ to be $*$ representations. This implies, as the $*$ operation of $\mathcal{A}^*$ is a cohomomorphism, $\Delta = (* \otimes *) \circ \Delta\circ *$,  that
    \begin{equation*}
      \sum_{\mu} \left(
      \tikzsetnextfilename{25435fbd-fd3e-4009-b5ae-f552e3eae421}
 \ .
    \end{equation*}
    This equation implies that for all $a,b,c\in \irr(\mathcal{A}^*)$ there exist a matrix $A_{ab}^c$ such that the dagger of the fusion tensor, denoted by gray tensors, can be expressed as
    \begin{equation*}
      %
 \ .
    \end{equation*}
  These matrices $A_{ab}^c$ are in fact positive, as using the orthogonality relations, they are in the form
    \begin{equation*}
      \left(A_{ab}^c\right)_{\mu\nu} = \frac{1}{D_c} \cdot
      %
 \ =
      \sum_\kappa \overline{\left(o_{ab}^c\right)_{\kappa\mu}} \left(o_{ab}^c\right)_{\kappa\nu} =
      \sum_\kappa \left(\left(o_{ab}^c\right)^\dagger\right)_{\mu\kappa} \left(o_{ab}^c\right)_{\kappa\nu} = \left(\left(o_{ab}^c\right)^\dagger o_{ab}^c\right)_{\mu\nu}.
    \end{equation*}
    Therefore changing the fusion tensors to
    \begin{equation*}
      \hat{V}_{ab}^{c\mu} = \sum_{\kappa} \left(\left(o_{ab}^c\right)^{-1}\right)_{\mu\nu} V_{ab}^{c\nu} \quad \text{and} \quad \left(\hat{W}_{ab}^c\right)_{\mu\nu} = \sum_{\kappa} \left(o_{ab}^c\right)_{\mu\kappa} \left(V_{ab}^c\right)_{\kappa\nu},
    \end{equation*}
    we obtain $\left(\hat{V}_{ab}^{c\mu}\right)^\dagger = \hat{W}_{ab}^{c\mu}$. That is, w.l.o.g.\ we can assume that the fusion tensors are the hermitian conjugates of each other,
    \begin{equation}\label{eq:fusion_dagger}
      \left(
      \begin{tikzpicture}[baseline=-1mm]
        \node[fusion tensor] (v) at (0,0) {};
        \draw[virtual,->-] let \p1=(v.west) in (\x1,0.25)--++(-0.5,0) node[midway,black,irrep] {$a$};
        \draw[virtual,->-] let \p1=(v.west) in (\x1,-0.25)--++(-0.5,0) node[midway,black,irrep] {$b$};
        \draw[virtual,-<-] let \p1=(v.east) in (\x1,0)--++(0.5,0) node[midway,black,irrep] {$c$};
        \node[anchor=south,inner sep=2pt] at (v.north) {$\mu$};
      \end{tikzpicture}\right)^{\dagger} =
      \begin{tikzpicture}[baseline=-1mm]
        \node[fusion tensor] (w) at (1.5,0) {};
        \draw[virtual,-<-] let \p1=(w.east) in (\x1,0.25)--++(0.5,0) node[midway,black,irrep] {$a$};
        \draw[virtual,-<-=0.5] let \p1=(w.east) in (\x1,-0.25)--++(0.5,0) node[midway,black,irrep] {$b$};
        \draw[virtual,->-] let \p1=(w.west) in (\x1,0)--++(-0.5,0) node[midway,black,irrep] {$c$};
        \node[anchor=south,inner sep=2pt] at (w.north) {$\mu$};
      \end{tikzpicture} \quad \text{and} \quad
      \left(
      \begin{tikzpicture}[baseline=-1mm]
        \node[fusion tensor] (w) at (1.5,0) {};
        \draw[virtual,-<-] let \p1=(w.east) in (\x1,0.25)--++(0.5,0) node[midway,black,irrep] {$a$};
        \draw[virtual,-<-=0.5] let \p1=(w.east) in (\x1,-0.25)--++(0.5,0) node[midway,black,irrep] {$b$};
        \draw[virtual,->-] let \p1=(w.west) in (\x1,0)--++(-0.5,0) node[midway,black,irrep] {$c$};
        \node[anchor=south,inner sep=2pt] at (w.north) {$\mu$};
      \end{tikzpicture}\right)^{\dagger} =
      \begin{tikzpicture}[baseline=-1mm]
        \node[fusion tensor] (v) at (0,0) {};
        \draw[virtual,->-] let \p1=(v.west) in (\x1,0.25)--++(-0.5,0) node[midway,black,irrep] {$a$};
        \draw[virtual,->-] let \p1=(v.west) in (\x1,-0.25)--++(-0.5,0) node[midway,black,irrep] {$b$};
        \draw[virtual,-<-] let \p1=(v.east) in (\x1,0)--++(0.5,0) node[midway,black,irrep] {$c$};
        \node[anchor=south,inner sep=2pt] at (v.north) {$\mu$};
      \end{tikzpicture} \ .
    \end{equation}

    Let us now investigate what restrictions the $C^*$ structure of $\mathcal{A}^*$ imposes on the representation of the unit of the algebra. The unit $1\in\mathcal{A}$ is invariant under both the $*$-operation as well as under the action of the antipode, therefore
    \begin{equation*}
      1(f) = 1^*(f) = \overline{1(S(f^*))} = \overline{1(f^*)}.
    \end{equation*}
    Let us write  $\ket{w_e} = \tikz{\draw[virtual,->-] (0,0) node[t,black] {} --++(-0.5,0) node[midway,black,irrep] {$e$};} $ and $ \bra{v_e} = \tikz{\draw[virtual,-<-] (0,0) node[t,black] {} --++(0.5,0) node[midway,black,irrep] {$e$};}$, then this equation reads
    \begin{equation*}
      \sum_{e\in E} \bra{v_e} \psi_e (f) \ket{w_e} =
      \sum_{e\in E} \overline{\bra{v_e} \psi_e (f^*) \ket{w_e}} =
      \sum_{e\in E} \bra{w_e} \psi_e (f^*)^\dagger \ket{v_e} =
      \sum_{e\in E} \bra{w_e} \psi_e (f) \ket{v_e},
    \end{equation*}
    where in the last equation we have used that $\psi_e$ is a $*$-representation. As this equation holds for all $f$, we have obtained that
    \begin{equation}\label{eq:unit_dagger}
    \left(\tikz[baseline=-1mm]{\draw[virtual,->-] (0,0) node[t,black] {} --++(-0.5,0) node[midway,black,irrep] {$e$};}\right)^\dagger = \tikz[baseline=-1mm]{\draw[virtual,-<-] (0,0) node[t,black] {} --++(0.5,0) node[midway,black,irrep] {$e$};} \ .
    \end{equation}
  Using that the fusion tensors and the vectors representing the unit are self-adjoint, we can now prove that the matrices $Z_a$ can be chosen such that $\left(Z_{\bar{a}}^{-1}\right)^\dagger = Z_a$. To see that, let us first take  the dagger of \cref{eq:w_a} using \cref{eq:fusion_dagger,eq:unit_dagger}:
  \begin{equation*}
  \  .
  \end{equation*}
  This implies that there is $\lambda_a\in\mathbb{C}$ such that $\left(Z_{\bar{a}}^{-1}\right)^\dagger = \lambda_a Z_a$ and $Z_a^\dagger = \bar{w}_a/(w_a \lambda_a) \cdot Z_{\bar{a}}^{-1}$. Changing $a$ to $\bar{a}$ in the first equation, we obtain that $\left(Z_{a}^{-1}\right)^\dagger = \lambda_{\bar{a}} Z_{\bar{a}}$, or, after rearranging, $\left(Z_{\bar{a}}^{-1}\right)^\dagger = \overline{\lambda_{\bar{a}}} Z_a$, and thus $\lambda_a = \overline{\lambda_{\bar{a}}}$. This implies that if $a$ is such that $\bar{a}\neq a$, then $\lambda_a = \mu_a \overline{\mu_{\bar{a}}}$  can be solved, e.g. by $\mu_a = \lambda_a$ and $\mu_{\bar{a}}=1$. With this choice, $\lambda_{\bar{a}} = \mu_{\bar{a}} \overline{\mu_{a}}$ also holds.  If $\bar{a} = a$, one has to be more careful. In this case, $\lambda_a = \overline{\lambda_{\bar{a}}}$ implies that $\lambda_a$ is real. To solve $\lambda_a = \mu_a \overline{\mu_{\bar{a}}} = |\mu_a|^2$, we have to show that $\lambda_a$ is not only real, but also positive. To show that, note that $\left(Z_{a}^{-1}\right)^\dagger = \lambda_a Z_a$, or $\id_a = \lambda_a Z_a Z_a^\dagger$. As both $\id_a$, and $Z_a Z_a^\dagger$ are positive, this implies that $\lambda_a$ is positive, and thus $\lambda_a = \mu_a \overline{\mu_{\bar{a}}}$ can be solved as well.  Let $\hat{Z}_a = \mu_a Z_a$, then $\left(\hat{Z}_{\bar{a}}^{-1}\right)^\dagger = \hat{Z}_a$. This means that we can assume w.l.o.g.\ that $\left(Z_{\bar{a}}^{-1}\right)^\dagger = Z_a$.

  Let us show now that the linear functional $g\in\mathcal{A}^*$ defined in \cref{eq:x_def} is a pivotal element that is positive and spherical. First, the equation $\left(Z_{\bar{a}}^{-1}\right)^\dagger = Z_a$ implies that both $w_a$ and $\psi_a(g)$ are positive, and thus $g$ is positive as well. Let us show now that the matrix $B_{ab}^c$ defined in \cref{cor:T} is positive as well. This matrix can be obtained by
  \begin{equation*}
    \left(B_{ab}^c\right)_{\mu\nu} \cdot g_c =
    \begin{tikzpicture}
      \node[fusion tensor] (v) at (1,0) {};
      \node[fusion tensor] (w) at (0,0) {};
      \draw[virtual,-<-] (v.east) --++(0.5,0) node[midway,black,irrep] {$c$};
      \draw[virtual,->-] (w.west) --++(-0.5,0) node[midway,black,irrep] {$c$};
      \draw[virtual,-<-] let \p1=(w.east) in (\x1,0.25) --++(0.5,0) node[midway,black,irrep] {$a$} coordinate (c);
      \node[tensor,label=above:$g$] (Lambda) at (c) {};
      \draw[virtual,->-] let \p1=(v.west) in (\x1,0.25) -- (Lambda)  node[midway,black,irrep] {$a$};
      \draw[virtual,-<-] let \p1=(w.east) in (\x1,-0.25) --++(0.5,0) node[midway,black,irrep] {$b$} coordinate (c);
      \node[tensor,label=below:$g$] (Lambda) at (c) {};
      \draw[virtual,->-] let \p1=(v.west) in (\x1,-0.25) -- (Lambda) node[midway,black,irrep] {$b$};
      \node[anchor=north,inner sep=3pt,font=\tiny] at (v.south) {$\mu$};
      \node[anchor=north,inner sep=3pt,font=\tiny] at (w.south) {$\nu$};
    \end{tikzpicture},
  \end{equation*}
  i.e.\ it is a positive matrix. As $B_{ab}^c$ is positive and it squares to the identity, it is the identity, and thus $g$ is a (positive) pivotal element. Note that it is spherical as well, because
  \begin{equation*}
    \tr \psi_a(g) = d_a \cdot \epsilon_{r_a}(1) =  d_a \cdot \epsilon_{r_{\bar{a}}}(1) =  \tr\psi_a(g^{-1}).
  \end{equation*}
  We have thus derived the well-known result that in a $C^*$-WHA there is a positive spherical element.
\end{proof}

Let us now use this positive pivotal element in the  construction in \cref{thm:main_spherical}. With this choice, the resulting element $\Omega$ has MPO representation
\begin{equation*}
  \phi^{\otimes n} \circ \Delta^{n-1} (\Omega)  = \sum_{a\in \irr(\mathcal{A}^*)} \frac{1}{\sum_{x:\, l_x = l_a} d_x^2 }\cdot \frac{d_a}{\xi_a}  \cdot d_a  \cdot
  	  \begin{tikzpicture}
  	    \draw[virtual] (1.5,0) rectangle (5.5,-0.7);
  	    \foreach \x/\t in {2/1,3/2,5/n}{
  	      \node[tensor] (t\x) at (\x,0) {};
  	      \draw[->-] (t\x) --++ (0,0.5);
  	      \draw[-<-] (t\x) --++ (0,-0.5);
  	    }
  	    \node[fill=white] (dots) at (4,0) {$\dots$};
        \draw[virtual,->-] (t3) -- (t2) node[midway,black,irrep,above] {$a$};
        \draw[virtual,->-] (dots) -- (t3) node[midway,black,irrep,above] {$a$};
        \draw[virtual,->-] (t5) -- (dots) node[midway,black,irrep,above] {$a$};
        \draw[virtual,->-] (1.5,-0.7) -- (5.5,-0.7);
  	  \end{tikzpicture}\ ,
\end{equation*}
where all $d_a$ are positive and $d_a = d_{\bar{a}}$. As  $d_a$ are positive,  $\Omega$ is a positive linear functional on $\mathcal{A}^*$: for any positive $f$ the representing matrix $\psi_a(f)$ is positive, and thus $\Omega(f) = \sum_a d_a \cdot \tr \psi_a(f)\geq 0$.  Similarly, $d_a = d_{\bar{a}}$ implies that $\Omega^* = \Omega$, and therefore, as $\Omega$ is also a projector, it is a positive element of $\mathcal{A}$.

We have thus seen that the following theorem holds:

\begin{theorem}\label{thm:main_c_star}
  Let $\mathcal{A}$ be a finite dimensional $C^*$-weak Hopf algebra over $\mathbb{C}$. Then $\mathcal{A}^*$ is semisimple and spherical with positive spherical element $g\in\mathcal{A}^*$ defined in \cref{eq:x_def}. Then the element $\Omega\in \mathcal{A}$ defined by
  \begin{equation*}
    \phi^{\otimes n} \circ \Delta^{n-1} (\Omega)  = \sum_{a\in \irr(\mathcal{A}^*)} \frac{1}{\sum_{x:\, l_x = l_a} d_x^2 }\cdot d_a  \cdot
    \tikzsetnextfilename{fd27623e-cad2-4d0a-835a-14cc67ee452e}
    \begin{tikzpicture}
      \draw[virtual] (1.5,0) rectangle (5.5,-0.7);
      \foreach \x/\t in {2/1,3/2,5/n}{
        \node[tensor] (t\x) at (\x,0) {};
        \draw[->-] (t\x) --++ (0,0.5);
        \draw[-<-] (t\x) --++ (0,-0.5);
      }
      \node[fill=white] (dots) at (4,0) {$\dots$};
      \draw[virtual,->-] (t3) -- (t2) node[midway,black,irrep,above] {$a$};
      \draw[virtual,->-] (dots) -- (t3) node[midway,black,irrep,above] {$a$};
      \draw[virtual,->-] (t5) -- (dots) node[midway,black,irrep,above] {$a$};
      \draw[virtual,->-] (1.5,-0.7) -- (5.5,-0.7);
    \end{tikzpicture}\
  \end{equation*}
  is a positive non-degenerate trace-like linear functional on $\mathcal{A}^*$, it is an orthogonal projector, and there exists a linear map $T:\mathcal{A}\to \mathcal{A}$ such that
  \begin{gather*}
    (1\otimes x) \cdot \Delta(\Omega) = (T(x)\otimes 1) \cdot \Delta(\Omega), \\
    \Delta\circ T = (T\otimes g\otimes T) \circ \Delta_{\mathrm{op}}^2.
  \end{gather*}
  Moreover, $\Omega$ satisfies
  \begin{equation*}
    (T\otimes g^{-1})\circ \Delta(\Omega) = S(\Omega) = \Omega.
  \end{equation*}
\end{theorem}

Let us note that the (positive) numbers $d_a$ satisfy (see \cref{cor:T}) $\sum_b N_{ab}^c d_b = \delta_{l_al_c} d_a \cdot d_c$. In particular, if $E$ consists of a single irrep $e$, then $\delta_{l_al_c}$ is never zero and thus $d_b$ defines a positive eigenvector for the matrix $N_a$ defined by $(N_a)_b^c = N_{ab}^c$. As $N_{a}$ is a non-negative matrix, this implies that the corresponding eigenvalue, $d_{a}$, is the spectral radius of $N_{a}$. This number is also called the Perron-Froebenius dimension of $a\in \irr(\mathcal{A}^*)$.

Note finally that if $\mathcal{A}$ is a $C^*$ weak Hopf algebra, then $\mathcal{A}^*$ is also a $C^*$ weak Hopf algebra. This allows us to define an element $\omega\in\mathcal{A}^*$ that is a positive non-degenerate trace-like linear functional on $\mathcal{A}$. Using this $\omega$ then one can define a scalar product on $\mathcal{A}$ by
\begin{equation*}
  \scalprod{x}{y} = \omega(x^* y).
\end{equation*}

\subsection{\texorpdfstring{$C^*$}{C*}-Hopf algebras}\label{sec:C_star_Hopf}

In this section we further particularize the results obtained in the previous sections to $C^*$-Hopf algebras. As a $C^*$-Hopf algebra is a special $C^*$-WHA, \cref{thm:main_c_star} applies. We will show that the element $\Omega$ obtained from this theorem is in fact the Haar integral of the Hopf algebra.

Before stating the definition of a $C^*$ Hopf algebra, note that for a WHA $\mathcal{A}$ the following statements are all  equivalent \cite{Bohm1999}:
\begin{itemize}
\item $\Delta_{\mathcal{A}}(1) = 1\otimes 1$,
\item $\Delta_{\mathcal{A}^*}(\epsilon) = \epsilon\otimes \epsilon$,
\item $\sum_{(x)} S(x_{(1)}) x_{(2)} = \epsilon(x) 1$,
\item $\sum_{(x)} x_{(1)} S(x_{(2)}) = \epsilon(x) 1$.
\end{itemize}
Keeping this equivalence in mind, one can define Hopf algebras and $\mathcal{C}^*$ Hopf algebras as follows:

\begin{definition}[Hopf algebra and $C^*$ Hopf algebra]
  A Hopf algebra is a weak Hopf algebra such that $\Delta(1) = 1\otimes 1$. A $C^*$ Hopf algebra is a $C^*$ weak Hopf algebra such that $\Delta(1) = 1\otimes 1$.
\end{definition}

Any $C^*$-Hopf algebra $\mathcal{A}$ is semisimple, and thus, due to the Larson-Radford theorem, $S^2=\id$. This implies that the (unique) positive pivotal element of $\mathcal{A}^*$ is $\epsilon$, and thus the numbers $\tr(\psi_a(g)) = d_a$ appearing in \cref{thm:main_hopf_c_star} are in fact $d_a = \tr(\psi_a(\epsilon)) = D_a$, the dimension of the irrep $a$. Moreover, $\Omega$ in \cref{thm:main_c_star} and $\Lambda$ in \cref{thm:wha_special_integral} coincide. In particular, as in a $C^*$ Hopf algebra the unique normalized integral is the Haar integral (Larson-Sweedler theorem), $\Omega$ is the Haar integral of $\mathcal{A}$. We have thus seen that
\begin{theorem}\label{thm:main_hopf_c_star}
  Let $\mathcal{A}$ be a finite dimensional $C^*$ Hopf algebra over $\mathbb{C}$. Then $\mathcal{A}^*$ is semisimple and spherical and its positive spherical element is $\epsilon\in\mathcal{A}^*$. The element $\Omega\in \mathcal{A}$ defined by
  \begin{equation*}
    \phi^{\otimes n} \circ \Delta^{n-1} (\Omega)  = \sum_{a\in \irr(\mathcal{A}^*)} \frac{D_a}{\sum_{x\in \irr(\mathcal{A}^*)} D_x^2} \cdot
    \tikzsetnextfilename{fd27623e-cad2-4d0a-835a-14cc67ee452e}
    \begin{tikzpicture}
      \draw[virtual] (1.5,0) rectangle (5.5,-0.7);
      \foreach \x/\t in {2/1,3/2,5/n}{
        \node[tensor] (t\x) at (\x,0) {};
        \draw[->-] (t\x) --++ (0,0.5);
        \draw[-<-] (t\x) --++ (0,-0.5);
      }
      \node[fill=white] (dots) at (4,0) {$\dots$};
      \draw[virtual,->-] (t3) -- (t2) node[midway,black,irrep,above] {$a$};
      \draw[virtual,->-] (dots) -- (t3) node[midway,black,irrep,above] {$a$};
      \draw[virtual,->-] (t5) -- (dots) node[midway,black,irrep,above] {$a$};
      \draw[virtual,->-] (1.5,-0.7) -- (5.5,-0.7);
    \end{tikzpicture}\ ,
  \end{equation*}
  where $D_a$ is the dimension of the irrep $a$, is the Haar integral of $\mathcal{A}$.
\end{theorem}
In particular, as $g=\epsilon$, the map $T=S$ and the equations
  \begin{gather*}
    (1\otimes x) \cdot \Delta(\Omega) = (T(x)\otimes 1) \cdot \Delta(\Omega), \\
    \Delta\circ T = (T\otimes g\otimes T) \circ \Delta_{\mathrm{op}}^2.
  \end{gather*}
simplify to the definition of the integral and to the fact that $S$ is an anti-cohomomorphism:
  \begin{gather*}
    (1\otimes x) \cdot \Delta(\Omega) = (S(x)\otimes 1) \cdot \Delta(\Omega), \\
    \Delta\circ S = (S\otimes S) \circ \Delta_{\mathrm{op}}.
  \end{gather*}

\section{Pulling-through algebras}\label{sec:pulling_through_alg}

In the previous sections we have seen (\cref{thm:main}, and its variants \cref{thm:main_spherical} and \cref{thm:main_c_star}) that in a semisimple pivotal WHA over $\mathbb{C}$ there is a special non-degenerate cocentral element $\Omega$ that behaves almost like an integral. In this section we abstract this notion and investigate the structure such an element provides, independent from semisimplicity or weak Hopf algebras. The main reason behind this abstraction is that this is exactly the structure we need to use in order to define a PEPS with symmetries, see \cref{sec:peps}.

\begin{definition}[Pulling-through algebra]
    A finite dimensional pre-bialgebra over $\mathbb{C}$ is called a pulling-through algebra if there is a cocentral, non-degenerate element $\Omega\in \mathcal{A}$, a linear map $T:\mathcal{A}\to \mathcal{A}$ and a group-like linear functional $g\in \mathcal{A}^*$ such that for all $x\in \mathcal{A}$,
    \begin{gather}
      (1\otimes x) \cdot \Delta(\Omega) = (T(x) \otimes 1) \cdot \Delta(\Omega), \label{eq:pulling_through}\\
      \Delta\circ T = (T\otimes g \otimes T)\circ\Delta_{\mathrm{op}}^2 \label{eq:T_cohomo}.
    \end{gather}
\end{definition}

\Cref{eq:pulling_through} will be called \emph{pulling-through} equation. The rationale behind the name will become clear when we give a tensor network description in \cref{eq:pulling_through_tn} below. In a pulling-through algebra, $\Omega$ uniquely determines $T$: if there was another map $\hat{T}:\mathcal{A}\to\mathcal{A}$ also satisfying \cref{eq:pulling_through}, then
\begin{equation*}
  (T(x) \otimes 1) \cdot \Delta(\Omega) = (1\otimes x) \cdot \Delta(\Omega) = (\hat{T}(x) \otimes 1) \cdot \Delta(\Omega),
\end{equation*}
and thus by non-degeneracy of $\Omega$, $\hat{T}(x) = T(x)$. Using that $\Delta(\Omega) = \Delta_{\mathrm{op}}(\Omega)$, one can show that $T$ is an involution:
\begin{equation*}
  (1\otimes x) \cdot\Delta(\Omega) = (T(x) \otimes 1) \cdot \Delta(\Omega) = (T(x) \otimes 1) \cdot\Delta_{\mathrm{op}}(\Omega) = (1 \otimes T^2(x))\cdot \Delta_{\mathrm{op}}(\Omega) = (1 \otimes T^2(x))\cdot \Delta(\Omega),
\end{equation*}
where the third equation is the pulling-through equation for $y= T(x)$. As above, due to non-degeneracy of $\Omega$, $x= T^2(x)$, i.e.\ $T$ is an involution. In particular, $T$ is invertible, and thus there can be at most one linear functional $g\in \mathcal{A}^*$ satisfying \cref{eq:T_cohomo}. Let us show now that $T$ is an anti-homomorphism by using the pulling-through equation \cref{eq:pulling_through} twice:
\begin{equation*}
  (T(xy)\otimes 1)\Delta(\Omega) = (1\otimes xy)\Delta(\Omega) = (T(y)\otimes x)\Delta(\Omega) = (T(y) T(x)\otimes 1)\Delta(\Omega),
\end{equation*}
and thus non-degeneracy of $\Omega$ implies that $T(xy) = T(y) \cdot T(x)$.

Let us note that in a pre-bialgebra $\mathcal{A}$ there might be several different elements $\Omega$ that make it a pulling-through algebra. For example, it is easy to check that given an element $\Omega\in\mathcal{A}$ defining a pulling-through structure on $\mathcal{A}$ and any central group-like element $c\in \mathcal{A}^*$ (if there is any), the  element $\hat{\Omega} = (\id\otimes c) \circ \Delta (\Omega)$ defines another pulling-though structure. This is the case if in the construction of \cref{thm:main} we use two different pivotal elements $g$ and $\hat{g}$ to arrive at the cocentral elements $\Omega$ and $\hat{\Omega}$, respectively.

Let $\mathcal{A}$ be a pulling-through algebra. As $\mathcal{A}$ is a pre-bialgebra, it has MPO representations: given any representation $\phi$ of $\mathcal{A}$ and an injective representation $\psi$ of $\mathcal{A}^*$ on a vector space $W$, let us define an MPO tensor by
\begin{equation*}
 =
  \sum_{x\in B} \phi(x) \otimes \psi(\delta_x),
\end{equation*}
where $B$ is a basis of $\mathcal{A}$, and $\delta_x$ denotes the dual basis (i.e.\ $\delta_x(y) = \delta_{xy}$). Given this MPO tensor, for all $x\in\mathcal{A}$ there is a matrix $b(x)\in \End(W)$ such that
\begin{equation*}
  \phi^{\otimes n}\circ \Delta^{n-1} (x)  = \
	  %
 \ .
\end{equation*}
Notice that in this section we do not assume semisimplicity of $\mathcal{A}^*$. The MPOs, nevertheless, are still multiplicative,
\begin{equation*}
 \tikzsetnextfilename{ada7f35a-8015-4146-91a6-81fadcac0933}
 %
 \ .
\end{equation*}
Let us introduce a new MPO tensor, denoted by white dot, as
\begin{equation*}
  %
 \  =
  \sum_{x\in B} \phi\circ T(x) \otimes \psi(\delta_x) .
\end{equation*}
Note that as $T$ is an anti-homomorphism, this tensor is flipped upside down w.r.t.\ the tensor denoted by full dot, i.e.\ the input index of $\phi$ is on the top of the tensor, not on the bottom. $T(x)$ can be expressed both with the original tensor and this new tensor:
\begin{equation*}
   \phi(T(x)) =
  %
  \ .
\end{equation*}
Note that the orientation of the red lines in the second MPO is the opposite of the orientation of the red lines in the first MPO. This is because to be able to compare the tensor networks on the two sides of the equation, the white tensor had to be rotated.

With the help of this new tensor, the identity $\sum_{(xy)} (xy)_{(1)} \otimes T((xy)_{(2)}) = \sum_{(x)(y)} x_{(1)} y_{(1)} \otimes T(y_{(2)}) T(x_{(2)}) $ has a particularly nice MPO representation:
\begin{equation*}
  \tikzsetnextfilename{9be023a2-9776-40e2-b449-7041cba9d397}
  %
  \ .
\end{equation*}
Let small green dot denote the (representation of the) group-like element $g$; then the tensor network representation of the identity $\Delta\circ T(x)  = (T\otimes g\otimes T) \circ \Delta^2_{op}(x)$ is
\begin{equation*}
  \Delta \left( \
  %
  \ ;
\end{equation*}
notice here that the third MPO is oriented in the opposite direction as the second one. This is due to the fact that it is the MPO corresponding to $\Delta_{\mathrm{op}}$, and not to $\Delta$. This equation describes the coproduct of the white tensor:
\begin{equation}\label{eq:white_tensor_coproduct}
  \Delta:
    %
  \ .
\end{equation}

We will denote the boundary corresponding to the pulling-through element $\Omega$ by a small blue dot without label:
\begin{equation*}
  \phi^{\otimes n}\circ \Delta^{n-1} (\Omega)  = \
	  %
 \ .
\end{equation*}
As $\Omega$ is co-central, the corresponding MPO is translation invariant, that is,
\begin{equation*}
	  %
 \ .
\end{equation*}

Let us now explain how these MPO representations will be used when depicting two-dimensional tensor networks in the next section. Keeping in mind that the orientation of each MPO tensor has to be kept, we can draw MPOs in arbitrary form such as on a circle. For example, the fact that $\Omega$ is cocentral can be represented by the equation
\begin{equation*}
  \phi^{\otimes 4} \circ \Delta^3 (\Omega) =
  %
\ .
\end{equation*}
The pulling through equation $(1\otimes x) \cdot \Delta(\Omega) = (T(x) \otimes 1) \cdot \Delta(\Omega)$ is represented by the following tensor network notation:
\begin{equation*}
  %
 \ .
\end{equation*}
As this equation holds for all $x\in\mathcal{A}^*$, it also holds with open indices:
\begin{equation}\label{eq:pulling_through_tn}
  %
 \ .
\end{equation}
Applying $T$ on the left tensor component of this equation, we obtain
\begin{equation*}
  %
 \ .
\end{equation*}
This equation is in fact why $(1\otimes x) \cdot \Delta(\Omega) = (T(x) \otimes 1) \cdot \Delta(\Omega)$ is called pulling-through equation: it states that the MPO described by the black MPO tensors can be pulled through the circular MPO describing $(T\otimes\id)\circ\Delta(\Omega)$.

To familiarize ourselves with the notation, we derive some equations that we will use in the following section. Taking different coproducts of \cref{eq:pulling_through_tn}, and using \cref{eq:white_tensor_coproduct}, we arrive at, for example,
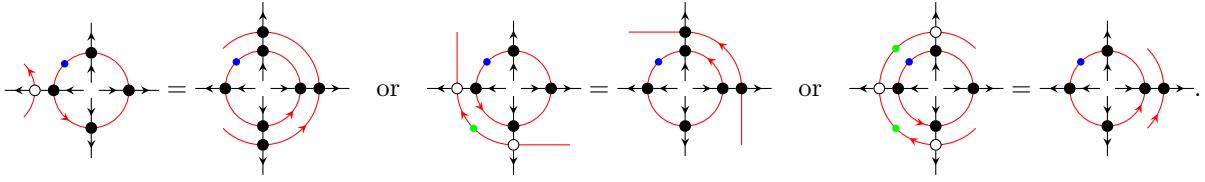
\begin{equation*}
  \tikzsetnextfilename{a5d1cf82-3817-4507-a522-db9334afdbdf}
.
\end{equation*}
Applying $T$ on some of the tensor components reverses the orientation of the lines corresponding to the physical indices of those tensor components, and changes the black tensors into white ones and vice versa, that is, for example, the following holds:
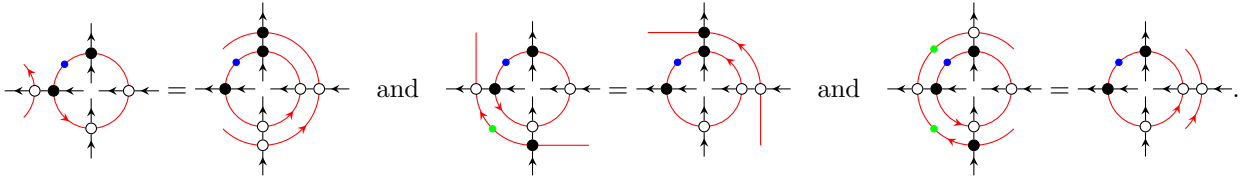
\begin{equation*}
  \tikzsetnextfilename{4a86140a-a3bf-432e-a27c-97cd90acc556}
  %
.
\end{equation*}
These identities are obtained from the previous ones by applying $T$ on the two tensor components on the bottom and on the right.

As $g\in\mathcal{A}^*$ is group-like, applying $g$ on the product of MPO tensors is the same as applying it on the individual tensors. For example, applying the linear functional $g$ on the upper left tensor component of the equation
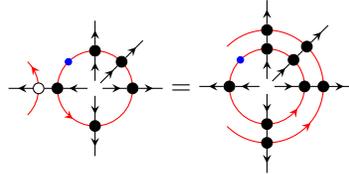
\begin{equation*}
  \tikzsetnextfilename{ca67ce4f-4b73-4d6c-b23a-8bd2d4e9585b}
  %
\end{equation*}
results in the equation
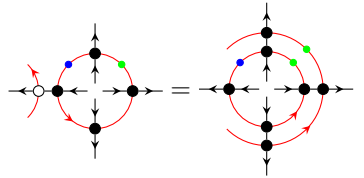
\begin{equation*}
  \tikzsetnextfilename{f24fe3a8-ed2b-463d-9e19-24062aa6e961}
  %
 \ .
\end{equation*}
Non-degeneracy of $\Omega$ means that the action of the linear functional $\lambda$ defined by $(\lambda\otimes\id)\circ\Delta(\Omega)=1$ removes closed loops:
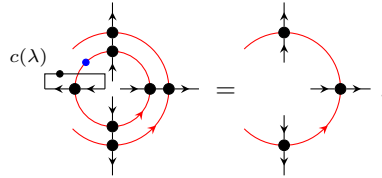
\begin{equation*}
  \tikzsetnextfilename{6d68c460-7a9f-4fa9-b34a-3818e3322b06}
  %
 \ .
\end{equation*}
Finally note that the above equations also hold for non-translation invariant MPO representations of the pulling-through algebra.

\section{MPO-injective PEPS}\label{sec:peps}

In this section we define projected entangled pair states that possess certain symmetries described by a pulling-through algebra $\mathcal{A}$.   Such a PEPS will also be called $\mathcal{A}$-injective PEPS or MPO-injective PEPS. As the primary examples for pulling-through algebras are pivotal cosemisimple weak Hopf algebras, and the representations of such algebras form a pivotal fusion category, our definition can be translated into a  category theoretical language. In fact, when the pulling-through algebra corresponds to a $C^*$-WHA, we believe that our definition is equivalent to the formalism presented in \cite{Lootens2020}.  We expect that the pulling-through structure is actually more general than cosemisimple weak Hopf algebras; for example, certain
non-semisimple pivotal Hopf algebras will also admit a pulling-through structure. States corresponding to those models have unusual properties. We show that MPO-injectivity is a topological property in the sense that it is invariant under the blocking of tensors (and thus under renormalization); in fact, our definition has been designed to satisfy this property. We also show, in the special case where $\mathcal{A}$ is a $C^*$-Hopf algebra, the relation between this PEPS and the generalization of Kitaev's toric code to $C^*$-Hopf algebras \cite{Buerschaper2013}. Following \cite{Bultinck2017}, we then construct a set of states that we call local excitations and that are states that differ only locally from the PEPS. We also construct a set of local operators that form a representation of the Drinfeld double $\mathcal{D}(\mathcal{A}^{*})$ of the Hopf algebra $\mathcal{A}^{*}$ and show that these local operators transform the local excitations among each other. We identify the anyons of the model as the irrep sectors of the local excitations under the action of these local operators.

Let us first recall the definition of two dimensional PEPS. A 2D PEPS is defined on a directed pseudo-graph\footnote{by pseudo-graph we mean a graph where multiple edges are allowed between vertices} $\mathcal{G}= (\mathcal{V},\mathcal{E})$  that can be drawn on an orientable 2-manifold manifold such that no edges intersect (like in the case of a planar graph). The figure below shows a local part of such a graph (locally it \emph{is} a planar graph). For convenience we have numbered the vertices and we have displayed a circular arrow around vertex $5$ depicting the orientation of the 2-manifold:
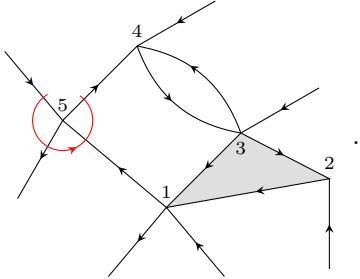
\begin{equation*}
  \tikzsetnextfilename{e42f46ff-1cb9-419c-b70c-7b937781cefb}
  \begin{tikzpicture}
    \def\r{1.2};
    \coordinate[] (c1) at (0,0) {};
    \coordinate[] (c2) at (10:2.2) {};
    \coordinate[] (c3) at (45:1.4) {};
    \coordinate[] (c4) at ( $(c3)+(140:1.8) $){};
    \coordinate[] (c5) at (140:1.8) {};
    \coordinate[] (e1) at ( $(c1)+(-130:\r) $);
    \coordinate[] (e2) at ( $(c1)+(-50:\r) $) {};
    \coordinate[] (e3) at ( $(c3)+(30:\r) $){};
    \coordinate[] (e4) at ( $(c2)+(-90:\r) $) {};
    \coordinate[] (e5) at ( $(c4)+(30:\r) $) {};
    \coordinate[] (e6) at  ( $(c5)+(-120:\r) $) {};
    \coordinate[] (e7) at ( $(c5)+(130:\r) $) {};
    % the plaquette
    \fill[lightgray!50] (c1)--(c2)--(c3)--(c1);
    % labels
    \node[anchor=south] at (c1) {1};
    \node[anchor=south] at (c2) {2};
    \node[anchor=north] at (c3) {3};
    \node[anchor=south] at (c4) {4};
    \node[anchor=south] at (c5) {5};
    % edges
    \draw[postaction={on each segment={-<-}}] (e1)-- (c1)--(c2)--(c3) to[bend left] (c4)--(c5)--(c1)--(e2);
    \draw[->-] (c3) to[bend right] (c4);
    \draw[postaction={on each segment={-<-}}] (c1)--(c3)--(e3);
    \draw[-<-] (c2)--(e4);
    \draw[-<-] (c4)--(e5);
    \draw[postaction={on each segment={-<-}}]  (e6)--(c5)-- (e7);
    % the arrow for orientation
    \draw[red, ->-] ($(c5)+(120:0.4)$) arc (120:415:0.4);
  \end{tikzpicture} \ .
\end{equation*}
The areas enclosed by minimal (unoriented) cycles in the graph are called plaquettes. Such a plaquette is denoted by light gray shading between the vertices $1$, $2$ and $3$. Let $\mathcal{E}^{op}$ denote the set of edges in $\mathcal{E}$ with reversed orientation, and for any $e\in \mathcal{E}$ or $e\in\mathcal{E}^{op}$ let $\bar{e}$ denote the edge $e$ with the opposite orientation.

To every edge $e\in\mathcal{E}$ we assign a finite dimensional complex vector space, $V_e$. Let us also assign a vector space $V_{\bar{e}}$ to every oppositely oriented edge $\bar{e}\in \mathcal{E}^{op}$ such that $V_{\bar{e}} = V_e^*$. Note that this relation is symmetric, i.e. $V_{e} = V_{\bar{e}}^* = V_{\bar{\bar{e}}}$.  To every vertex we assign a (finite dimensional) Hilbert space, $\mathcal{H}_v$. Finally, to each vertex $v$ we assign a tensor $A_v\in \mathcal{H}_v \otimes \bigotimes_{e\in \mathcal{N}_v} V_{e}$, where $\mathcal{N}_v$ denotes the set of edges $e\in \mathcal{E}\cup\mathcal{E}^{op}$ that connect $v$ with another vertex such that the orientation of $e$ points away from $v$. Given all these data, the state defined by the PEPS is obtained by contracting $\bigotimes_v A_v$ along the edges of the graph; note that this contraction is possible, because if $e$ is an edge between $v$ and $w$, then the tensor component corresponding to $e$ in the tensor $A_v$ is $V_e$, while in the tensor $A_w$ it is $V_{\bar{e}} = V_e^*$. The state defined by the PEPS can be represented, using the graphical notation of tensor calculus, as
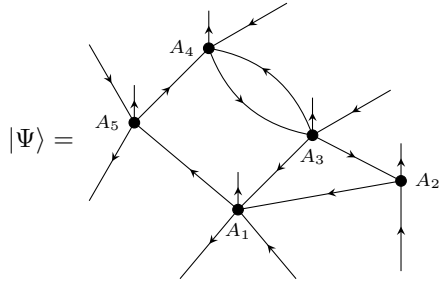
\begin{equation*}
  \tikzsetnextfilename{413b9e0b-6cf6-475f-8156-36635330bba2}
  \ket{\Psi} =
  \begin{tikzpicture}
    \def\r{1.2};
    \def\p{0.07};
    \coordinate[] (c1) at (0,0) {};
    \coordinate[] (c2) at (10:2.2) {};
    \coordinate[] (c3) at (45:1.4) {};
    \coordinate[] (c4) at ( $(c3)+(140:1.8) $){};
    \coordinate[] (c5) at (140:1.8) {};
    \coordinate[] (e1) at ( $(c1)+(-130:\r) $);
    \coordinate[] (e2) at ( $(c1)+(-50:\r) $) {};
    \coordinate[] (e3) at ( $(c3)+(30:\r) $){};
    \coordinate[] (e4) at ( $(c2)+(-90:\r) $) {};
    \coordinate[] (e5) at ( $(c4)+(30:\r) $) {};
    \coordinate[] (e6) at  ( $(c5)+(-120:\r) $) {};
    \coordinate[] (e7) at ( $(c5)+(120:\r) $) {};
    % labels
    \node[anchor = north, inner sep = 5pt] at (c1) {$A_1$};
    \node[anchor = west, inner sep = 5pt] at (c2) {$A_2$};
    \node[anchor = north, inner sep = 5pt] at (c3) {$A_3$};
    \node[anchor = east, inner sep = 5pt] at (c4) {$A_4$};
    \node[anchor = east, inner sep = 5pt] at (c5) {$A_5$};
    \draw[-<-] (e1)-- (c1);
    \draw[-<-] (c1)--(c2);
    \draw[-<-] (c2)--(c3);
    \draw[-<-] (c3) to[bend left] (c4);
    \draw[->-] (c3) to[bend right] (c4);
    \draw[-<-] (c4)--(c5);
    \draw[-<-] (c5)--(c1);
    \draw[-<-] (c1)--(e2);
    \draw[-<-] (c1)--(c3);
    \draw[-<-] (c3)--(e3);
    \draw[-<-] (c2)--(e4);
    \draw[-<-] (c4)--(e5);
    \draw[-<-]  (e6)--(c5);
    \draw[-<-]  (c5)-- (e7);
    \foreach \x in {1,2,3,4,5}{
      \node[tensor] at (c\x) {};
      \draw[->-] (c\x) --++ (0,0.5);
    }
  \end{tikzpicture} \ .
\end{equation*}

Let us now define MPO-injective PEPS. These are PEPS such that the PEPS tensors are invariant under certain symmetry operations acting on their virtual degrees of freedom, i.e.\ on the tensor components $\bigotimes_{e\in \mathcal{N}_v} V_{e}$ of the tensor $A_v$. More precisely, a PEPS tensor $A_v\in \mathcal{H}_v\otimes\bigotimes_{e\in \mathcal{N}_v} V_{e}$ is called $O_v$-injective for an operator $O_v\in \End(\bigotimes_{e\in \mathcal{N}_v} V_{e})$, if there are tensors $B_v$ and $C_v$ such that
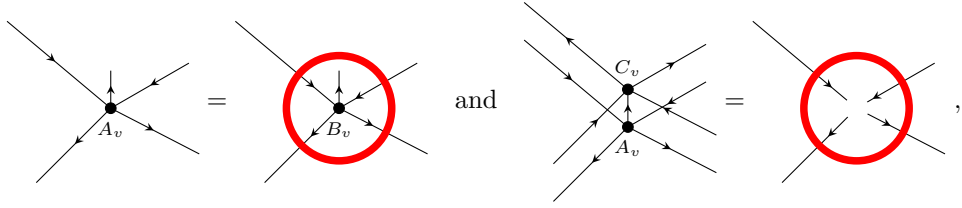
\begin{equation*}
  \tikzsetnextfilename{e1f2a962-37b9-4658-aaf0-402493511505}
 \ ,
\end{equation*}
where the red circle denotes the operator $O_v$. Note that if $O_v$ is a projector, then one can choose $B_v = A_v$. If, on the other hand, $O_v$ is nilpotent, $B_v$ and $A_v$ are different. This is the case when the symmetries of the tensor form a non-semisimple algebra. In the following we will define the symmetry operators $O_v$ for each vertex $v\in\mathcal{V}$. To do so, we need the following additional data. First, a pulling-through algebra $(\mathcal{A},\Omega,T,g)$, and second, representations  $\phi_{e}:\mathcal{A}\to \End(V_{e})$ on  the vector spaces $V_{e}$  such that $\phi_{\bar{e}} =  \bar{\phi}_{e}$, where $\bar{\phi}(x) = [\phi\circ T(x)]^T$ for all $x\in\mathcal{A}$. Note that as $T$ is idempotent, this relation is symmetric: if $\phi_{\bar{e}} = \bar{\phi}_{e}$, then $\phi_{e} = \bar{\phi}_{\bar{e}}$ as well. Finally, for every plaquette in the graph, we will choose a vertex from the ones surrounding the plaquette. Below we denote such a choice by putting a black dot close to the vertex chosen inside each plaquette:
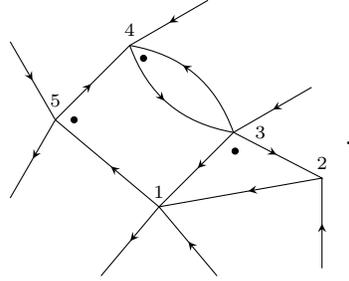
\begin{equation*}
 \ .
\end{equation*}
The symmetry operator $O_v$ around each vertex $v$ is an MPO representation of the element $\Omega$ of the pulling-through algebra $\mathcal{A}$. The concrete form of the MPO is designed such that it is a generalization of $G$-injective PEPS \cite{Schuch2010} and such that the PEPS remains MPO-injective even after blocking. Let us explain through an example what this MPO representation exactly is. The symmetry operator $O_v$, around the vertex $3$, takes the following form:
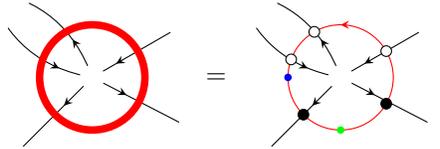
\begin{equation*}
  \tikzsetnextfilename{4e714e47-1932-46bd-acdf-1b8349a95e83}
  %
 .
\end{equation*}
As stated above, this MPO is an MPO representation of $\Omega$. The blue dot represents the boundary $b(\Omega)$; remember that as $\Omega$ is cocommutative, the placement of the boundary is not relevant. The orientation of the virtual index of the MPO follows the orientation of the surface, as the arrow on it shows. The MPO consists of two different MPO tensors. For an outgoing edge $e$, we use the MPO tensor denoted by black dots given by
\begin{equation*}
  %
 =
  \sum_{x} \phi_e(x) \otimes \psi(\delta_x),
\end{equation*}
where $x$ runs over a basis of $\mathcal{A}$ and $\psi$ is a representation of $\mathcal{A}^*$ on the virtual vector space denoted by red indices. For an incoming edge $e$, we use the white MPO tensor given by
\begin{equation*}
  %
 =
  \sum_{x} \phi_e\circ T(x) \otimes \psi(\delta_x) = \sum_{x} \phi_{\bar{e}}(x)^T \otimes \psi(\delta_x).
\end{equation*}
Note that this white tensor is constructed the same way as the black tensor, but for the construction we use the wrong representation: $\phi_{\bar{e}}$ instead of $\phi_e$ (in the formula transpose appears because on the l.h.s.\ we read the tensor as a linear map from the bottom to the top). Changing the orientation of the edge $e$ (i.e.\ replacing it with $\bar{e}$) changes the black tensors to white and white tensors to black. Finally, as the vertex $3$ was selected for the plaquette surrounded by vertices $1$, $2$ and $3$, we insert the linear functional $g$ between the edges connecting the vertex $3$ to the vertices $1$ and $2$.

Using the above definition of an MPO-injective PEPS tensor, an MPO-injective PEPS -- after applying the inverse tensors $C_v$ at each vertex $v$ -- can be written as the following tensor network:
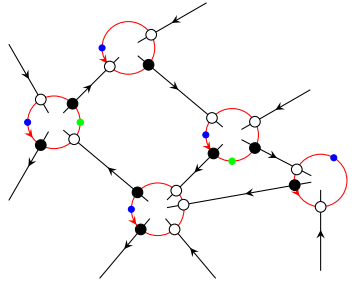
\begin{equation*}
  \tikzsetnextfilename{f841a111-5e76-45c2-bb1e-fc88fb73c9b1}
 \ .
\end{equation*}
In general pulling-through algebras (such as a pulling-through algebra that originates from a pivotal WHA that is  not a $C^*$-WHA) different placements of the  group-like elements $g$ lead to different states. In a pulling-through algebra that originates from a $C^*$-WHA, however, all these states are related to each other by local operations, and in fact, one can define MPO-injective PEPS in a more translation invariant way. To understand why, recall \cite{Bohm1999} that in a $C^*$-WHA the positive spherical element $g\in\mathcal{A}^*$ can be written as $g=g_L g_R^{-1} = g_R^{-1} g_L$ such that $g_R=S(g_L)$ and such that there are algebra elements $G_L,G_R\in\mathcal{A}$ satisfying the following relations:
\begin{gather*}
 \ .
\end{gather*}
The equation $G_L = S(G_R)$ also holds and it is easy to check that $G_L = T(G_R)$ holds as well. Applying thus $T$ on the previous relations, we obtain that the white tensors satisfy
\begin{gather*}
  %
 \ .
\end{gather*}
Let us write now $g=g_L g_R^{-1}$ in the definition of the MPO-injective PEPS. Using the relations above, we obtain then that $g_R^{-1}$ translates to a local operator acting on the vertex which $g$ belongs to, while $g_L$ is delocalized around the plaquette. Let us illustrate this fact by depicting one of the plaquettes of the previous PEPS:
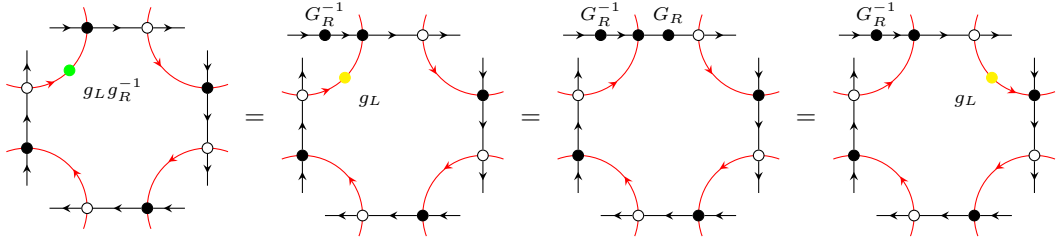
\begin{equation*}
  %
 \
\end{equation*}
Applying thus $G_R$ on the appropriate tensor components, we obtain that the state defined by
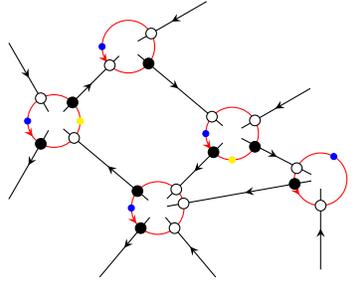
\begin{equation*}
  \tikzsetnextfilename{b1c76190-4708-48d5-b195-402298df74ea}
  %
 \ ,
\end{equation*}
where the yellow dot now denotes $g_L$ instead of $g=g_Lg_R^{-1}$, is also MPO-injective and in this state the yellow dots move freely around the plaquette:
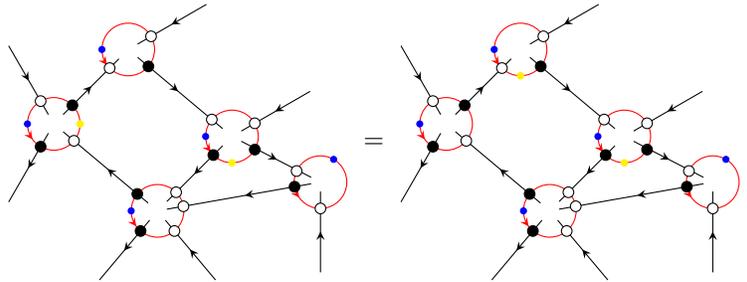
\begin{equation*}
  \tikzsetnextfilename{16e56885-50e0-406d-a8ba-3fa3de6fda24}
  %
 \ .
\end{equation*}
As we have seen above, the yellow dot does not even have to appear on the virtual index of the MPO, it can be instead inserted between two PEPS tensors anywhere around the plaquette.

\subsection{Scale independence}

In this section we introduce an operation on states that we call \emph{blocking} and show that an MPO-injective PEPS stays MPO-injective even after blocking. Blocking is the basis of renormalization (there it is followed by an isometry that gets rid of certain degrees of freedom) and has a natural representation in tensor networks. Blocking simply means that we treat certain -- neighboring -- particles together: for example, given a three-partite state $\ket{\psi}\in \mathcal{H}_1 \otimes \mathcal{H}_2 \otimes \mathcal{H}_3$, the blocking of particles $2$ and $3$ means that we reinterpret $\ket{\psi}$ as a two-partite state in $\mathcal{H}_1 \otimes \mathcal{H}_{23}$, where $ \mathcal{H}_{23} = \mathcal{H}_2 \otimes \mathcal{H}_3$.

In PEPS, blocking is a partial contraction of the tensor network, i.e. that in a given region we replace all tensors by one tensor that is the result of the contraction of the tensors in that region. The blocked tensor network is then another, coarser tensor network. For example, in the tensor network below blocking of four tensors results in the following new tensor network:
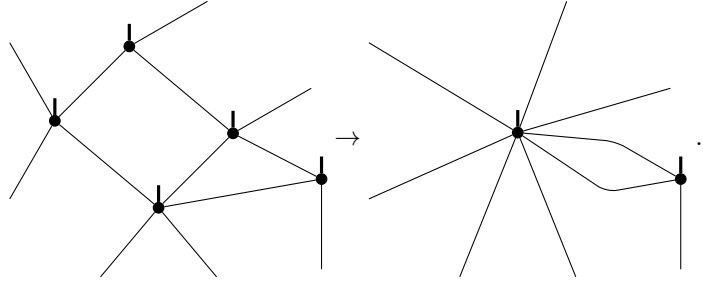
\begin{equation*}
  \tikzsetnextfilename{d54b4892-b167-45bf-b2ce-8b86630876c1}
 \ .
\end{equation*}
As the example above shows, sometimes two edges $e$ and $f$ that were distinct edges in the original PEPS become a single double edge in the blocked tensor network. We can block these two edges together, i.e.\ consider $V_e\otimes V_f$ as a single vector space. Pictorially we express this as writing a single edge instead of the two distinct ones:
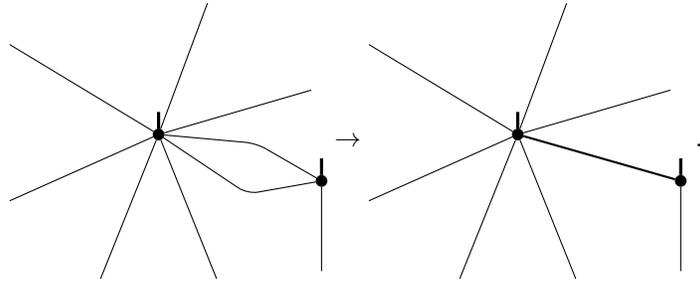
\begin{equation*}
   \tikzsetnextfilename{56daa60e-add5-4ae6-af7f-fa5b95ec5415}
   %
 \ .
\end{equation*}
In fact, every blocking can be broken down into a series of simple steps: in each step we either block two neighboring tensors that are connected by a singe edge, possibly leading to double edges in the PEPS, or block two edges together to remove a double edge from the PEPS. Note that in the first case the number of plaquettes does not change (if there are double edges after blocking, the area enclosed between the two edges is considered as a plaquette), while in the second case the number of plaquettes decreases by one.

In the following we will show that an MPO-injective PEPS stays MPO-injective even after blocking. Let us consider a pulling-through algebra $\mathcal{A}$ and an $\mathcal{A}$-injective PEPS. We first show that blocking two neighboring tensors in the $\mathcal{A}$-injective PEPS results in another $\mathcal{A}$-injective PEPS, and then, that removing double edges does not change the $\mathcal{A}$-injectivity property either.

Let us now consider the blocking of two neighboring $\mathcal{A}$-injective PEPS tensors that are connected to each other with a single edge. After applying the inverse tensors on the two PEPS tensors and using the pulling-through property, we obtain
  \begin{equation*}
    \tikzsetnextfilename{ccc4186c-fc3d-401f-805d-d56fea9181d0}
  \ .
\end{equation*}
Due to non-degeneracy, applying a suitable linear functional to the inner two indices, the inner MPO disappears leading to a single MPO on the boundary:
  \begin{equation*}
    \tikzsetnextfilename{94471771-0b2b-4c01-8e7d-b1331adbe46d}
    %
  \ .
\end{equation*}
We have thus constructed an inverse tensor (the inverse tensors of the individual PEPS tensors contracted with the matrix describing the linear functional used above) for the blocked tensor, and showed that it is MPO-injective. Let us note that in the process neither the ``outgoing arrow is black tensor, incoming arrow is white tensor'' nor the ``one green dot per plaquette'' property have changed, and thus not only the blocked tensor is MPO-injective but the whole PEPS remains MPO-injective as well.

Let us now consider the blocking of two neighboring edges $e$ and $f$. Assume that the edges are oriented in the opposite direction and that the ``one green dot per plaquette'' is on the left vertex:
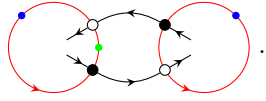
\begin{equation*}
 \ .
\end{equation*}
Here we did not indicate the other MPO tensors on the two vertices. In order to block these two edges together, let us first reorient the lower edge, $f$. The reorientation of the arrow implies that we change the corresponding vector space $V_f$ to $V_{\bar{f}}$ and the representation $\phi_f$ to $\phi_{\bar{f}}$. This also changes the black and white tensors and thus the tensor network notation becomes:
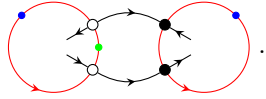
\begin{equation*}
  \tikzsetnextfilename{88191006-d386-461f-9d39-119cf2d42b29}
  %
 \ .
\end{equation*}
Blocking the edges $e$ and $f$ now simply means  replacing the two individual $\mathcal{A}$-modules $V_e$ and $V_f$ by their tensor product $V_f\boxtimes V_e$. After blocking, we obtain a single edge with two MPO tensors at the two ends:
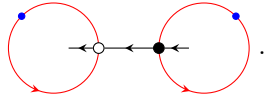
\begin{equation*}
  \tikzsetnextfilename{85885744-aded-466d-93ea-97d5683ec0fb}
  %
 \ .
\end{equation*}
Here the black MPO tensor is built using the representation $\phi_{fe} = \phi_f \boxtimes \phi_e$, while the white MPO tensor is built using the transpose of the representation $\bar{\phi}_{fe}$,
\begin{equation*}
  \phi_{fe} \circ T = (\phi_{f} \otimes \phi_e) \circ \Delta \circ T = (\phi_f \circ T \otimes g \otimes \phi_e \circ T) \circ \Delta_{\mathrm{op}}^2.
\end{equation*}
This equation shows that the single (blocked) white tensor is the concatenation of the two white tensors with the green dot in the middle. The fact that the equation contains $\Delta_{\mathrm{op}}$ instead of $\Delta$ simply reflects the fact that as the orientation of the two circles are the same, the arrow on the virtual index (red line) on the left is oriented in the  opposite direction as on the right.

As the blocking leaves the rest of the PEPS invariant, blocking two edges keeps the PEPS MPO-injective. We have thus shown that MPO-injectivity is invariant both under the blocking of neighboring vertices and under the blocking of neighboring edges. As blocking any number of tensors in a simply connected region can be decomposed into a series of such simple blocking steps, we have proven that MPO-injectivity is invariant under blocking.

\subsection{Relation to the Kitaev model}

The tensor networks defined above are strongly linked to the generalized Kitaev models defined in \cite{Buerschaper2013}. In this section we show the concrete connection. First, let us consider an MPO-injective PEPS constructed from a $C^*$-Hopf algebra $\mathcal{A}$. Based on the construction of this PEPS, we can define another multi-partite state $\ket{\Psi}$ where the individual degrees of freedom are described on the Hilbert space $\mathcal{A}$. Second, we construct a parent Hamiltonian for $\ket{\Psi}$ -- the Kitaev Hamiltonian -- that is the sum of commuting projectors. Finally, we define the Drinfeld double $D(\mathcal{B})$ of any Hopf algebra $\mathcal{B}$ and construct a set of local operators that form a representation of $D(\mathcal{A}^*)$; we also construct a set of local deformations of $\ket{\Psi}$ and show that this set of states $S$ is invariant under the action of $D(\mathcal{A}^*)$. We identify the anyons of the model as the irreps sectors of $S$ under the action of $\mathcal{D}(\mathcal{A}^*)$.

Let us first show how the state $\ket{\Psi}$ is defined. For that, consider the PEPS defined by a $C^*$-Hopf algebra. The construction of these states is easier to understand than in the general case, because there are no group-like elements $g$ present in it (as $g=\epsilon$). The PEPS, after applying the inverse tensors, read as
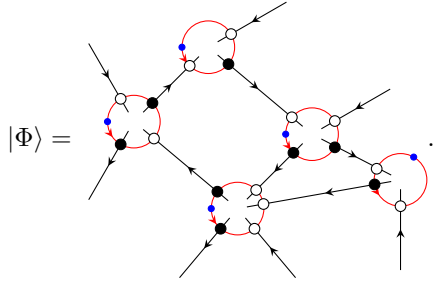
\begin{equation*}
  \ket{\Phi} = \
  \tikzsetnextfilename{6c0bdfbb-3d16-459c-a23e-93e1b7de5d40}
.
\end{equation*}
This vector $\ket{\Phi}$ lives in the vector space $ \bigotimes_{v\in \mathcal{V}} \left(\bigotimes_{e\in \mathcal{N}_v} V_{e} \right)$, because at every vertex $v$ the degrees of freedom are $\bigotimes_{e\in \mathcal{N}_v} V_{e}$. Note that in this tensor product every edge appears exactly twice, once with the orientation defined by the graph and once with opposite orientation. Let us now rearrange this tensor product and group together the degrees of freedom corresponding to the two orientations of each edge. After this regrouping, we can interpret the previous vector space as $ \bigotimes_{e \in \mathcal{E}} \left(V_{e} \otimes V_{\bar{e}}\right)$. Finally note that as $V_{\bar{e}} = V_e^*$ and for every finite dimensional vector space $V$, the tensor product $V\otimes V^*$ is isomorphic to $ \End(V)$, this vector space can also be interpreted as  $\bigotimes_{e \in \mathcal{E}} \End(V_e)$, i.e.\ one can think of $\ket{\Phi}$ as
\begin{equation*}
  \ket{\Phi} \in \bigotimes_{e \in \mathcal{E}} \End(V_e).
\end{equation*}
Note that, by construction, $\ket{\Phi}$ is not supported on the whole Hilbert space $\bigotimes_{e \in \mathcal{E}} \End(V_e)$, but instead only on the subspace $\bigotimes_{e \in \mathcal{E}} \phi_{e} (\mathcal{A})$, where $\phi_{e}$ is the representation of $\mathcal{A}$ used on the edge $e$. This implies that we can define another vector, $\ket{\Psi}$, as
\begin{equation*}
  \ket{\Psi} = \left(\bigotimes_{e \in \mathcal{E}} \phi_{e}^{-1}\right) \ket{\Phi}
  \in  \bigotimes_{e \in \mathcal{E}} \mathcal{A} .
\end{equation*}
If $\mathcal{A}$ is a $C^*$-Hopf algebra, there is a scalar product on $\mathcal{A}$, and thus the vector space $\bigotimes_{e \in \mathcal{E}} \mathcal{A}$ is a finite dimensional Hilbert space. Therefore $\ket{\Psi}$ can be interpreted as a (possibly unnormalized) state. The maps $\phi_{e}^{-1}$ act locally, therefore this state is again a PEPS -- the virtual legs of the PEPS tensor are the red lines in the figure above. This PEPS is then the same (up to a choice of orientation of the arrows on both the red and black lines) as the one described in \cite{Buerschaper2013}.

\subsubsection{The Kitaev Hamiltonian}

Building on the results of \cite{Buerschaper2013}, in this section we explicitly construct the Kitaev parent Hamiltonian for the state $\ket{\Psi}$ defined in the previous section. To make the reading easier, we will use a graphical language to depict the action of the defined Hamiltonian. This graphical language is nothing but the tensor network representation of the state using the representation $\bigotimes_{e\in\mathcal{E}} \phi_{e}$. For simplicity, let us restrict ourselves to a square lattice, and fix the orientation of the lattice such that all vertical edges point from bottom to top and all horizontal ones from right to left (i.e.\ the product of two elements on the vertical edge reads from top to bottom and on the horizontal edges from left to right).

The Kitaev Hamiltonian consists of two type of terms, the plaquette terms $B_p$ and vertex terms $A_v$. The total Hamiltonian is the sum of these terms,
\begin{equation*}
  H = - \sum_{p\in \text{plaquettes}} B_p - \sum_{v\in \text{vertices}} A_v.
\end{equation*}
Each plaquette term $B_p$ acts on the edges surrounding the plaquette $p$, while each vertex term $A_v$ act on the edges connected to the vertex $v$. We will define $B_p$ and $A_v$ such that they are orthogonal projectors and any two such terms commute. The MPO-injective PEPS $\ket{\Psi}$ is a frustration-free ground state of this Hamiltonian.

Let us now define the operator $A_v$ for a given vertex $v$. As stated above, this operator acts on the edges surrounding the vertex $v$. On these four particles, its action is given by:
\begin{equation*}
  A_v : x\otimes y\otimes z\otimes v \mapsto \sum_{(\Omega)} x \cdot \Omega_{(1)} \otimes S(\Omega_{(2)}) \cdot y \otimes S(\Omega_{(3)}) \cdot z \otimes v\cdot \Omega_{(4)} \ ,
\end{equation*}
where $\Omega$ is the Haar integral of $\mathcal{A}$, $x$ is the particle above the vertex, $y$ is the one on its right, $z$ is the one below and $v$ is the one on the left. The concrete form of $A_v$ depends on the orientation of the lattice. Using the graphical representation, it is easier to visualize the action of $A_v$:
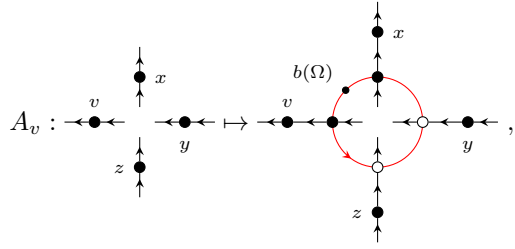
\begin{equation*}
  A_v :
  \tikzsetnextfilename{91b54dc5-64fa-4048-b4b5-eb17fc778154}
  \begin{tikzpicture}[baseline=-1mm]
    \def\r{0.6};
    \foreach \d/\c/\l in {0/1/y,90/0/x,180/0/v,270/1/z}{
      \node[tensor,label = \d-90: $\l$] (r\d) at (\d:\r) {};
      \ifnum\c=0 \def\mystyle{->-};\else \def\mystyle{-<-};\fi
      \draw[\mystyle] (r\d) --++ (\d:0.4);
      \ifnum\c=0 \def\mystyle{-<-};\else \def\mystyle{->-};\fi
      \draw[\mystyle] (r\d) --++ (\d+180:0.4);
     }
  \end{tikzpicture}
  \mapsto
  \tikzsetnextfilename{9c398d80-2d91-474d-bd7c-a01e4f97405f}
  \begin{tikzpicture}[baseline=-1mm]
    \def\r{0.6};
    \draw[virtual,->-=0.65] (0,0) circle (\r);
    \node[t,label=135:$b(\Omega)$] at (135:\r) {};
    \foreach \d/\c/\l in {0/1/y,90/0/x,180/0/v,270/1/z}{
      \node[tensor, fill={\ifnum\c=0 black\else white\fi}] (r\d) at (\d:\r) {};
      \node[tensor,label = \d-90: $\l$] (t\d) at (\d:2*\r) {};
      \ifnum\c=0 \def\mystyle{->-};\else \def\mystyle{-<-};\fi
      \draw[\mystyle] (r\d) --++ (t\d);
      \draw[\mystyle] (t\d) --++ (\d:0.4);
      \ifnum\c=0 \def\mystyle{-<-};\else \def\mystyle{->-};\fi
      \draw[\mystyle] (r\d) --++ (\d+180:0.4);
     }
  \end{tikzpicture} \ ,
\end{equation*}
i.e.\ it multiplies the four particles by the (translation invariant) MPO $\Delta^3(\Omega)$, each from the side that is closer to the vertex. As $\Omega$ is a projector, it is clear that $A_v$ is a projector as well. As both representations $x \mapsto (y\mapsto yx)$ and $x \mapsto (y\mapsto S(x)y)$ of $\mathcal{A}^{op}$ are $*$-representations, $A_v$ is also self-adjoint. Note that for any two different vertices $v_1$ and $v_2$ the Hamiltonian terms $A_{v_1}$ and $A_{v_2}$ clearly commute: if $v_1$ and $v_2$ are not neighboring vertices, they act on different particles; if $v_1$ and $v_2$ are neighboring, there is a single particle on which both of them acts, but if $A_{v_1}$ acts from the left, then $A_{v_2}$ acts from the right of the particle.

Let us now define the operator $B_p$ for a given plaquette $p$. As stated above, this operator acts on the edges surrounding the plaquette. On these four particles, its action is given by
\begin{equation*}
  B_p: x\otimes y \otimes z \otimes v \mapsto \sum_{xyzv} \omega(S(x_{(1)}) S(y_{(1)} z_{(2)} v_{(2)}  )  ) \cdot x_{(2)} \otimes y_{(2)} \otimes z_{(1)} \otimes v_{(1)},
\end{equation*}
where $\omega$ is the Haar integral of $\mathcal{A}^*$ and $x$ is the particle on the right of the plaquette, $y$ is the one on top, $z$ is the one on the left and $v$ is the particle below the plaquette. Again,  the action of $B_p$ is easier to understand using the graphical representation:
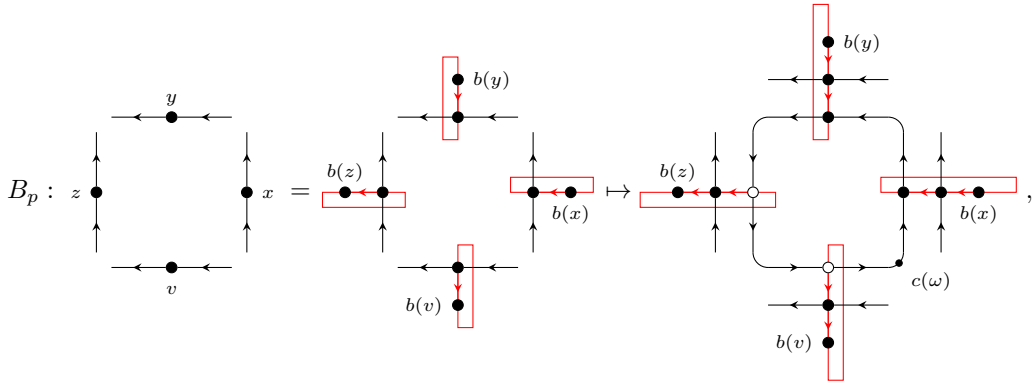
\begin{equation*}
  B_p :
  \tikzsetnextfilename{a533c912-1693-49e1-88c8-4aa63c10f959}
  \begin{tikzpicture}
    \foreach \d/\c/\l in {0/0/x,90/0/y,180/1/z,270/1/v}{
      \node[tensor,label=\d:$\l$] (t\d) at (\d:1) {};
      \ifnum\c=0 \def\mystyle{->-};\else \def\mystyle{-<-};\fi
      \draw[\mystyle] (t\d) --++ (\d+90:0.8);
      \ifnum\c=0 \def\mystyle{-<-};\else \def\mystyle{->-};\fi
      \draw[\mystyle] (t\d) --++ (\d-90:0.8);
    }
  \end{tikzpicture} =
  \tikzsetnextfilename{033583ee-7b1b-43e0-8a8c-136c8f528f4c}
  \begin{tikzpicture}
    \foreach \d/\c/\l in {0/0/x,90/0/y,180/1/z,270/1/v}{
      \draw[virtual] (\d:1.8) rectangle ($(\d:0.7)+(\d+90:0.2)$);
      \node[tensor] (r\d) at (\d:1.0) {};
      \node[tensor,label = \d-90:$b(\l)$] (b\d) at (\d:1.5) {};
      \ifnum\c=0 \def\mystyle{->-};\else \def\mystyle{-<-};\fi
      \draw[virtual,\mystyle] (b\d) -- (r\d);
      \draw[\mystyle] (r\d) --++ (\d+90:0.8);
      \ifnum\c=0 \def\mystyle{-<-};\else \def\mystyle{->-};\fi
      \draw[\mystyle] (r\d) --++ (\d-90:0.8);
    }
  \end{tikzpicture}
  \mapsto
  \tikzsetnextfilename{79babb21-2050-4259-9f5a-0589e293d749}
  \begin{tikzpicture}
    \foreach \d/\c/\l in {0/0/x,90/0/y,180/1/z,270/1/v}{
      \draw[virtual] (\d:2.5) rectangle ($(\d:0.7)+(\d+90:0.2)$);
      \node[tensor,fill={\ifnum\c=1 white\else black\fi}] (t\d) at (\d:1) {};
      \node[tensor] (r\d) at (\d:1.5) {};
      \node[tensor,label = \d-90:$b(\l)$] (b\d) at (\d:2) {};
      \ifnum\c=0 \def\mystyle{->-};\else \def\mystyle{-<-};\fi
      \draw[virtual,\mystyle] (b\d) -- (r\d);
      \draw[virtual,\mystyle] (r\d) -- (t\d);
      \draw[->-] (t\d) --++ (\d+90:0.8) coordinate (c);
      \draw (c) arc (\d:\d+90:0.2);
      \draw[\mystyle] (r\d) --++ (\d+90:0.8);
      \ifnum\c=0 \def\mystyle{-<-};\else \def\mystyle{->-};\fi
      \draw[-<-] (t\d) --++ (\d-90:0.8);
      \draw[\mystyle] (r\d) --++ (\d-90:0.8);
    }
    \node[t,label=-45:$c(\omega)$] at ($(0.8,-0.8) + (-45:0.2)$) {};
  \end{tikzpicture} \ ,
\end{equation*}
where the matrix $c(\omega)$ is the boundary describing the MPO representation of $\omega$, see \cref{eq:dual_mpo}. Similar to the vertex terms, the operators $B_p$ are projectors, and as the representations $f \mapsto (x\mapsto x_{(1)} f(x_{(2)}))$ and $f\mapsto (x\mapsto f\circ S(x_{(1)}) \cdot x_{(2)})$ are both $*$-representations of $\mathcal{A}^*$, they are also self-adjoint. If $p_1$ and $p_2$ are plaquettes that are not neighboring, then $B_{p_1}$ and $B_{p_2}$ act on different particles, and thus they commute. If $p_1$ and $p_2$ are neighboring, then there is one particle both acts on. One of the operators, however, act from the right, the other from the left, and thus even in this case, $B_{p_1}$ and $B_{p_2}$ commute.

Let us now show that the operators $A_v$ and $B_p$ commute, i.e.\ that the Kitaev Hamiltonian is indeed a sum of commuting orthogonal projectors. If the vertex $v$  is not a vertex on the plaquette $p$, then $A_v$ and $B_p$ are acting on different particles and thus they trivially commute. If the vertex $v$ is one of the vertices around the plaquette, we first calculate the graphical representation of the action of $A_v B_p$ on the six particles surrounding both the plaquette and the vertex:
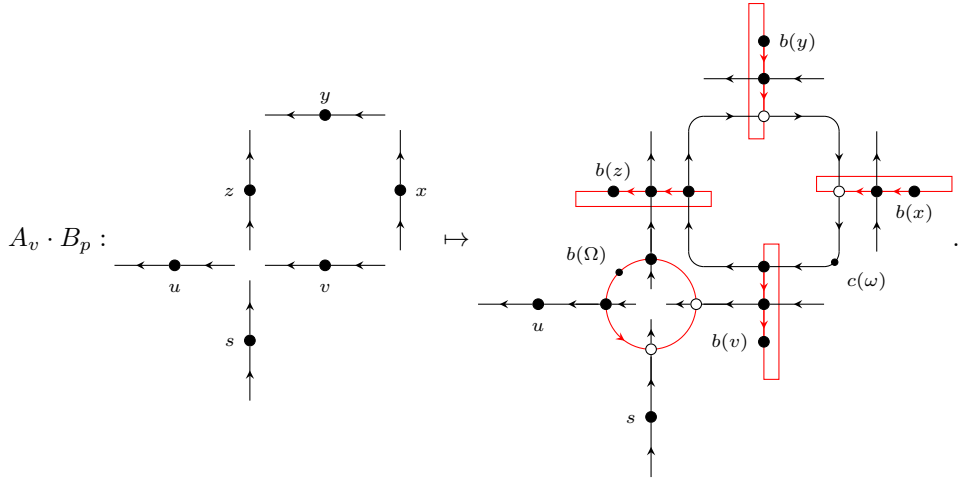
\begin{equation*}
  A_v\cdot B_p :
  \tikzsetnextfilename{ba414df7-303e-4db3-ba51-26c3a5df87c2}
.
\end{equation*}
Let us now calculate the action of $B_p A_v$.
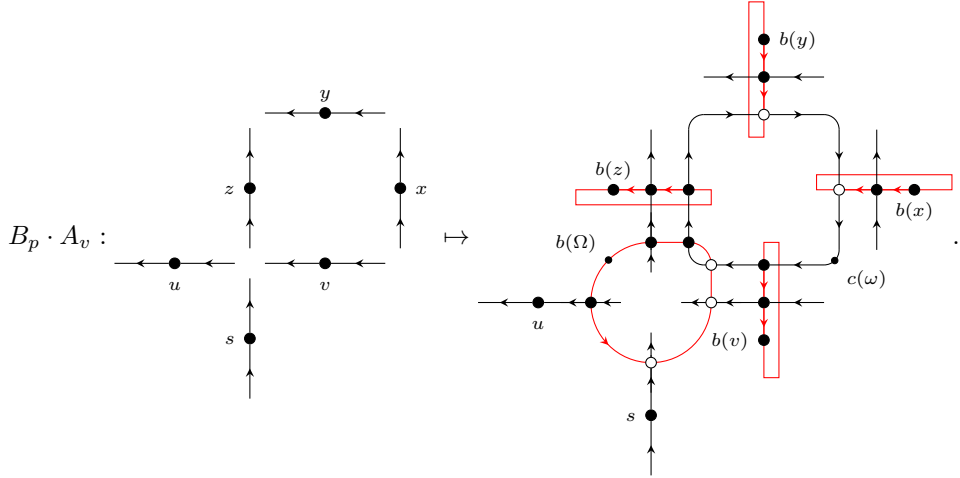
\begin{equation*}
  B_p \cdot A_v:
  \tikzsetnextfilename{d8dac855-c038-4d17-aa2f-d37a1e276175}
  %
.
\end{equation*}
The two tensors in the middle -- closed by arbitrary boundary condition $b(w)$ on the virtual index and after straightening the physical index -- read as:
\begin{equation}\label{eq:hopf_black_white_prod}
  \tikzsetnextfilename{48758cf5-9f84-4c74-8d88-35346f72f315}
  %
   \ .
\end{equation}
Using this identity in the expression for $B_p A_v$, we obtain that the two loops, the virtual and physical one, can be untangled, and thus $B_p A_v = A_v B_p$, i.e.\ the Hamiltonian terms $A_v$ and $B_p$ commute.

Let us finally show that the state $\ket{\Psi}$  is a ground state of the Hamiltonian $H$. First, as $\Omega$ is a projector, $\ket{\Psi}$ is clearly invariant under each term $A_v$. Let us now show that it is also invariant under all $B_p$, then this will mean that $\ket{\Psi}$ is a frustration-free ground state of $H$. To see that $\ket{\Psi}$ is invariant under the action of $B_p$, note that the state, locally around a plaquette, looks like
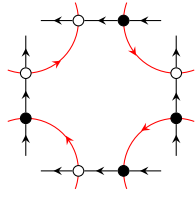
\begin{equation*}
  \tikzsetnextfilename{7cbeca65-2580-4166-8c04-45af17c3e411}

\end{equation*}
The action of $B_p$ is thus
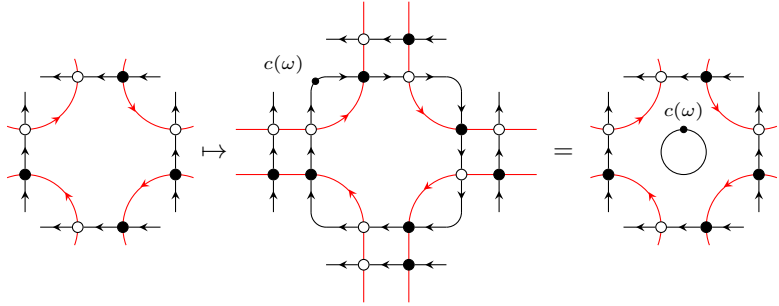
\begin{equation*}
  \tikzsetnextfilename{428c8867-fc01-447b-bab2-0b6b6c8f9996}
  %
\end{equation*}
Where in the last equation we have used \cref{eq:hopf_black_white_prod} in each corner. The result is therefore just a multiplication with the complex number $\omega(1) = 1$, i.e.\ the state $\ket{\Psi}$ is invariant under $B_p$.

\subsubsection{The Drinfeld double and anyons}

In this section we define the Drinfeld double $D(\mathcal{A})$ of a finite dimensional Hopf algebra $\mathcal{A}$ and, for each pair of plaquette $p$ and neighboring vertex $v$, a set of local operators including the Hamiltonian terms $A_v$ and $B_p$ that forms a representation of $D(\mathcal{A}^*)$. We also define a set of sates that differ only locally from the  state $\ket{\Psi}$. The defined local operators transform these states amongst each other, and thus these states form a $D(\mathcal{A}^*)$-module. We identify the anyons of the model as the subsets of states that form \emph{irreducible} $\mathcal{D}(\mathcal{A}^*)$-modules (see also \cite{Bultinck2017}).

Let us first define the Drinfeld double of a Hopf algebra.
\begin{definition}
  Let $\mathcal{A}$ be a finite dimensional Hopf algebra. The Drinfeld double $\mathcal{D}(\mathcal{A})$ is a Hopf algebra constructed as follows. As a vector space, it is $\mathcal{A}^* \otimes \mathcal{A}$. Given $f\in \mathcal{A}^*$ and $x\in\mathcal{A}$ we will write $f\bowtie x$ for their tensor product. The coproduct in $\mathcal{D}(\mathcal{A})$ is given by
  \begin{equation*}
    \Delta(f\bowtie x) = \sum_{(f),(x)} (f_{(2)} \bowtie x_{(1)}) \otimes (f_{(1)}\bowtie x_{(2)}).
  \end{equation*}
  The product in $\mathcal{D}(\mathcal{A})$ is given by
  \begin{equation}\label{eq:drinfeld_commutation}
    (f\bowtie x)\cdot (g\bowtie y) = \sum_{(g) (x)} g_{(1)}\circ S^{-1} (x_{(3)}) \cdot g_{(3)}(x_{(1)}) \cdot fg_{(2)} \bowtie x_{(2)}y.
  \end{equation}
\end{definition}

One can verify that the above product and coproduct indeed define a Hopf algebra. The unit of $\mathcal{D}(\mathcal{A})$ is $\epsilon\bowtie 1$. Linear functionals on $\mathcal{D}(\mathcal{A})$ are of the form $x\bowtie f$ ($x\in\mathcal{A}$ and $f\in\mathcal{A}^*$) and in particular, the counit is given by $1\bowtie \epsilon$. Finally, the antipode in $\mathcal{D}(\mathcal{A})$ is given by $S(f\bowtie x) = S^{-1}(f)\bowtie S(x)$. If $\mathcal{A}$ is a $C^*$-Hopf algebra, then $\mathcal{D}(\mathcal{A})$ is also a $C^*$-Hopf algebra with $*$ operation $f\bowtie x \mapsto f^* \bowtie x^*$. Note that, as $\sum_{(\epsilon)} \epsilon_{(1)}\otimes \epsilon_{(2)}\otimes \epsilon_{(3)} = \epsilon\otimes\epsilon\otimes\epsilon$, the map
\begin{equation*}
  \mathcal{A} \to \mathcal{D}(\mathcal{A}), \ x \mapsto \epsilon \bowtie x
\end{equation*}
is both a homomorphism and a cohomomorphism. Similarly, as $\sum_{(1)} 1_{(1)}\otimes 1_{(2)}\otimes 1_{(3)} = 1\otimes 1\otimes 1$, the map
\begin{equation*}
  \mathcal{A}^* \to \mathcal{D}(\mathcal{A}), \ f \mapsto f \bowtie 1
\end{equation*}
is a homomorphism and an anti-cohomomorphism. Using the images of these maps, all elements in the Drinfeld double can be written as
\begin{equation*}
  f\bowtie x = (f\bowtie 1) \cdot (\epsilon\bowtie x).
\end{equation*}
The elements $(f\bowtie 1)$ and $(\epsilon\bowtie x)$ satisfy the following commutation relation:
\begin{equation*}
  (\epsilon\bowtie x) \cdot (f\bowtie 1) = \sum_{(f)(x)} f_{(1)}\circ S^{-1} (x_{(3)}) \cdot f_{(3)}(x_{(1)}) \cdot (f_{(2)} \bowtie 1) \cdot (\epsilon \bowtie x_{(2)}).
\end{equation*}

Let $\mathcal{A}$ now be a finite dimensional Hopf algebra and $\mathcal{A}^*$ be its dual Hopf algebra. Let us now construct the Drinfeld double of $\mathcal{A}^*$. As a vector space, it is $\mathcal{A}\otimes\mathcal{A}^*$, and thus (as a vector space) it is canonically isomorphic to $\mathcal{D}(\mathcal{A})$. Let us make use of this isomorphism and write the elements of $\mathcal{D}(\mathcal{A}^*)$ as $f\bowtie x$ instead of $x\bowtie f$. In this notation, the coproduct of the Drinfeld double $\mathcal{D}(\mathcal{A}^*)$ is given by
\begin{equation*}
    \Delta(f\bowtie x) = \sum_{(f),(x)} (f_{(1)} \bowtie x_{(2)}) \otimes (f_{(2)}\bowtie x_{(1)}),
\end{equation*}
and the product is given by
  \begin{equation}\label{eq:drinfeld_a_star_product}
    (f\bowtie x)\cdot (g\bowtie y) = \sum_{(f) (y)} f_{(3)}\circ S^{-1} (y_{(1)}) \cdot f_{(1)}(y_{(3)}) \cdot f_{(2)}g \bowtie xy_{(2)}.
  \end{equation}
Again, the maps
\begin{equation*}
  \mathcal{A} \to \mathcal{D}(\mathcal{A}^*): \ x \mapsto \epsilon \bowtie x \quad \text{and} \quad \mathcal{A}^* \to \mathcal{D}(\mathcal{A}^*):\  f\mapsto f\bowtie 1
\end{equation*}
are homomorphisms. Similar as above, every element of $\mathcal{D}(\mathcal{A}^*)$ can be written as
\begin{equation*}
  f\bowtie x = (\epsilon\bowtie x) \cdot (f\bowtie 1).
\end{equation*}
The elements $(\epsilon\bowtie x)$ and $(f\bowtie 1)$ of $\mathcal{D}(\mathcal{A}^*)$ satisfy the commutation relation
  \begin{equation}\label{eq:drinfeld_a_star_commutation}
    (f\bowtie 1)\cdot (\epsilon\bowtie x) = \sum_{(f) (x)} f_{(3)}\circ S^{-1} (x_{(1)}) \cdot f_{(1)}(x_{(3)}) \cdot (\epsilon  \bowtie x_{(2)}) \cdot (f_{(2)}\bowtie 1).
  \end{equation}

Let us now define a set of local operators acting on an $\mathcal{A}$-injective PEPS and show that they form a representation of the Drinfeld double $\mathcal{D}(\mathcal{A}^*)$. All of these operators will act on the particles surrounding a neighboring plaquette and vertex pair $(p,v)$. We will define two types of operators. The first type is denoted by $A_{(p,v)}^w$ for any $w\in\mathcal{A}$, it acts only on the particles surrounding the vertex and represents the element $\epsilon\bowtie w\in \mathcal{D}(\mathcal{A}^*)$. The second type is denoted by $B_{(p,v)}^f$ for any $f\in\mathcal{A}^*$, it acts only only on the particles surrounding the plaquette and it represents the element $f\bowtie 1\in\mathcal{D}(\mathcal{A}^*)$.

Let us first define the operators $A_{(p,v)}^w$ for a given plaquette $p$, vertex $v$ and algebra element $w\in\mathcal{A}$. We define such an operator by its action on the four particles around the vertex; this action is defined by the graphical representation, using the injective representations $\phi_{(v,w)}$ on each edge. In the figure below, the plaquette $p$ is in the upper right corner of the vertex. The action of $A_{(p,v)}^w$ is given by
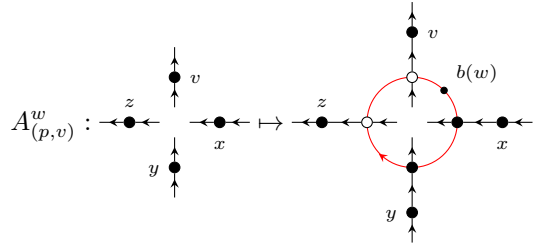
\begin{equation*}
  A^w_{(p,v)} :
  \tikzsetnextfilename{a96d713a-d44c-4eb4-8070-53bca94b4ae9}
  \begin{tikzpicture}[baseline=-1mm]
    \def\r{0.6};
    \foreach \d/\c/\l in {0/1/x,90/0/v,180/0/z,270/1/y}{
      \node[tensor,label = \d-90: $\l$] (r\d) at (\d:\r) {};
      \ifnum\c=0 \def\mystyle{->-};\else \def\mystyle{-<-};\fi
      \draw[\mystyle] (r\d) --++ (\d:0.4);
      \ifnum\c=0 \def\mystyle{-<-};\else \def\mystyle{->-};\fi
      \draw[\mystyle] (r\d) --++ (\d+180:0.4);
     }
  \end{tikzpicture}
  \mapsto
  \tikzsetnextfilename{04f99e60-78c9-4ede-9e23-28d0e73c679c}
  \begin{tikzpicture}[baseline=-1mm]
    \def\r{0.6};
    \draw[virtual,-<-=0.65] (0,0) circle (\r);
    \node[t,label=45:$b(w)$] at (45:\r) {};
    \foreach \d/\c/\l in {0/1/x,90/0/v,180/0/z,270/1/y}{
      \node[tensor, fill={\ifnum\c=0 white\else black\fi}] (r\d) at (\d:\r) {};
      \node[tensor,label = \d-90: $\l$] (t\d) at (\d:2*\r) {};
      \ifnum\c=0 \def\mystyle{->-};\else \def\mystyle{-<-};\fi
      \draw[\mystyle] (r\d) --++ (t\d);
      \draw[\mystyle] (t\d) --++ (\d:0.4);
      \ifnum\c=0 \def\mystyle{-<-};\else \def\mystyle{->-};\fi
      \draw[\mystyle] (r\d) --++ (\d+180:0.4);
     }
  \end{tikzpicture} \ .
\end{equation*}
Note that the MPO representing this operator is oriented in the opposite way as in the definition of the PEPS and in the definition of the vertex term of the Hamiltonian.
This expression reads as
\begin{equation*}
  A^w_{(p,v)} : x\otimes y\otimes z\otimes v \mapsto \sum_{(w)} w_{(4)} \cdot x   \otimes w_{(3)} \cdot y \otimes  z \cdot S(w_{(2)})\otimes v\cdot S( w_{(1)}) \ .
\end{equation*}
Here and in what follows we omit the representations $\phi$ when translating between the figures and the algebraic formulas. It is easy to check that the operators $A_{(p,v)}^w$ form a representation of $\mathcal{A}$ and that this representation is a $*$-representation. Notice that the vertex term $A_v$ is exactly $A_{(p,v)}^\Omega$, because the Haar integral $\Omega$ is invariant under the antipode $S$:
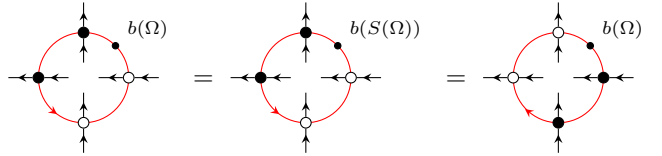
\begin{equation*}
  \tikzsetnextfilename{29f48ecf-0929-4da0-af9f-191544fd03fa}
  \begin{tikzpicture}[baseline=-1mm]
    \def\r{0.6};
    \draw[virtual,->-=0.65] (0,0) circle (\r);
    \node[t,label=45:$b(\Omega)$] at (45:\r) {};
    \foreach \d/\c/\l in {0/1/y,90/0/x,180/0/v,270/1/z}{
      \node[tensor, fill={\ifnum\c=0 black\else white\fi}] (r\d) at (\d:\r) {};
      \ifnum\c=0 \def\mystyle{->-};\else \def\mystyle{-<-};\fi
      \draw[\mystyle] (r\d) --++ (\d:0.4);
      \ifnum\c=0 \def\mystyle{-<-};\else \def\mystyle{->-};\fi
      \draw[\mystyle] (r\d) --++ (\d+180:0.4);
     }
  \end{tikzpicture} \  = \
  \tikzsetnextfilename{71b03d64-5b0a-47ad-996e-8c85e1b9c044}
  \begin{tikzpicture}[baseline=-1mm]
    \def\r{0.6};
    \draw[virtual,->-=0.65] (0,0) circle (\r);
    \node[t,label=45:$b(S(\Omega))$] at (45:\r) {};
    \foreach \d/\c/\l in {0/1/y,90/0/x,180/0/v,270/1/z}{
      \node[tensor, fill={\ifnum\c=0 black\else white\fi}] (r\d) at (\d:\r) {};
      \ifnum\c=0 \def\mystyle{->-};\else \def\mystyle{-<-};\fi
      \draw[\mystyle] (r\d) --++ (\d:0.4);
      \ifnum\c=0 \def\mystyle{-<-};\else \def\mystyle{->-};\fi
      \draw[\mystyle] (r\d) --++ (\d+180:0.4);
     }
  \end{tikzpicture} \ = \
  \tikzsetnextfilename{0f34cea0-19fb-41f0-842a-75dec18e1bd1}
  \begin{tikzpicture}[baseline=-1mm]
    \def\r{0.6};
    \draw[virtual,-<-=0.65] (0,0) circle (\r);
    \node[t,label=45:$b(\Omega)$] at (45:\r) {};
    \foreach \d/\c/\l in {0/1/y,90/0/x,180/0/v,270/1/z}{
      \node[tensor, fill={\ifnum\c=0 white\else black\fi}] (r\d) at (\d:\r) {};
      \ifnum\c=0 \def\mystyle{->-};\else \def\mystyle{-<-};\fi
      \draw[\mystyle] (r\d) --++ (\d:0.4);
      \ifnum\c=0 \def\mystyle{-<-};\else \def\mystyle{->-};\fi
      \draw[\mystyle] (r\d) --++ (\d+180:0.4);
     }
  \end{tikzpicture} \ ,
\end{equation*}
where the last equation is \cref{eq:S_x_MPO}. The plaquette term $A_v$ is the multiplication with the MPO on the left, while $A_{(p,v)}^\Omega$ is multiplication with the MPO on the right. As $\Omega$ is cocommutative, the particular choice of the plaquette $p$ does not matter, i.e.\ $A_v = A_{(p,v)}^\Omega$ holds for any choice of the plaquette $p$ next to the vertex $v$.

Let us now define the operator $B_{(p,v)}^f$ for the plaquette $p$, vertex $v$ and linear functional $f\in\mathcal{A}^*$. In the figure below, the vertex $v$ is in the lower left corner of the plaquette $p$. The action of $B_{(p,v)}^f$ is given by
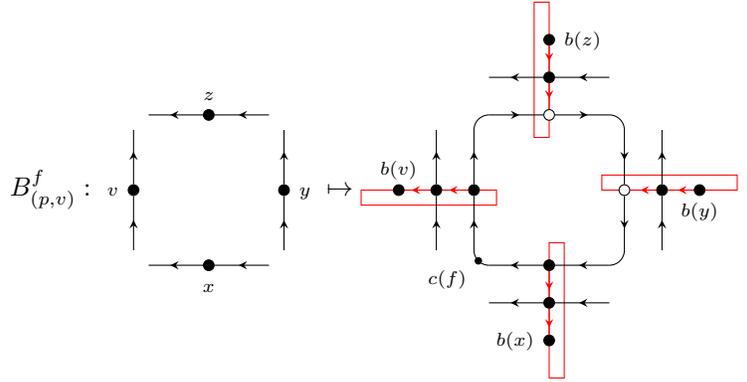
\begin{equation*}
  B_{(p,v)}^f :
  \tikzsetnextfilename{9339ddf6-d53c-44af-a3b9-335b1495e78b}
  \begin{tikzpicture}
    \foreach \d/\c/\l in {0/0/y,90/0/z,180/1/v,270/1/x}{
      \node[tensor,label=\d:$\l$] (t\d) at (\d:1) {};
      \ifnum\c=0 \def\mystyle{->-};\else \def\mystyle{-<-};\fi
      \draw[\mystyle] (t\d) --++ (\d+90:0.8);
      \ifnum\c=0 \def\mystyle{-<-};\else \def\mystyle{->-};\fi
      \draw[\mystyle] (t\d) --++ (\d-90:0.8);
    }
  \end{tikzpicture}
  \mapsto
  \tikzsetnextfilename{290d8ce5-0bb0-460f-99b4-599fef684ca2}
  \begin{tikzpicture}
    \foreach \d/\c/\l in {0/0/y,90/0/z,180/1/v,270/1/x}{
      \draw[virtual] (\d:2.5) rectangle ($(\d:0.7)+(\d+90:0.2)$);
      \node[tensor,fill={\ifnum\c=1 black\else white\fi}] (t\d) at (\d:1) {};
      \node[tensor] (r\d) at (\d:1.5) {};
      \node[tensor,label = \d-90:$b(\l)$] (b\d) at (\d:2) {};
      \ifnum\c=0 \def\mystyle{->-};\else \def\mystyle{-<-};\fi
      \draw[virtual,\mystyle] (b\d) -- (r\d);
      \draw[virtual,\mystyle] (r\d) -- (t\d);
      \draw[-<-] (t\d) --++ (\d+90:0.8) coordinate (c);
      \draw (c) arc (\d:\d+90:0.2);
      \draw[\mystyle] (r\d) --++ (\d+90:0.8);
      \ifnum\c=0 \def\mystyle{-<-};\else \def\mystyle{->-};\fi
      \draw[->-] (t\d) --++ (\d-90:0.8);
      \draw[\mystyle] (r\d) --++ (\d-90:0.8);
    }
    \node[t,label=-135:$c(f)$] at ($(-0.8,-0.8) + (-135:0.2)$) {};
  \end{tikzpicture} \ ,
\end{equation*}
or equivalently, by
\begin{equation*}
  B_{(p,v)}^f: x\otimes y \otimes z \otimes v \mapsto \sum_{xyzv} f(x_{(2)} S(y_{(1)}) S(z_{(1)}) v_{(2)} ) \cdot x_{(1)} \otimes y_{(2)} \otimes z_{(2)} \otimes v_{(1)}.
\end{equation*}
It is easy to check that these operators form a $*$-representation of $\mathcal{A}^*$. The plaquette term $B_p$ of the Kitaev Hamiltonian is exactly the operator $B_{(p,v)}^\omega$, where $\omega$ is the Haar integral of $\mathcal{A}^*$. As $\omega$ is cocommutative, the particular choice of $v$ does not matter, i.e.\  $B_p = B_{(p,v)}^\omega$ for any vertex $v$ around the plaquette $p$.

Let us now define a linear map $\mathcal{D}(\mathcal{A}^*) \to \End(\mathcal{A}^6)$ for a given pair of plaquette $p$ and vertex $v$ by
\begin{equation}\label{eq:Drinfeld_representation}
  f\bowtie w \mapsto A_{(p,v)}^w \cdot B_{(p,v)}^f ,
\end{equation}
for all $f\in\mathcal{A}^*$ and $w\in\mathcal{A}$, and by linear extension, on the whole $\mathcal{D}(\mathcal{A})$. Below we show that this map defines a representation of the Drinfled double $\mathcal{D}(\mathcal{A})$ (see also \cite{Buerschaper2013}). As the maps
\begin{equation*}
  f\bowtie \epsilon \mapsto B_{(p,v)}^f  \quad \text{and} \quad   \epsilon\bowtie w \mapsto A_{(p,v)}^w
\end{equation*}
form representations of $\mathcal{A}^*$ and $\mathcal{A}$, respectively, and $(f\bowtie w) = (\epsilon \bowtie w) \cdot (f\bowtie 1)$, we only have to check that the commutation relation \cref{eq:drinfeld_a_star_commutation} holds. Let us therefore compare the action of $A_{(p,v)}^w\cdot B_{(p,v)}^f$ and the action of $B_{(p,v)}^f \cdot A_{(p,v)}^w$. The graphical representation of the action of the operator $A_{(p,v)}^w\cdot B_{(p,v)}^f$ is the following:
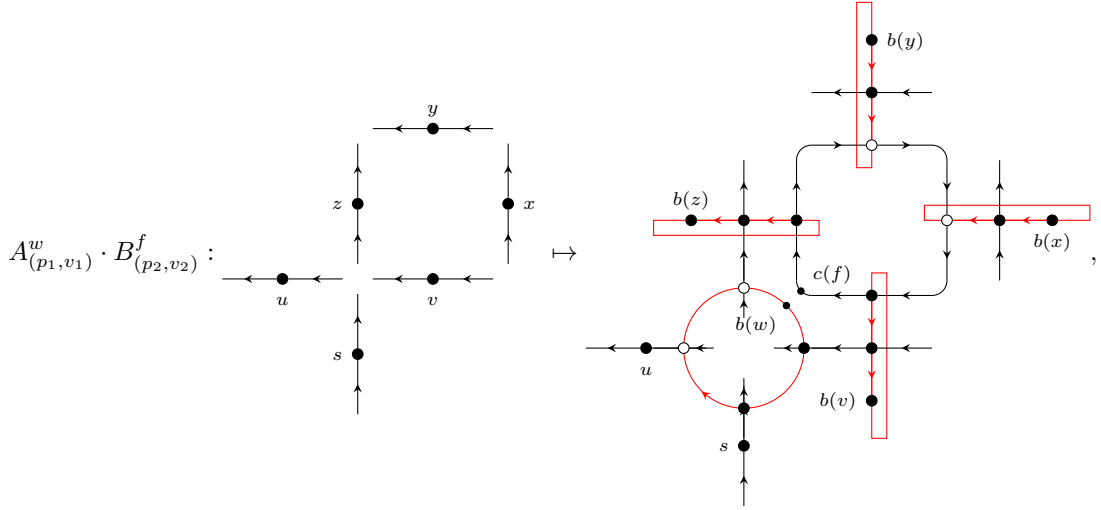
\begin{equation*}
  A_{(p_1,v_1)}^w\cdot B_{(p_2,v_2)}^f :
  \tikzsetnextfilename{a560c3b2-27aa-4dee-8a94-ced52c25022b}
,
\end{equation*}
where $v$ is the vertex with all four edges drawn, and $p$ is the plaquette with all four bordering edges drawn. Let us now depict the action of the operator $B_{(p,v)}^f \cdot A_{(p,v)}^w $. To do that, note that as we act first with the operator $A_{(p,v)}^w$, the coproduct in the operator $B_{(p,v)}^f$ also applies ot the MPO representation of $w$. The graphical representation of this action is thus
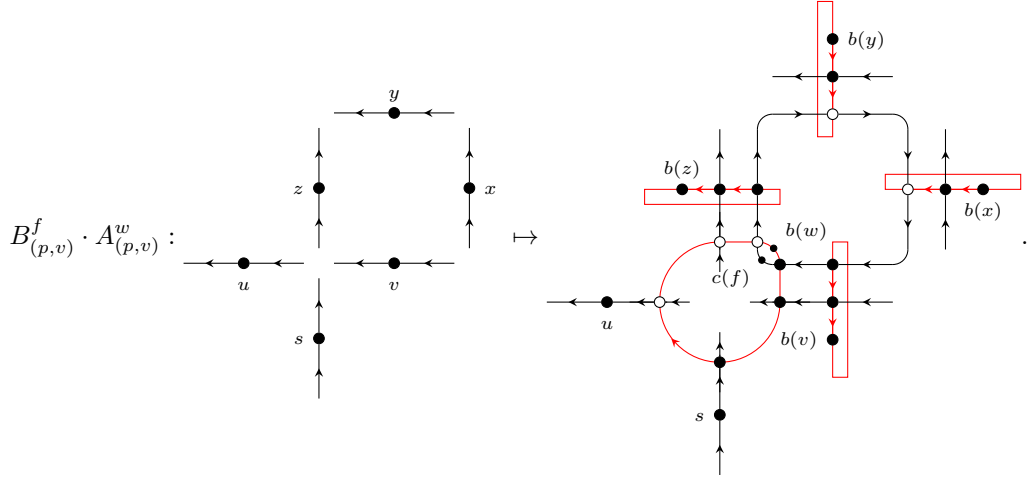
\begin{equation*}
  B_{(p,v)}^f \cdot A_{(p,v)}^w:
  \tikzsetnextfilename{0ac95b49-8fc4-4962-8667-376a6548b56d}
  %
.
\end{equation*}
This operator describes the action of the operator
\begin{equation*}
  \sum_{(f)(w)} f_{(3)}\circ S( w_{(1)}) \cdot f_{(1)}(w_{(3)}) \cdot A_{(p,v)}^{w_{(2)}} \cdot B_{(p,v)}^{f_{(2)}},
\end{equation*}
and thus, as this commutation relation is the same\footnote{remember that $\mathcal{A}$ is a $C^*$-Hopf algebra, and thus $S^{-1} = S$} as \cref{eq:drinfeld_a_star_commutation}, the map in \cref{eq:Drinfeld_representation} describes a representation of the Drinfeld double $\mathcal{D}(\mathcal{A}^*)$.

Let us now construct a  set of states that differ  from $\ket{\Psi}$ only on a given pair of plaquette $p$ and vertex $v$. The states, in the bulk, are given by
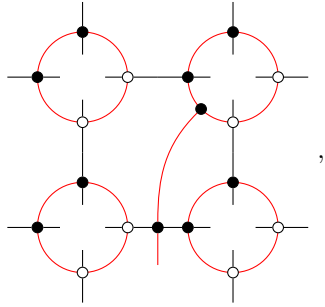
\begin{equation*}
  \tikzsetnextfilename{90b0cc67-faf9-48da-b944-83ea667e316a}
  \begin{tikzpicture}
    \foreach \x/\y in {0/0,0/2,2/0,2/2}{
      \draw[virtual] (\x,\y) circle (0.6);
      \foreach \d/\c in {0/0,90/1,180/1,270/0}{
        \node[tensor,shift={(\x,\y)},fill={\ifnum\c=0 white\else black\fi}] (t) at (\d:0.6) {};
        \draw (t)--++(\d:0.4);
        \draw (t)--++(\d+180:0.3);
      }
    }
    \node[tensor] (s) at (1,0) {};
    \node[tensor,shift={(2,2)}] (t) at (-135:0.6) {};
    \draw[virtual] (t) to[in=90,out=-135] (s) --++ (0,-0.5);
  \end{tikzpicture} \ ,
\end{equation*}
i.e.\ at the position defined by the vertex and plaquette we insert a rank-three tensor in $W \otimes \psi(\mathcal{A}) \subset W\otimes \End(W)$, where $W$ is the virtual vector space of the MPO and $\psi$ is the corresponding representation of $\mathcal{A}^*$, and we continue with the MPO starting from the given point. If the state is defined by open boundary, then the boundary has one more index than the boundary describing the ground state $\ket{\Psi}$ of the PEPS. If the state is defined on closed boundary, then the string has to terminate at some point; at this termination we will insert another rank-three tensor the same way as above. That is, on closed boundary, we do not define a state with a single defect in it, instead, only states that differ from $\ket{\Psi}$ at least in two different positions (note, however, that one can construct periodic boundary states with an odd number of defects).

Both operators $B_{(p,v)}^f$ and $A_{(p,v)}^w$ map a state with a defect at the pair of plaquette $p$ and vertex $v$ to another state with a defect at the same position, but this state might be described by a different rank-three tensor. The action of the operators $B_{(p,v)}^f$ and $A_{(p,v)}^w$ on the rank-three tensor can be depicted as
\begin{equation*}
  B_{(p,v)}^f: \
  \tikzsetnextfilename{d2d79f49-52a0-448d-b980-0ab7f158b3af}
  \begin{tikzpicture}
    \draw[virtual] (0,0) arc (-135:-90:0.8);
    \draw[virtual] (0,0) arc (-135:-180:0.8);
    \node[tensor] (t) at (0,0) {};
    \draw[virtual] (t) --++ (-135:0.8);
  \end{tikzpicture} \mapsto
  \tikzsetnextfilename{e364c450-0b9b-4be0-856c-2b71f0dbe40a}
  \begin{tikzpicture}
    \draw[virtual] (0,0) arc (-135:-90:0.8);
    \draw[virtual] (0,0) arc (-135:-180:0.8);
    \node[tensor] (t) at (0,0) {};
    \draw[virtual] (t) --++ (-135:0.8);
    \draw (0,0) circle (0.4);
    \node[tensor] at (-135:0.4) {};
    \node[tensor] at (120:0.4) {};
    \node[tensor,fill=white] at (-30:0.4) {};
    \node[t,label=above right:$c(f)$] at (45:0.4)  {};
  \end{tikzpicture}
  \quad \text{and} \quad
  A_{(p,v)}^w :
  \tikzsetnextfilename{2000615f-d33c-4a29-873e-a8a03686e313}
  \begin{tikzpicture}
    \node[tensor] (t) at (0,0) {};
    \draw[virtual] (t) --++ (0.5,0);
    \draw[virtual] (t) --++ (-0.5,0);
    \draw[virtual] (t) --++ (0,-0.5);
  \end{tikzpicture}
  \mapsto
  \sum_{\mu}
  \tikzsetnextfilename{76653b83-2c88-40ad-9a47-8003471f1bd0}
  \begin{tikzpicture}
    \node[tensor] (t) at (0,0) {};
    \draw[virtual] (t) --++ (0.5,0);
    \draw[virtual] (t) --++ (-0.5,0);
    \draw[virtual] (t) --++ (0,-0.5);
    \node[t,label=above:$b(w)$] (b) at (0,0.5) {};
    \draw[virtual] (b) --++ (0.5,0);
    \draw[virtual] (b) --++ (-0.5,0);
    \node[fusion tensor] (v) at (0.5,0.25) {};
    \node[fusion tensor] (w) at (-0.5,0.25) {};
    \draw[virtual] (v.east) --++ (0.5,0);
    \draw[virtual] (w.west) --++ (-0.5,0);
    \node[anchor=north,inner sep=3pt,font=\tiny] at (v.south) {$\mu$};
    \node[anchor=north,inner sep=3pt,font=\tiny] at (w.south) {$\mu$};
  \end{tikzpicture}
\end{equation*}
Or, by formulas, if the tensor is given by $\ket{v}\otimes b(x)$, then
\begin{gather*}
  B_{(p,v)}^f: \ket{v}\otimes b(x) \mapsto \sum_{(f)} \psi(f_{(2)}) \ket{v} \otimes f_{(1)}(x_{(1)}) \cdot f_{(3)}\circ S(x_{(3)})\cdot  b(x_{(2)}),\\
  A_{(p,v)}^w: \ket{v}\otimes b(x) \mapsto \ket{v} \otimes b(xw).
\end{gather*}
These states thus form a $\mathcal{D}(\mathcal{A}^*)$-module. We identify the \emph{anyons} of the model as the irrep sectors of this module, i.e.\ an anyon is a set of states, each of which are locally differ from the ground state of the Hamiltonian and such that the above defined local operators do not mix the different anyons. Note that as the Hamiltonian, in general, is not in the center of $\mathcal{D}(\mathcal{A}^*)$, an anyon might not have a definite energy (i.e.\ the different states in the anyonic sector might have different energies, see also \cite{Komar2017}).

\section*{Acknowledgments}

We thank Laurens Lootens, Gabriella Böhm and Frank Verstraete for inspiring discussions. This work has received support from the European Union’s Horizon 2020 program through the ERC CoG SEQUAM (No. 863476), ERC CoG GAPS (No. 648913) and the ERC AdG QENOCOBA (No. 742102), from the DFG (German Research Foundation) under Germany’s Excellence Strategy (EXC-2111-390814868), from the Spanish Ministry of Science and Innovation through the Agencia Estatal de Investigación MCIN/AEI/10.13039/501100011033 (PID2020-113523GB-I00 and grant BES-2017-081301 under the ``Severo Ochoa Programme for Centres of Excellence in R\&D'' CEX2019-000904-S and ICMAT Severo Ochoa project SEV-2015-0554), from CSIC Quantum Technologies Platform PTI-001, from Comunidad Autónoma de Madrid through the grant QUITEMAD-CM (P2018/TCS-4342)

\printbibliography

\appendix

\section{Proof of \cref{lem:foldback} and of \cref{cor:T}}\label{sec:foldback_lemma_proof}

Before proving \cref{lem:foldback}, we need the following lemma:

\begin{lemma}\label{lem:important}
  In a cosemisimple WHA, the following holds for all $a,b,c\in \irr(\mathcal{A}^*)$:
  \begin{equation*}
    \sum_{\mu}
    \tikzsetnextfilename{9a8f9640-8375-46b5-aecb-e0722b08c427}
  \ ,
  \end{equation*}
  where the constants $w_c$ and $w_a$ are defined in \cref{eq:w_a}.
\end{lemma}
\begin{proof}
  Using \cref{eq:w_a}, one can write
  \begin{align*}
    \frac{1}{w_c} \sum_{\mu}
    \tikzsetnextfilename{65fe7e3f-3f82-4b36-861f-3fb48d3bb5ae}
    %
  \ ,
  \end{align*}
  where in the second equality we have used that $N_{dc}^{r_c} = N_{r_c\bar{c}}^{d}$ (see \cref{eq:N_symmetries}) and that $r_c = l_{\bar{c}}$ (see \cref{thm:N_symmetry_consequence}), and thus that the sum over $\mu$ on the r.h.s.\ runs from 1 to $N_{dc}^{r_c} = N_{r_c\bar{c}}^{d} = N_{l_{\bar{c}}\bar{c}}^{d} = \delta_{d\bar{c}}$, i.e.\ it is an empty sum if $d\neq \bar{c}$ and it consists of a single term if $d=\bar{c}$. Due to associativity (\cref{eq:associative}), the r.h.s.\ can be further rewritten as
  \begin{align*}
    \sum_{d\mu\nu}
    \tikzsetnextfilename{838902c5-f691-467b-9978-46cf1d710827}
,
  \end{align*}
  where in the first equation we have used that $N_{\bar{a}d}^e = N_{ae}^d$ (see \cref{eq:N_symmetries}), and that it is non-zero for $e\in E$ if and only if $e=r_{a}$, and in this case it is $\delta_{da}$ (see \cref{eq:N_unit}); in the second equation we have used again \cref{eq:w_a}. Finally note that $N_{\bar{a}\bar{b}}^{\bar{c}} = N_{\bar{b}c}^{a}$ is non-zero only if $r_{c} = r_a$, and thus one can drop the prefactor $\delta_{r_cr_a}$. Rearranging the resulting equation leads to the desired result.
\end{proof}

Let us now restate and prove \cref{lem:foldback}:
\Blemma*
\begin{proof}
  Let us apply \cref{eq:fold-back} three times. The first application yields
  \begin{equation*}
    \tikzsetnextfilename{0c0a5742-5420-4a57-b74c-ab8bb6b6f340}
 .
  \end{equation*}
  The second application yields
  \begin{equation*}
    \tikzsetnextfilename{0e23f04a-751b-42fd-872f-c6fe8eae5437}
    %
 .
  \end{equation*}
  Finally, the third application yields
  \begin{equation*}
    \tikzsetnextfilename{41e73019-03b5-44f8-9b85-e610f1c35fd6}
    %
 .
  \end{equation*}
  We have thus obtained that there exists an invertible matrix $Y_{ab}^c$ such that
  \begin{equation*}
    \tikzsetnextfilename{7fd38fea-7af9-4fc2-b1c3-ebc91b62bf28}
    %
 \ ,
  \end{equation*}
  and in fact, $Y_{ab}^c = C_{\bar{b}\bar{a}}^{\bar{c}} \hat{C}_{b\bar{c}}^{\bar{a}} C_{\bar{a}c}^b$. Similarly, by three consecutive applications of the inverse relations of \cref{thm:foldback},
  \begin{equation*}
    \tikzsetnextfilename{fac63e27-7116-4b57-afe5-c8fc6fba05f8}
    %
 \ ,
  \end{equation*}
  we conclude that there is an invertible matrix $X_{ab}^c$ such that
  \begin{equation*}
    \tikzsetnextfilename{b236f76e-db6d-4d4d-ba7c-02c71e65cf1a}
    %
 \ .
  \end{equation*}
  In fact, $X_{ab}^c = (C_{\bar{a}c}^b)^{-1}  (\hat{C}_{b\bar{c}}^{\bar{a}})^{-1} (C_{\bar{b}\bar{a}}^{\bar{c}})^{-1}$, i.e.\ $X_{ab}^c Y_{ab}^{c} = \id$. Using similar arguments, we conclude that there are invertible matrices $\hat{X}$ and $\hat{Y}$ such that
  \begin{equation*}
    \tikzsetnextfilename{eb9fa399-19e0-4200-b470-45bb5a776e2f}
    %
 \ .
  \end{equation*}
  Just as above, $\hat{X}_{ab}^c \hat{Y}_{ab}^c = \id$. Moreover, by orthogonality of the fusion tensors, $\hat{Y}_{ab}^c = (Y_{ab}^c)^T$:
  \begin{equation*}
    \left(Y_{ab}^c\right)_{\mu\nu} \cdot
    \tikzsetnextfilename{9dc3429e-88ee-4de5-9304-3c3325a959c2}
    %
 ,
  \end{equation*}
  where in the first equality we have used the definition of $Y_{ab}^c$ and the orthogonality of the fusion tensors, while in the second the definition of $\hat{Y}_{ab}^c$ and again the orthogonality of the fusion tensors. By a similar argument, we obtain $\hat{X}_{ab}^c = (X_{ab}^c)^T$ as well.

  By definition of the linear functional $g$, we obtain that
  \begin{gather}\label{eq:3g_expand}
    %
 \ .
  \end{gather}
  We have thus obtained that \cref{eq:B} holds with
  \begin{equation}\label{eq:B_XY}
    B_{ab}^c = \frac{d_ad_b}{d_c} \frac{w_{\bar{c}}}{w_{\bar{a}}w_{\bar{b}}} \cdot \hat{Y}_{ab}^c X_{\bar{b}\bar{a}}^{\bar{c}} .
  \end{equation}
Combining the two equations in \cref{eq:B}, we obtain that
\begin{equation}\label{eq:B_2_id_intermediate}
  \sum_{\mu}
  %
 \ .
\end{equation}
On the other hand, using \cref{lem:important} three times, we obtain that
\begin{equation*}
  \sum_{\mu}
  %
 \ ,
\end{equation*}
or equivalently, as $(d_a/w_{\bar{a}})^2 = w_{a}/w_{\bar{a}}$, that
\begin{equation*}
  \sum_{\mu}
  %
 \ .
\end{equation*}
This equation can be rearranged as
\begin{equation*}
  \sum_{\mu}
  %
 \ ,
\end{equation*}
and therefore, comparing this equation to \cref{eq:B_2_id_intermediate}, we obtain that $\left(B_{ab}^c\right)^2 = \id$. Finally, using \cref{eq:B_XY} and then \cref{lem:important}, we can write
\begin{equation*}
  \sum_{\mu\nu} \left(B_{ab}^c\right)_{\mu\nu}
  %
 \ .
\end{equation*}
Part of the last statement follows using that $d_a/w_{\bar{a}} = w_a/d_a$. Finally, using again \cref{lem:important}, we conclude that
 \begin{equation*}
   \sum_{\mu}  w_a \
   %
    \ .
 \end{equation*}
\end{proof}

Let us restate and prove now \cref{cor:T}:
\Tlemma*

\begin{proof}
Using \cref{thm:foldback},
\begin{equation*}
  \sum_{\kappa} w_a
  \tikzsetnextfilename{49d5009e-baec-46b4-bcdb-98c16383e96d}
  %
 \ .
\end{equation*}
Therfore, using \cref{lem:foldback}, we obtain that $B_{ab}^c$ can be expressed as
\begin{equation}\label{eq:B_C_connection}
  B_{ab}^c = w_a \cdot \frac{d_c}{d_a d_b} \cdot \left(C_{\bar{a}c}^b\right)^T \left(C_{ab}^c\right)^{-1}.
\end{equation}
As $B_{ab}^c$ squares to the identity, the eigenvalues of $B_{ab}^c$ are $\pm 1$, and thus $T_{ab}^c$, the trace of $B_{ab}^c$, is an integer with $|T_{ab}^c|\leq N_{ab}^c$. Note now that \cref{eq:B_C_connection} can be used to obtain the following expression for $B_{\bar{a}c}^b$:
\begin{equation*}
  B_{\bar{a}c}^b = w_{\bar{a}} \cdot \frac{d_b}{d_a d_c} \left(C_{ab}^c\right)^{T} \left(C_{\bar{a}c}^b\right)^{-1} = C_{\bar{a}c}^b \cdot \left[\left(B_{ab}^c\right)^{-1}\right]^T \cdot \left(C_{\bar{a}c}^b\right)^{-1} =  C_{\bar{a}c}^b \cdot \left(B_{ab}^c\right)^T \cdot \left(C_{\bar{a}c}^b\right)^{-1},
\end{equation*}
where in the first equation we have used that $d_a = d_{\bar{a}}$, in the second equation we have used that $w_{\bar{a}} d_b / (d_ad_c) = (w_{a} d_c /(d_a d_b))^{-1}$ and in the third that $B_{ab}^c$ squares to the identity. Taking the trace of the two sides in this equation, we obtain that $T_{\bar{a}c}^b = T_{ab}^c$. By a similar argument, we obtain that $T_{ab}^c = T_{c\bar{b}}^a$ as well. The definition of $T_{ab}^c$ and \cref{lem:foldback} implies that
\begin{equation*}
  \sum_b T_{ab}^c d_b \cdot \id_c=
  \sum_{b\mu}  d_b
  \tikzsetnextfilename{adf87f52-1e22-4bde-b1c1-bfdab6dc14e3}
 \ .
\end{equation*}
Using \cref{eq:w_a}, the fact that $d_a^2 = w_a w_{\bar{a}}$, and the counit axiom \cref{eq:counit_consequence}, the r.h.s.\ can be further written as
\begin{equation*}
  \sum_b T_{ab}^c d_b \cdot \id_c=
  \sum_{b\mu} d_ad_c
  \tikzsetnextfilename{8f8d6eeb-fbf3-4ab9-9dc8-58f289b816d6}
  %
 = \begin{cases}
      d_a d_c\cdot \id_c, & \text{if $l_a=l_c$,}\\
      0                & \text{otherwise.}
  \end{cases}
\end{equation*}
Therefore the equation $\sum_b T_{ab}^c d_b = d_a \cdot \delta_{l_a l_c} d_c$ holds.

Following the proof of Theorem 2.3 in \cite{Etingof2002} (see also \cite{Bartlett2015}), we can now prove that $d_a^2 = w_a w_{\bar{a}}$ is a positive number. Notice first that the matrix $T_a^T T_a = T_{\bar{a}} T_a$ is positive semidefinite (as $T_a$ is a real matrix), and thus all of its eigenvalues are non-negative. Let us show that $d_a^2 = w_a w_{\bar{a}}$ is one of its eigenvalues, then, as neither $w_a$ nor $w_{\bar{a}}$ is 0, this implies the the positivity of $d_a^2$. To see that it is an eigenvalue, we can check that $\delta_{r_a l_b}d_b$ is the corresponding eigenvector:
  \begin{equation*}
    \sum_{bc} T_{\bar{a}c}^d T_{ab}^c \delta_{r_a l_b} d_b =\sum_{bc} T_{\bar{a}c}^d T_{ab}^c d_b = d_a \sum_c T_{\bar{a}c}^d \delta_{l_a l_c} d_c = d_a \sum_c T_{\bar{a}c}^d d_c = d_a d_{\bar{a}} \delta_{r_a l_d} d_d,
  \end{equation*}
  where in the first equation we have used that  $T_{ab}^c = 0$ if $\delta_{r_al_b} = 0$ (as then $N_{ab}^c=0$ and $|T_{ab}^c| < N_{ab}^c$); the same relation is used in the third equation, together with $r_{\bar{a}} = l_a$. In the last equation we have used $r_a = l_{\bar{a}}$. Finally note that $d_{\bar{a}} = d_a$ and that the vector defined by $\delta_{r_a l_b} d_b$ is non-zero.

  Finally, let us prove that $\sum_{x:l_x = l_a} d_x^2 = \sum_{x:l_x =r_a} d_x^2$. For that, note that $T_{ab}^c = T_{\bar{a}c}^b = T_{b\bar{c}}^{\bar{a}}  = T_{\bar{b}\bar{a}}^{\bar{c}}$, and thus
  \begin{equation*}
    d_a \sum_{x:l_x = l_a} d_x^2 = \sum_{xb} T_{xb}^a d_x d_b = \sum_{xb} T_{\bar{b}\bar{x}}^{\bar{a}} d_x d_b = \sum_{\bar{x}\bar{b}} T_{\bar{b}\bar{x}}^{\bar{a}} d_{\bar{x}} d_{\bar{b}} = d_{\bar{a}} \sum_{\bar{b}:l_{\bar{b}} = l_{\bar{a}}} d_{\bar{b}}^2 = d_a \sum_{x:l_x = r_a} d_{x}^2 .
  \end{equation*}
  As $d_a\neq 0$, the statement follows.
\end{proof}

\section{Example: string-net models}

In the following we create WHAs from fusion categories that appear in the construction of string-net models. For simplicity, we will restrict to fusion categories where the fusion multiplicities are all $0$ or $1$, $N_{ab}^c \leq 1$ for all $a,b,c$. With this assumption, the pentagon equations for the F-symbols simplify to the equation:
\begin{equation}\label{eq:pentagon_simple}
  \left[F_{fcd}^e\right]_l^g \left[F_{abl}^e\right]_k^f = \sum_h \left[F_{abc}^g\right]_h^f \left[F_{ahd}^e\right]_k^g \left[F_{bcd}^k\right]_l^h.
\end{equation}
The MPO tensor used in the construction of the string-net models is then defined as
\begin{equation*}
  \tikzsetnextfilename{48b7458c-54b9-4c81-97c3-3509d56d5fb8}
 = \left[F_{abl}^e\right]_{k}^f \cdot \delta_{aa'} \delta_{bb'} \delta_{ll'} \delta_{ee'} \delta_{ff'} \delta_{kk'},
\end{equation*}
where the dotted lines serve as a visual reminder of the $\delta$ prefactors. We will often make use of this visual reminder of the $\delta$ prefactors and only write the non-zero components of the MPO tensor,
\begin{equation*}
  \tikzsetnextfilename{aecb4a03-e559-4e2f-a0eb-02fa4c47b2eb}
  %
 = \left[F_{abl}^e\right]_{k}^f .
\end{equation*}
Each line is labeled by the simple objects of the category, i.e.\ it is an $N$-dimensional vector space ($N$ is the number of simple objects in the fusion category); thus the bond dimension of this tensor is $N^3$. This bond dimension can be reduced as the tensor is block diagonal and it contains a zero block. The projectors reducing the tensor (those that correspond to a non-zero block) are
\begin{equation}\label{eq:stringnet_mpo_tensor}
  \tikzsetnextfilename{e66ea457-0c0b-43eb-8e56-a96fb661fa4a}
  %
  = N_{ab}^c \cdot \delta_{bk} \cdot \delta_{aa'} \delta_{bb'} \delta_{cc'}.
\end{equation}
These are indeed projectors as $N_{fl}^e\in \{0,1\}$ and it is easy to see that they commute with the MPO tensor (the $F$-symbols satisfy $\left[F_{abl}^e\right]_k^f=0$ unless $N_{fl}^e\neq 0$ and $N_{bl}^k\neq 0$). The bond dimension of each block is then $D_l = \sum_{fe} N_{fl}^e$. A similar decomposition holds for the vertical direction as well. There the projectors that decompose the MPO tensor are
\begin{equation*}
  \tikzsetnextfilename{52dc6821-24a9-4635-b084-e09384354da6}

  = N_{ab}^f \cdot \delta_{bk} \cdot \delta_{aa'} \delta_{bb'} \delta_{cc'}.
\end{equation*}
Starting from this MPO tensor, let us define a linear space $\mathcal{A}$ as
\begin{equation*}
  \mathcal{A} = \left\{ \
    \tikzsetnextfilename{98ca6f2a-8680-4470-866c-3aa3a42f7c9f}
    %
 \
    \middle| X \in \bigoplus_l \mathcal{M}_{D_l}
  \right\} .
\end{equation*}
This linear space $\mathcal{A}$ has a natural coalgebra structure, with the coproduct given by
\begin{equation*}
    \Delta\left( \
    \tikzsetnextfilename{c3004314-cb5e-47d7-8f7b-510139e75fd5}
    %
 \     .
\end{equation*}
As above, taking repeated coproducts of a coalgebra element is the same as growing the size of the MPO. In particular, the operation $\Delta$ defined by this equation is associative. Note that as the pentagon equation \cref{eq:pentagon_simple} can be rearranged as
\begin{equation*}
  \left[F_{ahd}^e\right]_k^g \left[F_{bcd}^k\right]_l^h = \sum_f \left[F_{fcd}^e\right]_l^g \left[F_{abl}^e\right]_k^f \left[\left(F_{abc}^g\right)^{-1}\right]_f^h,
\end{equation*}
the coproduct can also be expressed as
\begin{equation*}
  \tikzsetnextfilename{9c9aa475-6ac1-40b8-855b-208da1ea7910}
,
\end{equation*}
where we have defined the fusion tensors of the physical level as
\begin{equation*}
  \tikzsetnextfilename{190fe9f9-b9e6-4fa6-82c8-88ab637c1459}
  %
 = \left[\left(F_{abc}^g\right)^{-1}\right]_f^h \cdot \delta_{ff'} \delta_{gg'} \delta_{ee'}.
\end{equation*}
As the construction is built on a fusion category, there is a unique vacuum label that we denote by $1$. The counit of this coalgebra is given by
\begin{equation*}
    \epsilon : \
    \tikzsetnextfilename{280cedd8-7c56-440f-9ba7-7ed0ba0c6ba5}
    %
 \ ,
\end{equation*}
where we have used that $\left[F_{1bl}^k\right]_{k}^b = \alpha_b/\alpha_k$ independent of $l$ and we have defined
\begin{equation*}
  \tikzsetnextfilename{6bea4aa4-45f7-476e-9e0f-5dbcc765dadf}
  %
 = \alpha_k\cdot\delta_{ck} \cdot \delta_{a1} .
\end{equation*}

Let us now show that $\mathcal{A}$ is not only a coalgebra, but it also has a pre-bialgebra structure. For that, we define fusion tensors corresponding to this MPO tensor as
\begin{equation*}
  \tikzsetnextfilename{d506d6d8-7a59-4691-9ea3-cb9c239e70c0}
  %
 = \left[\left(F_{fcd}^e\right)^{-1}\right]_g^{l} \cdot \delta_{ff'} \delta_{gg'} \delta_{ee'}.
\end{equation*}
These tensors satisfy the associativity equations \cref{eq:associative}, because the equations
\begin{equation*}
  \tikzsetnextfilename{d8612263-3641-4e5d-bf03-fe6b5b122581}
  %
\end{equation*}
both follow from the pentagon equation \cref{eq:pentagon_simple}. More precisely, the left equation is exactly \cref{eq:pentagon_simple}, while the right equation is its inverse,
\begin{equation*}
  \left[\left(F_{fcd}^e\right)^{-1}\right]_g^l \left[\left(F_{abl}^e\right)^{-1}\right]_f^k = \sum_h \left[\left(F_{abc}^g\right)^{-1}\right]_f^h \left[\left(F_{ahd}^e\right)^{-1}\right]_g^k \left[\left(F_{bcd}^k\right)^{-1}\right]_h^l.
\end{equation*}
We can now check for  the key equation \cref{eq:MPO_tensor_product}. The r.h.s.\ of \cref{eq:MPO_tensor_product} is
\begin{equation*}
  \tikzsetnextfilename{655df682-a394-491d-ac0e-dc0e88781b83}
,
\end{equation*}
where we have only written out the components that are not automatically zero due to the delta functions. This equation is the pentagon equation rearranged:
\begin{equation*}
   \left[F_{abc}^g\right]_h^f \left[F_{ahd}^e\right]_k^g = \sum_l \left[F_{fcd}^e\right]_l^g \left[F_{abl}^e\right]_k^f \left[\left(F_{bcd}^k\right)^{-1}\right]_h^l.
\end{equation*}
Let us now check for the l.h.s.\ of \cref{eq:MPO_tensor_product}, i.e.\ the orthogonality of the fusion tensors:
\begin{equation*}
  \sum_{f'gg'e'}
  \tikzsetnextfilename{56971c56-e5a4-47d9-9b2a-f2b6c03316ca}
  \begin{tikzpicture}[every node/.style={font=\tiny,fill=white,inner sep=1pt}]
    \pic (f) at (0,0) {stringnet fusion};
    \pic[xscale=-1] (i) at (-0.95,0) {stringnet fusion};
    %%%%%%%%%%%%%%%%%%%%%%%%%%%%%%%%%%%%%%%%%%%%%%%%%%%%%%%%%%%%%
    % labels on right
    \node[anchor=north] at (f-right-down) {$e$};
    \node               at (f-right-mid)  {$l$};
    \node[anchor=south] at (f-right-up)   {$f$};
    % labels on left up
    \node[anchor=south] at (f-left-up-up)    {$\tinyprime{f}$};
    \node               at (f-left-up-mid)   {$c$};
    \node[anchor=north] at (f-left-up-down)  {$g$};
    % labels on left down
    \node[anchor=south] at (f-left-down-up)    {$\tinyprime{g}$};
    \node               at (f-left-down-mid)   {$d$};
    \node[anchor=north] at (f-left-down-down)  {$\tinyprime{e}$};
    %%%%%%%%%%%%%%%%%%%%%%%%%%%%%%%%%%%%%%%%%%%%%%%%%%%%%%%%%%%%%%%
    % labels on right
    \node[anchor=north] at (i-right-down) {$\tinypprime{e}$};
    \node               at (i-right-mid)  {$\tinypprime{l}$};
    \node[anchor=south] at (i-right-up)   {$\tinypprime{f}$};
  \end{tikzpicture} =
  \sum_{f'gg'e'cd}
  \left[\left(F_{fcd}^e\right)^{-1}\right]_g^{l'} \cdot \delta_{ff'} \delta_{gg'} \delta_{ee'} \left[F_{fcd}^e\right]_l^g \cdot \delta_{f'f''} \delta_{ee''} = \delta_{ee''} \delta_{ff''} \delta_{ll''} N_{fl}^e,
\end{equation*}
where again $N_{fl}^e$ is one or zero, which gives the r.h.s.\ of \cref{eq:MPO_tensor_product}. We have thus checked that $\mathcal{A}$ admits a pre-bialgebra structure with the usual multiplication, except that we have not shown that $\mathcal{A}$ has a unit. For that, note that the unit is of the form
\begin{equation*}
     1 =
     \tikzsetnextfilename{05869e72-b2d0-40ea-ab60-933ac9c7bb3e}
    \begin{tikzpicture}[every node/.style={font=\tiny,fill=white,inner sep=1pt,outer sep=1pt}]
      \pic (s) at (0,0) {stringnet mpo};
      \pic[rotate=180] (unitup) at (0.95,0) {stringnet unit};
      \pic (unitdown) at (-0.95,0) {stringnet unit};
      % labels on top
      \node[anchor=west]  at (s-up-right)   {$b$};
      \node               at (s-up-mid)     {$a$};
      \node[anchor=east]  at (s-up-left)    {$f$};
      % labels on left
      \node[anchor=south] at (s-left-up)    {$f$};
      \node               at (s-left-mid)   {$l$};
      \node[anchor=north] at (s-left-down)  {$c$};
      % labels below
      \node[anchor=east]  at (s-down-left)  {$c$};
      \node               at (s-down-mid)   {$a$};
      \node[anchor=west]  at (s-down-right) {$k$};
      % labels on the right
      \node[anchor=north] at (s-right-down) {$k$};
      \node               at (s-right-mid)  {$l$};
      \node[anchor=south] at (s-right-up)   {$b$};
    \end{tikzpicture} \ ,
\end{equation*}
where we have defined
\begin{equation*}
  \tikzsetnextfilename{086bcb90-967c-4f4a-a85e-7390943dd797}
  \begin{tikzpicture}[every node/.style={font=\tiny,fill=white,inner sep=1pt}]
    \pic (v) at (0,0) {stringnet unit};
    %%%%%%%%%%%%%%%%%%%%%%%%%%%%%%%%%%%%%%%%%%%%%%%%%%%%%%%%%%%%%
    % labels
    \node[anchor=north] at (v-down) {$c$};
    \node               at (v-mid)  {$b$};
    \node[anchor=south] at (v-up)   {$a$};
  \end{tikzpicture} = \alpha_a^{-1} \cdot \delta_{ac} \delta_{b1}
  \qquad \text{and} \qquad
  \tikzsetnextfilename{dc2e9c48-f66d-45f6-86e7-91b680499916}
  \begin{tikzpicture}[every node/.style={font=\tiny,fill=white,inner sep=1pt}]
    \pic[xscale=-1] (v) at (0,0) {stringnet unit};
    %%%%%%%%%%%%%%%%%%%%%%%%%%%%%%%%%%%%%%%%%%%%%%%%%%%%%%%%%%%%%
    % labels
    \node[anchor=north] at (v-down) {$c$};
    \node               at (v-mid)  {$b$};
    \node[anchor=south] at (v-up)   {$a$};
  \end{tikzpicture} = \alpha_a\cdot\delta_{ac} \delta_{b1} .
\end{equation*}
As the unit is described by a rank-one boundary that is supported only in the vacuum sector, the pre-bialgebra $\mathcal{A}$ automatically satisfies the unit axiom. Dually, as the counit is also described by a rank-one boundary that is supported only in the vacuum sector, $\mathcal{A}$ also satisfies the counit axiom. Therefore $\mathcal{A}$ is a WBA.

Let us now show that $\mathcal{A}$ is a WHA as well. A cosemisimple WBA is a WHA if and only if for every irrep sector $a$ there is another irrep sector $\bar{a}$ such that the symmetries $N_{ab}^c = N_{\bar{a}c}^b$ and $N_{ab}^c = N_{c\bar{b}}^a$ hold. In our case, as the fusion multiplicities originate from a fusion category, and thus these symmetries trivially hold. Therefore $\mathcal{A}$ is a WHA. The matrices $Z_c$ and $Z^{-1}_c$ describing the antipode can be expressed with the help of the fusion tensors as \cref{eq:w_a}:
\begin{equation*}
  Z_c =
  \tikzsetnextfilename{a801d23f-7184-4b91-918f-4ffce6e6e92a}
 .
\end{equation*}
It is easy to see that these matrices are the inverses of each other.
The square of the antipode is realized by the matrices
\begin{equation}\label{eq:stringnet_g}
  \tikzsetnextfilename{9bbbf3c1-e886-4229-9c85-fd4ad17dbfd6}
  %
  = \frac{1}{\left[F_{c\bar{c}c}^{c}\right]_{1}^{1}}  \left[F_{a\bar{c}c}^a\right]_1^f
  \left[\left(F_{a\bar{c}c}^a\right)^{-1}\right]^1_f \ .
\end{equation}

\subsection{Fibonacci anyons}

The fusion rules are given by:
\begin{align*}
  N_{11}^1 = N_{\tau 1}^\tau = N_{1\tau}^\tau = N_{\tau\tau}^1 = N_{\tau \tau}^1 =& 1, \\
  N_{11}^\tau = N_{\tau 1}^1 = N_{1\tau }^1 = & 0.
\end{align*}
The F-symbol $\left[F_{fcd}^e\right]_l^g$ is proportional to $N_{fc}^g \cdot N_{cd}^l \cdot N_{fg}^e \cdot N_{ld}^e$. Therefore the following entries are the only non-zero ones and are given by
\begin{align*}
  \left[F_{111}^1\right]_1^1 =
  \left[F_{11\tau}^\tau\right]_\tau^1 =
  \left[F_{\tau 11}^\tau \right]_1^\tau=
  \left[F_{1\tau 1}^\tau \right]_\tau^\tau =
  \ \left[F_{1\tau\tau}^1\right]_1^\tau  =
  \left[F_{1\tau\tau}^\tau\right]_\tau^\tau =  \\
  \left[F_{\tau 1\tau}^\tau\right]_\tau^\tau =
  \left[F_{\tau 1\tau}^1\right]_\tau^\tau  =
  \left[F_{\tau\tau 1}^\tau\right]_\tau^\tau =
  \left[F_{\tau\tau 1}^1\right]_\tau^1 =
  \left[F_{\tau\tau\tau}^1\right]_\tau^\tau =1 \\
  \left[F_{\tau\tau\tau}^\tau \right]_1^1 = \varphi, \ .
  \left[F_{\tau\tau\tau}^\tau \right]_\tau^\tau = -\varphi, \
  \left[F_{\tau\tau\tau}^\tau \right]_1^\tau =
  \left[F_{\tau\tau\tau}^\tau \right]_\tau^1 = \sqrt{\varphi},
\end{align*}
where $\varphi = (\sqrt{5}-1)/2$. We also have $\left(\left(F_{abc}^d\right)^{-1}\right)_e^f=\left(F_{abc}^d\right)_f^e$.

\subsubsection{The MPO tensor}

The MPO tensor given in \cref{eq:stringnet_mpo_tensor} has bond- and physical dimension $8$. It is block-diagonal both in the horizontal and vertical direction, with the two blocks labeled by $1$ and $\tau$. The block $1$ corresponds to the subspace spanned by  $\ket{111}, \ket{\tau 1 \tau}$ and the block $\tau$ to the subspace spanned by  $\ket{1\tau \tau}, \ket{\tau \tau 1}, \ket{\tau\tau\tau}$. The two simple cocommutative elements are given by
\begin{align*}
  a_1 & \equiv
  \left(
\right) ,
\end{align*}
where the basis is ordered as given above. The square of the antipode, as given by \cref{eq:stringnet_g}, is implemented by
\begin{equation*}
  \tikzsetnextfilename{11677d60-2b07-41c2-80b8-e22ffe04ff50}
  %
\right) .
\end{equation*}

\end{document}